\documentclass[english,aip,reprint,nofootinbib]{revtex4-1}
\usepackage[T1]{fontenc}
\usepackage[latin9]{inputenc}
\usepackage{amsmath}
\usepackage{amsthm}
\usepackage{amssymb}

\makeatletter

\providecommand{\tabularnewline}{\\}

\theoremstyle{plain}
\newtheorem{thm}{\protect\theoremname}
  \theoremstyle{plain}
  \newtheorem{prop}[thm]{\protect\propositionname}

\usepackage{url}
\usepackage{hyperref}
\usepackage{babel}
\usepackage{nccmath}
\usepackage{chngcntr}
\newcommand{\rep}[1]{\mathbf{#1}} 
\newcommand{\ket}[1]{\left|{#1}\right\rangle }
\newcommand{\group}{\mathcal{PSL}_{2}\left(7\right)} 
\newcommand{\subga}{S_4}
\newcommand{\subgb}{\mathcal{T}_7}
\newcommand{\olcite}[1]{Ref. [\onlinecite{#1}]}
\newcommand{\olcitem}[1]{Refs. [\onlinecite{#1}]}

\newcommand{\mr}[1]{\rho_{\rep{#1}}}
\newcommand{\cpx}[1]{\Gamma_{#1}}

\newcommand{\tsprod}[2]{\rep{#1}\otimes\rep{#2}}
\newcommand{\covp}[3]{\left[\tsprod{#1}{#2}\right]^{}_{\rep{#3}}}
\newcommand{\cova}[3]{\left(\tsprod{#1}{#2}\right)^{}_{\rep{#3}}}
\newcommand{\covs}[3]{\big\{\tsprod{#1}{#2}\big\}^{}_{\rep{#3}}}
\newcommand{\covt}[3]{\big\{\tsprod{#1}{#2}\big\}^{}_{\rep{#3}}}
\newcommand{\covsub}[3]{\big\{{#1}\otimes{#2}\big\}^{}_{{#3}}}
\newcommand{\covlg}[3]{\left[{#1}\otimes{#2}\right]^{}_{{#3}}}

\newcommand{\cgEqFontsize}{\large}

\newcommand{\tsprodx}[2]{\rep{#1}\otimes\rep{#2}}
\newcommand{\subcgs}[3]{\big\{ \tsprod{#1}{#2}\big\}^{}_{\rep{#3}}}
\newcommand{\subcgt}[3]{\big\{ \tsprod{#1}{#2}\big\}^{}_{\rep{#3}}}
\newcommand{\bfl}{\begin{fleqn}[25pt]}
\newcommand{\efl}{\end{fleqn}}

\makeatother

\usepackage{babel}
  \providecommand{\propositionname}{Proposition}
\providecommand{\theoremname}{Theorem}

\begin{document}

\title{Clebsch-Gordan coefficients of discrete groups in subgroup bases}

\author{Gaoli Chen}

\affiliation{Institute for Fundamental Theory, Department of Physics, University
of Florida, Gainesville, Florida 32611, USA}
\email{gchen@ufl.edu}

\begin{abstract}
We express each Clebsch-Gordan (CG) coefficient of a discrete group
as a product of a CG coefficient of its subgroup and a factor, which
we call an embedding factor. With an appropriate definition, such
factors are fixed up to phase ambiguities. Particularly, they are
invariant under basis transformations of irreducible representations
of both the group and its subgroup. We then impose on the embedding
factors constraints, which relate them to their counterparts under
complex conjugate and therefore restrict the phases of embedding factors.
In some cases, the phase ambiguities are reduced to sign ambiguities.
We describe the procedure of obtaining embedding factors and then
calculate CG coefficients of the group $\mathcal{PSL}_{2}\left(7\right)$
in terms of embedding factors of its subgroups $S_{4}$ and $\mathcal{T}_{7}$.
\end{abstract}
\maketitle

\section{Introduction}

Discrete subgroups of $SU(3)$ are widely used in flavor model building
in particle physics, where one needs to study the mathematical properties
of the selected group, e.g, its Clebsch-Gordan (CG) coefficients,
which determine how fields are coupled in the model. 

For some continuous groups, CG coefficients are usually expressed
in terms of those of their subgroups. For example, $SU(3)$ CG coefficients
can be factored into $SU(2)$ CG coefficients and so-called isoscalar
factors\cite{deSwart:1963pdg}. With the notations of \olcite{deSwart:1963pdg},
the $SU(3)$ CG coefficients for the tensor product $\rep{\mu_{1}}\otimes\rep{\mu_{2}}=\rep{\mu_{3}}\oplus\cdots$,
with $\mu_{i}$ being irreducible representations (irreps), can be
written as 
\begin{equation}
\left(\begin{array}{ccc}
\mu_{1} & \mu_{2} & \mu_{3}\\
\nu_{1} & \nu_{2} & \nu_{3}
\end{array}\right)=C_{I_{1z}I_{2z}I_{3z}}^{I_{1}I_{2}I_{3}}\left(\begin{array}{cc}
\mu_{1} & \mu_{2}\\
I_{1}Y_{1} & I_{2}Y_{2}
\end{array}\Bigg|\begin{array}{c}
\mu_{3}\\
I_{3}Y_{3}
\end{array}\right),\label{eq:Intro_mumumu}
\end{equation}
 Such expressions are convenient because the number of isoscalar factors
is smaller than the number of $SU(3)$ CG coefficients. 

In this paper, we discuss the relation between CG coefficients of
a discrete group and those of its subgroup. To our knowledge, a general
study of this idea for discrete groups has not been done in the literature. 

For a discrete group $G$ with subgroup $H$, we want to find a set
of factors that relate CG coefficients of $G$ and $H$, denoted as
$M^{\left(G\right)}$ and $M^{\left(H\right)}$,
\begin{equation}
M_{AB}^{\left(G\right)}=\mathcal{E}_{A}M_{B}^{\left(H\right)},\label{eq:Intro_MMM}
\end{equation}
 where $B$ is a collection of indices to specify the subgroup CG
coefficients and $A$ a collection of indices carrying the information
of how irreps of $H$ are embedded in irreps of $G$. In analogy to
the isoscalar factors, we call the coefficients in $\mathcal{E}$
the embedding factors. Since CG coefficients are basis-dependent,
Eq. (\ref{eq:Intro_MMM}) is also basis-dependent. Then one may ask
the following questions: 1) how can we define $\mathcal{E}$ independent
of bases? 2) are the coefficients in $\mathcal{E}$ unique for all
CG coefficients $M^{\left(H\right)}$ in all bases? To answer these
questions, we show that embedding factors $\mathcal{E}$ can be defined
in a way that is invariant under basis transformations of irreps of
$G$ and $H$. It implies that there does exist a set of embedding
factors for all bases of irreps of $G$ and $H$. This does not exhaust
all the ambiguities of the embedding factors because the coefficients
of $\mathcal{E}$ still have phase ambiguities, which stem from those
of the subgroup CG coefficients.

We analyze the phase ambiguities of the embedding factors and propose
a way to reduce them, with some ambiguities remained to be eliminated
by other conventions. The advantage of our convention is that it only
depends on general properties of groups and irreps, and hence it can
apply to any groups. We impose on embedding factors constraints which
require that a contraction of two irreps, in the form $\left(\rep{X}\otimes\rep{Y}\right)_{\rep{Z}}$,
behaves the same as the corresponding irrep $\rep{Z}$ under the action
of complex conjugate. Such constraints lead to the following consequences:
\begin{itemize}
\item Case I: If all of $\rep{X}$, $\rep{Y}$, and $\rep{Z}$ are real
or pseudoreal irreps, the overall phase of embedding factors is fixed
up to a sign factor. In particular, if the irreps can be decomposed
into real or pseudoreal subgroup irreps, then the corresponding embedding
factors are real numbers fixed up to sign factors.
\item Case II: If $\rep{Z}$ is real or pseudoreal and $\rep{X}$ is complex
conjugate of $\rep{Y}$, the overall phase of embedding factors is
also fixed up to a sign factor. 
\item Case III: In other cases, the embedding factors for $\rep{X}\otimes\rep{Y}\to\rep{Z}$
and $\rep{\bar{X}}\otimes\rep{\bar{Y}}\to\rep{\bar{Z}}$ are complex
conjugate to each other. Here, $\rep{\bar{X}}$ represents conjugate
of $\rep{X}$ if it is complex, or $\rep{X}$ itself otherwise. This
statement looks trivial and one may think that CG coefficients always
have such a property. But actually such relation not always holds
for CG coefficients (See Proposition \ref{prop:CG_phase} of Section
\ref{sec:RedPhaseAmbiguity}). This is another advantage of embedding
factors compared with CG coefficients. 
\end{itemize}
We introduce a procedure to calculate the embedding factors. The calculation
involves complicated cyclotomic numbers, which are polynomials of
roots of unity $e^{2\pi i/n}$, and hence it is difficult to obtain
simplified results. We therefore implement an algorithm to perform
arithmetic calculations of cyclotomic numbers.

We apply our technique to the group $\group$ and its subgroups, the
groups $\subga$ and $\subgb$, and obtain representation matrices
of $\group$ in its subgroup bases, in which subgroup elements are
block diagonal matrices. For both subgroups, we acquire complete lists
of embedding factors of the group $\group$. We also find the embedding
factors for $\subga$ and its subgroup $A_{4}$. We automate much
of the procedure in the Mathematica code in \olcite{Chen:2017GitHub},
which can be easily adjusted for calculating embedding factors of
other discrete groups. 

The complete list of $\group$ CG coefficients and the presentation
matrices of $\group$ in $\subga$ basis are new results, which could
be helpful for flavor model building based on $\group$, or studying
connections among $\group$ models\cite{Chen:2014wiw,PSL27Aliferis:2016blw},
$\subgb$ models \cite{Kile:2013gla,T7Cao:2010mp,T7Cao:2011cp,T7Hagedorn:2008bc,T7Luhn:2007sy,T7Luhn:2012bc},
and $\subga$ models (for a review of $\subga$ models, see \olcite{S4Bazzocchi:2012st}).
A subgroup tree of $SU\left(3\right)$ discrete subgroups can be found
in \olcite{Merle:2011vy}. For systematic analysis of discrete groups
used in flavor model buildings, see \olcitem{Ishimori:2010au,Ludl:2009ft}.

The remainder of the paper is organized as follows. In Section \ref{sec:Amb_CGC},
we analyze the ambiguities of CG coefficients. In Section \ref{sec:CGCs-in-subgroup},
we define the embedding factors and then show that they are basis
independent but still have phase ambiguities. In Section \ref{sec:RedPhaseAmbiguity},
phase conventions are introduce to reduce the phases ambiguities of
embedding factors. In Section \ref{sec:Calc_CGC}, we describe the
procedure of calculating the embedding factors. The procedure is then
applied to the group $\group$ and its subgroups $\subga$ and $\subgb$.
Specifically, in Section \ref{sec:Gen_PSL27}, the representation
matrices of the $\group$ group are obtained in its subgroup bases.
In Section \ref{sec:Calc_PSL27_CGC}, we calculate embedding factors
of the tensor product $\rep{6}\otimes\rep{6}\to\rep{6}$ of the group
$\group$ as an example. The group theory properties of relevant groups
are given in Appendix \ref{sec:Group-Theory}. In Appendix \ref{sec:Tech_Detail},
we describe the algorithm for arithmetic calculations of cyclotomic
numbers. Appendix \ref{sec:App_CGC_PSL27_S4} and \ref{sec:App_CGC_PSL27_T7}
list complete sets of CG coefficients of $\group$ in $\subga$ and
$\subgb$ bases.

\subsection*{Conventions}

Some of the conventions used in the main text are as follows. $G$
represents a discrete group and $H$ is its subgroup. The boldface
and capitalized letters $\rep{X}$, $\rep{Y}$, and $\rep{Z}$ are
irreps of the group $G$ and $X$, $Y$, $Z$ are corresponding vectors
in Hilbert spaces of these irreps. The boldface and lowercase letters
$\rep{x}$, $\rep{y}$, $\rep{z}$ are irreps of the subgroup $H$,
and $x$, $y$, $z$ are corresponding vectors in Hilbert space of
these irreps. Contraction of irreps of $G$ are denoted as $\covlg{X}{Y}{Z}$,
which means a contraction of $X$ and $Y$ to $Z$. Similarly, $\covsub{x}{y}{z}$
represents a contraction of $H$ irreps. The letters $i,\,j,\,k$
label a single component of a vector while $a,\,b,\,c$ label an irrep
of the subgroup. For example, $x_{i}$ means the $i$-th component
of $x$ while $x_{a}$ is a vector of the irreps $\rep{x_{a}}$. A
matrix realization (representation matrix) of group element $g$ in
$\rep{X}$ irrep is denoted as $\mr{X}\left(g\right)$. Representation
matrices are always unitary in this paper.

\section{\label{sec:Amb_CGC}Clebsch-Gordan coefficients and their ambiguities}

The tensor product of two irreps of a group is in general reducible.
Let $G$ be a discrete group with irreps $\rep{X}$, $\rep{Y}$, $\rep{Z}$,
and $X$, $Y$, $Z$ be vectors of the corresponding Hilbert spaces
on which the group elements act. If $\rep{Z}$ is contained in the
tensor product $\rep{X}\otimes\rep{Y}$, then given $X$ and $Y$,
there exist a $Z$ and a set of coefficients $M_{k\alpha}^{\left(XY\to Z\right)}$
such that, for all $g\in G$, 
\begin{equation}
\left(\mr{Z}\left(g\right)Z\right)_{k}=\sum_{\alpha}M_{k\alpha}^{\left(XY\to Z\right)}\left(\mr{X}\left(g\right)X\otimes\mr{Y}\left(g\right)Y\right)_{\alpha},\label{eq:Amb_CG_Def}
\end{equation}
 where $\mr{Z}$ is a matrix realization of irrep $\rep{Z}$ and elements
of the tensor product $\left(\mr{X}\left(g\right)X\otimes\mr{Y}\left(g\right)Y\right)$,
a column vector with dimension $\left(\dim X\times\dim Y\right)$,
are of the form $\left(\mr{X}\left(g\right)X\right)_{i}\left(\mr{Y}\left(g\right)Y\right)_{j}$.
The coefficients $M_{k\alpha}^{\left(XY\to Z\right)}$ are the Clebsch-Gordan
coefficients. Note that CG coefficients can also be defined as a unitary
transformation between the tensor products of group elements and the
direct sum of their irreps, e.g, \olcite{van1978clebsch}. In this
paper, it is more convenient to define CG coefficients with vector
spaces. In the remainder of this paper, we may suppress the superscript
$\left(XY\to Z\right)$ if there is no chance of confusion.

If $Z$ and $M^{\left(XY\to Z\right)}$ satisfy eq. (\ref{eq:Amb_CG_Def}),
then for any nonzero c-number $\lambda$, $\lambda Z$ and $\lambda M^{\left(XY\to Z\right)}$
also satisfy the equation. So it is then conventional to impose on
CG coefficients, in addition to orthogonality, normalization constraints,
which give rise to unitary CG coefficients
\begin{equation}
\sum_{\alpha}M_{k\alpha}^{*}M_{l\alpha}=\delta_{kl}.\label{eq:Amb_CG_Orthnorm}
\end{equation}

There are, however, other ambiguities. Firstly, CG coefficients are
basis-dependent. Under basis transformations \begin{subequations}
\begin{equation}
X\to X^{\prime}=U_{X}X,\quad\mr{X}\left(g\right)\to U_{X}\mr{X}\left(g\right)U_{X}^{-1},
\end{equation}
 and similarly for $Y$ and $Z$, eqs. (\ref{eq:Amb_CG_Def}) and
(\ref{eq:Amb_CG_Orthnorm}) are invariant if the matrices $M$ simultaneously
transform as 
\begin{equation}
M\to M^{\prime}=U_{Z}M\left(U_{X}^{-1}\otimes U_{Y}^{-1}\right).
\end{equation}
\end{subequations}We shall call this basis ambiguity in the remainder
of the paper. The second ambiguity is phase ambiguity, meaning that
eqs. (\ref{eq:Amb_CG_Def}) and (\ref{eq:Amb_CG_Orthnorm}) are invariant
under the transformation $M\to e^{i\theta}M$, $Z\to e^{i\theta}Z$.
The last ambiguity exists when there are nontrivial multiplicities
in a tensor product, i.e., $\rep{X}\otimes\rep{Y}\to\rep{Z}^{\left(1\right)}\oplus\cdots\oplus\rep{Z}^{\left(\mu\right)}$
, then a linear combination of the CG coefficients $\lambda_{1}M^{\left(1\right)}+\cdots+\lambda_{n}M^{\left(\mu\right)}$
with $\sum_{i}\lambda_{i}^{*}\lambda_{i}=1$ is also a set of unitary
CG coefficients. 

\section{\label{sec:CGCs-in-subgroup}Embedding factors}

Let $H$ be a subgroup of $G$ and assume that the irreps $\rep{X}$,
$\rep{Y}$, $\rep{Z}$ can be decomposed into irreps of $H$ as 
\begin{equation}
\rep{X}=\bigoplus_{a}\rep{x_{a}},\quad\rep{Y}=\bigoplus_{b}\rep{y_{b}},\quad\rep{Z}=\bigoplus_{c}\rep{z_{c}},\label{eq:SBCG_Embed}
\end{equation}
where $\rep{x_{a}}$, $\rep{y_{b}}$, $\rep{z_{c}}$ are irreps of
$H$. For the contraction $\rep{X}\otimes\rep{Y}\to\rep{Z}$, we can
write 
\begin{equation}
P_{Z\to c}Z=\sum_{a,b}\mathcal{E}_{c,ab}^{\left(XY\to Z\right)}\left\{ \left(P_{X\to a}X\right)\otimes\left(P_{Y\to b}Y\right)\right\} _{z_{c}},\label{eq:SBCG_Def1}
\end{equation}
where $P_{X\to a}$ is a matrix of dimension $\dim\rep{x_{a}}\times\dim\rep{X}$
to project the $x_{a}$ components from $X$ and $\mathcal{E}_{c,ab}^{\left(XY\to Z\right)}$
are the embedding factors. The projection matrices act like similarity
transformations on $H$ elements between representation $\rep{X}$
and $\rep{x_{a}}$, i.e., 
\begin{equation}
\left(P_{X\to a}\right)_{ik}\left[{\rm \mr{X}}\left(h\right)\right]_{kl}\left(P_{X\to b}^{\dagger}\right)_{lj}=\delta_{ab}\left[\mr{x_{a}}\left(h\right)\right]_{ij},\quad\forall h\in H.\label{eq:SBCG_Projection}
\end{equation}
In analogy to eq. (\ref{eq:Amb_CG_Def}), the coefficients $\mathcal{E}_{c,ab}^{\left(XY\to Z\right)}$
should actually satisfy a stronger constraint \begin{widetext} 
\begin{equation}
P_{Z\to c}\mr{Z}\left(g\right)Z=\sum_{a,b}\mathcal{E}_{c,ab}^{\left(XY\to Z\right)}\left\{ \left(P_{X\to a}\mr{X}\left(g\right)X\right)\otimes\left(P_{Y\to b}\mr{Y}\left(g\right)Y\right)\right\} _{z_{c}},\quad\forall g\in G,\label{eq:SBCG_Def2}
\end{equation}
\end{widetext} which reduces to eq. (\ref{eq:SBCG_Def1}) when $g=e$. 

In the following, we may write $\mathcal{E}_{c,ab}^{\left(XY\to Z\right)}$
simply as $\mathcal{E}_{c,ab}$ for convenience. The orthonormalization
constraints (\ref{eq:Amb_CG_Orthnorm}) now become
\begin{equation}
\sum_{ab}\mathcal{E}_{c,ab}^{*}\mathcal{E}_{d,ab}=\delta_{cd}.\label{eq:SBCG_orthnorm}
\end{equation}
 We remark that the $\mathrm{rhs}$ of each equations of (\ref{eq:SBCG_Embed})
may contain duplicated irreps. Such a case can be avoided by choosing
a large enough subgroup $H$. Therefore, for simplicify, we only consider
the case that no irrep is contained twice in an irrep of the large
group. 
\begin{prop}
The coefficients $\mathcal{E}_{c,ab}^{\left(XY\to Z\right)}$ defined
as eq. (\ref{eq:SBCG_Def2}) are invariant under basis transformations
of irreps of both the group and its subgroup.
\end{prop}

\begin{proof}
Under a basis transformation of the group $G$, the vector $X$, projection
matrix $P_{X\to a}$, and matrices $\rho_{X}\left(g\right)$ transform
as $X^{\prime}=U_{X}X$, $P_{X\to a}^{\prime}=P_{X\to a}U_{X}^{-1}$,
and $\rho_{X}^{\prime}\left(g\right)=U_{X}\mr{X}\left(g\right)U_{X}^{-1}$
respectively. It follows that $\left(P_{X\to a}\mr{X}\left(g\right)X\right)$
is invariant, so are $\left(P_{Y\to b}\mr{Y}\left(g\right)Y\right)$
and $\left(P_{Z\to c}\mr{Z}\left(g\right)Z\right)$ . Hence, both
sides of eq. (\ref{eq:SBCG_Def2}) are invariant under the transformation.

Now consider basis transformations of the subgroup irreps. For simplicity,
we can now fix projection matrices to special forms, since we can
always perform basis transformations to bring projection matrices
to desired forms without changing $\mathcal{M}_{c,ab}$ coefficients.
So we choose bases such that elements of $H$ are block diagonal matrices,
which means that the projection matrices have the form 
\begin{equation}
P_{X\to a}=\begin{pmatrix}O_{1} & \cdots & O_{a-1} & P_{X\to a}^{\left(u\right)} & O_{a+1}\cdots\end{pmatrix},\label{eq:SBCG_PXa}
\end{equation}
 where $P_{X\to a}^{\left(\mathrm{u}\right)}$ is a unitary matrix
of dimension $\dim\rep{x_{a}}\times\dim\rep{x_{a}}$ and $O_{b}$
are zero matrices of dimension $\dim\rep{x_{a}}\times\dim\rep{x_{b}}$.
For basis transformations of $H$\begin{subequations}\label{eq:SBCG_Tran_sub}
\begin{equation}
\left(x_{a},\,y_{b},\,z_{c}\right)\to\left(x_{a}^{\prime},\,y_{b}^{\prime},\,z_{c}^{\prime}\right)=\left(U_{a}x_{a},\,U_{b}y_{b},\,U_{c}z_{c}\right),
\end{equation}
the subgroup contraction $\left\{ x_{a}\otimes y_{b}\right\} _{z_{c}}$
should transform as 
\begin{equation}
z_{c}=\left\{ x_{a}\otimes y_{b}\right\} _{z_{c}}\to z_{c}^{\prime}=\left\{ U_{a}x_{a}\otimes U_{b}y_{b}\right\} _{z_{c}^{\prime}}^{\prime},
\end{equation}
 \end{subequations}where we have used the primed brackets to indicate
the new CG coefficients in the new basis. It then implies 
\begin{equation}
\left\{ U_{a}x_{a}\otimes U_{b}y_{b}\right\} _{z_{c}^{\prime}}^{\prime}=U_{c}\left\{ x_{a}\otimes y_{b}\right\} _{z_{c}}.\label{eq:SBCG_Tran_SubCG}
\end{equation}
 Under these transformations, $\mr{X}\left(g\right)X$ transforms
as 
\begin{equation}
\mr{X}\left(g\right)X\to U_{X}\mr{X}\left(g\right)X,\label{eq:SBCG_Tran_rho_X}
\end{equation}
 where 
\begin{equation}
U_{X}=\bigoplus_{\rep{x_{a}}\text{ in }\rep{X}}\left(P_{X\to a}^{\left(\mathrm{u}\right)}\right)^{-1}U_{a}P_{X\to a}^{\left(u\right)}.
\end{equation}
Since $P_{X\to a}$ is in the form of (\ref{eq:SBCG_PXa}), it is
easy to see that 
\begin{equation}
P_{X\to a}U_{X}=U_{a}P_{X\to a}.\label{eq:SBCG_Pxa_UX}
\end{equation}
(\ref{eq:SBCG_Tran_rho_X}) and (\ref{eq:SBCG_Pxa_UX}) imply that
$P_{X\to a}\mr{X}\left(g\right)X$ transforms as 
\[
P_{X\to a}\mr{X}\left(g\right)X\to U_{a}P_{X\to a}\mr{X}\left(g\right)X.
\]
 Now consider the both sides of eq. (\ref{eq:SBCG_Def2}) under the
basis transformations. The $\mathrm{lhs}$ becomes $U_{c}P_{Z\to c}\mr{Z}\left(g\right)Z$
and the $\mathrm{rhs}$ becomes 
\begin{align*}
\mathrm{rhs} & =\sum_{a,b}\mathcal{E}_{c,ab}\left\{ U_{a}P_{X\to a}\mr{X}\left(g\right)X\otimes U_{b}P_{Y\to b}\mr{Y}\left(g\right)Y\right\} _{z_{c}^{\prime}}^{\prime}\\
 & =U_{c}\sum_{a,b}\mathcal{E}_{c,ab}\left\{ P_{X\to a}\mr{X}\left(g\right)X\otimes P_{Y\to b}\mr{Y}\left(g\right)Y\right\} _{z_{c}}\\
 & =U_{c}P_{Z\to c}\mr{Z}\left(g\right)Z=\mathrm{lhs},
\end{align*}
 where we have used eq. (\ref{eq:SBCG_Tran_SubCG}) in the second
equality. Eq. (\ref{eq:SBCG_Def2}) is therefore invariant under the
basis transformations (\ref{eq:SBCG_Tran_sub}). 
\end{proof}
We have showed that embedding factors are independent of bases of
both the group and its subgroup. There are, however, still ambiguities
in embedding factors due to phase ambiguities of the projection matrices
and subgroup CG coefficients. Consider the $U\left(1\right)$ transformations
on projection matrices and subgroup CG coefficients \footnote{According to (\ref{eq:SBCG_Projection}), projection matrices can
also transformation as 
\[
P_{X\to a}\to P_{X\to a}\mr{X}\left(g^{\prime}\right),
\]
 where $g^{\prime}$ is an element of center of $G$. But we can see
that it is equivalent to replacing $g$ with $g^{\prime}g$ in eq.
(\ref{eq:SBCG_Def2}) and hence does not change embedding factors.}
\begin{align}
P_{X\to a} & \to e^{i\theta_{a}^{\left(X\right)}}P_{X\to a},\label{eq:SBCG_Tran_PXa}\\
\covsub{x_{a}}{y_{b}}{z_{c}} & \to e^{i\phi^{\left(ab\to c\right)}}\covsub{x_{a}}{y_{b}}{z_{c}}.\label{eq:SBCG_Tran_Subcg}
\end{align}
Under these $U\left(1\right)$ transformations, the embedding factors
transform as 
\begin{equation}
\mathcal{E}_{c,ab}^{\left(XY\to Z\right)}\to e^{i\left(\theta_{a}^{\left(X\right)}+\theta_{b}^{\left(Y\right)}-\theta_{c}^{\left(Z\right)}+\phi^{\left(ab\to c\right)}\right)}\mathcal{E}_{c,ab}^{\left(XY\to Z\right)}.\label{eq:SBCG_Tran_M}
\end{equation}
 We see that there are in general four phase ambiguities for each
embedding factor. They can be removed or reduced by appropriate phase
conventions, which are usually basis-dependent. For example, for $SU\left(2\right)$
CG coefficients, it is conventional to choose a particular CG coefficient
to be real and positive. In the following section, we introduce a
basis-independent convention, which can reduce the number of $U\left(1\right)$
ambiguities and, in some cases, reduce the ambiguities to $Z_{2}$
ambiguities, i.e., ambiguities of sign factors.

\section{\label{sec:RedPhaseAmbiguity}Reducing phase ambiguities }

To introduce the convention, we first discuss real and complex representations.
A real or pseudoreal representation is a representation whose complex
conjugate is equivalent to itself while a complex representation is
a representation that is inequivalent to its complex conjugate. So
we can define the complex conjugate of a real (or pseudoreal) representation
to be itself. Now if $\rho_{\rep{X}}$ is a matrix realization of
$\rep{X}$ and $\rho_{\rep{\bar{X}}}$ the one of $\rep{\bar{X}}$,
then there exists a unitary matrix $\cpx{X}$ such that
\begin{equation}
\cpx{X}\rho_{\rep{X}}\left(g\right)^{*}\cpx{X}^{\dagger}=\rho_{\rep{\bar{X}}}\left(g\right),\quad\forall g\in G.\label{eq:Red_Gamma_rho}
\end{equation}
When $\rep{X}$ is complex, we can always choose $\rho_{\rep{\bar{X}}}\left(g\right)=\rho_{\rep{X}}\left(g\right)^{*}$
so that $\cpx{X}$ can be the identity matrix. When $\rep{X}$ is
real or pseudoreal, then $\rho_{\rep{\bar{X}}}$ is identical to $\rho_{\rep{X}}$
and $\cpx{X}$ is in general a nontrivial unitary matrix depending
on the basis of the representation. For real $\rep{X}$, $\cpx{X}$
is symmetric; for pseudoreal $\rep{X}$, $\cpx{X}$ is antisymmetric\cite{Ramond:2010zz}.
The unitarity of $\cpx{X}$ then implies that $\cpx{X}\cpx{X}^{*}=\pm\mathbf{1}$,
where $+$ is for real $\rep{X}$ and $-$ for pseudoreal $\rep{X}$. 

Eq. (\ref{eq:Red_Gamma_rho}) implies that $\cpx{X}X^{*}$ should
transform as a vector in the representation space of $\rep{\bar{X}}$.
We therefore can define a vector $\bar{X}$ to be
\begin{equation}
\bar{X}\equiv\cpx{X}X^{*}.\label{eq:Red_Gamma_X}
\end{equation}
With such a definition, it is natural to impose the following constraints\begin{subequations}\label{eq:CG_GammaCon}
\begin{align}
\cpx{Z}\left(\covlg{X}{Y}{Z}\right)^{*} & =\covlg{\bar{X}}{\bar{Y}}{\bar{Z}},\label{eq:Red_Gamma_Con1}\\
\cpx{c}\left(\covsub{x_{a}}{y_{b}}{z_{c}}\right)^{*} & =\covsub{\bar{x}_{a}}{\bar{y}_{b}}{\bar{z}_{c}},\label{eq:Red_Gamma_Con2}
\end{align}
where the matrix $\cpx{Z}$ transforms $Z^{*}$ to $\bar{Z}$, and
$\cpx{c}$ transforms $z_{c}^{*}$ to $\bar{z}_{c}$. The constraints
(\ref{eq:Red_Gamma_Con1}) and (\ref{eq:Red_Gamma_Con2}) imply that
the contractions $\covlg{X}{Y}{Z}$ and $\covsub{x_{a}}{y_{b}}{z_{c}}$
should behave the same as $Z$ and $z_{c}$ under the complex conjugate
operation. It is also natural to impose similar constraints on subgroup
irreps embedded in a large group irrep, meaning that $\cpx{c}$ should
transform the complex conjugate of $z_{c}$ components of $Z$ to
$\bar{z}_{c}$ components of $\bar{Z}$, 
\[
\cpx{c}\left(P_{Z\to c}Z\right)^{*}=P_{\bar{Z}\to\bar{c}}\bar{Z}=P_{\bar{Z}\to\bar{c}}\cpx{Z}Z^{*},
\]
from which it follows that
\begin{equation}
\cpx{c}P_{Z\to c}^{*}=P_{\bar{Z}\to\bar{c}}\cpx{Z}.\label{eq:Red_Gamma_Con3}
\end{equation}
\end{subequations}

These constraints lead to the following consequences.
\begin{prop}
\label{prop:CG_phase}The CG coefficients $M^{\left(x_{a}y_{b}\to z_{a}\right)}$
and $M^{\left(\bar{x}_{a}\bar{y}_{b}\to\bar{z}_{c}\right)}$ are related
by 
\begin{equation}
M^{\left(\bar{x}_{a}\bar{y}_{b}\to\bar{z}_{c}\right)}=\cpx{c}\left(M^{\left(x_{a}y_{b}\to z_{a}\right)}\right)^{*}\left(\cpx{a}^{-1}\otimes\cpx{b}^{-1}\right).\label{eq:Red_M_xaybzc}
\end{equation}
\end{prop}

\begin{proof}
The CG coefficients $M^{\left(x_{a}y_{b}\to z_{a}\right)}$ and $M^{\left(\bar{x}_{a}\bar{y}_{b}\to\bar{z}_{c}\right)}$
are defined as 
\[
z_{c}=M^{\left(x_{a}y_{b}\to z_{a}\right)}\left(x_{a}\otimes y_{b}\right),\quad\bar{z}_{c}=M^{\left(\bar{x}_{a}\bar{y}_{b}\to\bar{z}_{c}\right)}\left(\bar{x}_{a}\otimes\bar{y}_{b}\right).
\]
The constraint (\ref{eq:Red_Gamma_Con2}) then implies 
\begin{multline*}
\cpx{c}M^{\left(x_{a}y_{b}\to z_{a}\right)*}\left(x_{a}^{*}\otimes y_{b}^{*}\right)=M^{\left(\bar{x}_{a}\bar{y}_{b}\to\bar{z}_{c}\right)}\left(\bar{x}_{a}\otimes\bar{y}_{b}\right)\\
=M^{\left(\bar{x}_{a}\bar{y}_{b}\to\bar{z}_{c}\right)}\left(\cpx{a}\otimes\cpx{b}\right)\left(x_{a}^{*}\otimes y_{b}^{*}\right),
\end{multline*}
 where we used $\bar{x}_{a}=\cpx{x}x_{a}^{*}$ and $\bar{y}_{b}=\cpx{b}y_{b}^{*}$
in the second equality. Comparing the coefficients of both sides yields
eq. (\ref{eq:Red_M_xaybzc}).
\end{proof}
We see that the relation between $M^{\left(x_{a}y_{b}\to z_{c}\right)}$
and $M^{\left(\bar{x}_{a}\bar{y}_{b}\to\bar{z}_{c}\right)}$ depends
on basis of irreps. $M^{\left(\bar{x}_{a}\bar{y}_{b}\to\bar{z}_{c}\right)}$
is in general not simply the complex conjugate of $M^{\left(x_{a}y_{b}\to z_{c}\right)}$
unless that $\Gamma_{a,b,c}$ matrices are all identity matrices.
The matrix $\Gamma_{a}$ is the identity matrix in two cases: 1) $\rep{x_{a}}$
is complex; 2) $\rep{x_{a}}$ is real (not pseudoreal) irrep and is
in a basis that the matrices of its generators are all real. The latter
implies that, if all the three irreps are real, there exist bases
that the CG coefficients $M^{\left(x_{a}y_{b}\to z_{c}\right)}$ are
all real. The overall phase of $M^{\left(x_{a}y_{b}\to z_{c}\right)}$
is fixed up to a sign factor in two cases: i) $\rep{x_{a}}$, $\rep{y_{b}}$,
$\rep{z_{c}}$ are all real or pseudoreal irreps; ii) $\rep{z_{a}}$
is real or pseudoreal and $\rep{x_{a}}$ is the complex conjugate
of $\rep{y_{b}}$. In these two cases, the phase $\phi^{\left(ab\to c\right)}$
in eq. (\ref{eq:SBCG_Tran_M}) can only be $0$ or $\pi$.
\begin{prop}
When $\Gamma_{a}$ and $\Gamma_{X}$ are fixed, the phase ambiguity
of the projection matrix $P_{X\to a}$, denoted as $\theta_{a}^{\left(X\right)}$
in eq. (\ref{eq:SBCG_Tran_PXa}), are constrained by 
\[
e^{i\theta_{\bar{a}}^{\left(\bar{X}\right)}}=e^{-i\theta_{a}^{\left(X\right)}}.
\]
\end{prop}

\begin{proof}
This follows directly from  (\ref{eq:Red_Gamma_Con3}).
\end{proof}
If both $\rep{X}$ and $\rep{x_{a}}$ are real or pseudoreal, we have
$\theta_{a}^{\left(X\right)}$ is fixed to $0$ or $\pi$ and the
$U\left(1\right)$ ambiguity is reduced to a $\mathbb{Z}_{2}$ ambiguity;
if any of the irreps is complex, we have $\theta_{\bar{a}}^{\left(\bar{X}\right)}=-\theta_{a}^{\left(X\right)}$,
which implies that two $U\left(1\right)$ ambiguities are reduced
to one $U\left(1\right)$ ambiguity.
\begin{prop}
The embedding factors $\mathcal{E}_{c,ab}^{\left(XY\to Z\right)}$
and $\mathcal{E}_{\bar{c},\bar{a}\bar{b}}^{\left(\bar{X}\bar{Y}\to\bar{Z}\right)}$
satisfy
\begin{equation}
\mathcal{E}_{\bar{c},\bar{a}\bar{b}}^{\left(\bar{X}\bar{Y}\to\bar{Z}\right)}=\left(\mathcal{E}_{c,ab}^{\left(XY\to Z\right)}\right)^{*}.\label{eq:Red_M_cab}
\end{equation}
\end{prop}

\begin{proof}
The coefficients $\mathcal{E}_{c,ab}^{\left(XY\to Z\right)}$ and
$\mathcal{E}_{\bar{c},\bar{a}\bar{b}}^{\left(\bar{X}\bar{Y}\to\bar{Z}\right)}$
are defined as
\begin{align}
P_{Z\to c}Z & =\sum_{a,b}\mathcal{E}_{c,ab}^{\left(XY\to Z\right)}\covsub{P_{X\to a}X}{P_{Y\to b}Y}{z_{c}},\label{eq:Red_PZc_Z}\\
P_{\bar{Z}\to\bar{c}}\bar{Z} & =\sum_{a,b}\mathcal{E}_{\bar{c},\bar{a}\bar{b}}^{\left(\bar{X}\bar{Y}\to\bar{Z}\right)}\covsub{P_{\bar{X}\to\bar{a}}\bar{X}}{P_{\bar{Y}\to\bar{b}}\bar{Y}}{\bar{z}_{c}}.\label{eq:Red_PZc_Zb}
\end{align}
Applying eq. (\ref{eq:Red_Gamma_Con2}) to eq. (\ref{eq:Red_PZc_Zb})
yields
\begin{equation}
P_{\bar{Z}\to\bar{c}}\bar{Z}=\sum_{a,b}\mathcal{E}_{\bar{c},\bar{a}\bar{b}}^{\left(\bar{X}\bar{Y}\to\bar{Z}\right)}\Gamma_{c}\left(\covsub{P_{X\to a}X}{P_{Y\to b}Y}{z_{c}}\right)^{*}.\label{eq:Red_PZc_Zb2}
\end{equation}
On the other hand, we can also write $P_{\bar{Z}\to\bar{c}}\bar{Z}$
as, using eqs. (\ref{eq:Red_Gamma_Con3}) and (\ref{eq:Red_PZc_Z}),
\begin{multline}
P_{\bar{Z}\to\bar{c}}\bar{Z}=P_{\bar{Z}\to\bar{c}}\cpx{Z}Z^{*}=\Gamma_{c}P_{Z\to c}^{*}Z^{*}\\
=\sum_{a,b}\left(\mathcal{M}_{c,ab}^{\left(XY\to Z\right)}\right)^{*}\cpx{c}\left(\covsub{P_{X\to a}X}{P_{Y\to b}Y}{z_{c}}\right)^{*}\label{eq:Red_PZc_Zb3}
\end{multline}
 Comparing $\mathrm{rhs}$ of eqs. (\ref{eq:Red_PZc_Zb2}) and (\ref{eq:Red_PZc_Zb3})
gives rise to eq. (\ref{eq:Red_M_cab})
\end{proof}
Note that eq. (\ref{eq:Red_M_cab}) holds only when the CG coefficients
of the subgroup obey the constraint (\ref{eq:Red_Gamma_Con2}). The
constraint (\ref{eq:Red_M_cab}) on $\mathcal{E}_{c,ab}$ is much
simpler than the constraint (\ref{eq:Red_M_xaybzc}) on $M$. It is
basis-independent and $\mathcal{E}^{\left(\bar{X}\bar{Y}\to\bar{Z}\right)}$
is simply the complex conjugate of $\mathcal{E}^{\left(XY\to Z\right)}$.
The overall phase of $\mathcal{E}^{\left(XY\to Z\right)}$ is fixed
up to a sign factor when $\covlg{X}{Y}{Z}$ and $\covlg{\bar{X}}{\bar{Y}}{\bar{Z}}$
represent the same contraction, which occurs in the following two
cases:
\begin{itemize}
\item Case I: all of $\rep{X}$, $\rep{Y}$, and $\rep{Z}$ are real or
pseudoreal. Particularly, if subgroup irreps $\rep{z_{c}}$, $\rep{x_{a}}$,
$\rep{y_{b}}$ are also real or pseudo real, then the coefficient
$\mathcal{E}_{c,ab}^{\left(XY\to Z\right)}$ is a real number and
fixed up to a sign factor.
\item Case II: $\rep{Z}$ is real or pseudoreal and $\rep{X}=\rep{\bar{Y}}$
are complex.
\end{itemize}
For contractions of more than two vectors, we have similar constraints
as (\ref{eq:Red_Gamma_Con1}). Consider a contraction of three vectors
$\covlg{X}{\covlg{V}{W}{Y}}{Z}$ and its counterpart under complex
conjugate $\covlg{\bar{X}}{\covlg{\bar{V}}{\bar{W}}{\bar{Y}}}{\bar{Z}}$.
The relation between these two is
\begin{equation}
\cpx{Z}\left(\covlg{X}{\covlg{V}{W}{Y}}{Z}\right)^{*}=\covlg{\bar{X}}{\covlg{\bar{V}}{\bar{W}}{\bar{Y}}}{\bar{Z}}.\label{eq:Red_Gamma_XVW}
\end{equation}
 It can be shown as follows:
\begin{align*}
\mathrm{lhs} & =\cpx{Z}M^{\left(XY\to Z\right)*}\left(X^{*}\otimes\left(\covlg{V}{W}{Y}\right)^{*}\right)\\
 & =\cpx{Z}M^{\left(XY\to Z\right)*}\left(\cpx{X}^{-1}\otimes\cpx{Y}^{-1}\right)\left(\bar{X}\otimes\left(\covlg{\bar{V}}{\bar{W}}{\bar{Y}}\right)\right)\\
 & =M^{\left(\bar{X}\bar{Y}\to\bar{Z}\right)}\left(\bar{X}\otimes\left(\covlg{\bar{V}}{\bar{W}}{\bar{Y}}\right)\right)\\
 & =\covlg{\bar{X}}{\covlg{\bar{V}}{\bar{W}}{\bar{Y}}}{\bar{Z}}
\end{align*}
 where in the third equality we used eq. (\ref{eq:Red_M_xaybzc}).
We can generalize eq. (\ref{eq:Red_Gamma_XVW}) to contractions of
arbitrary number of vectors
\[
\cpx{Z}\left(\covlg{X\otimes Y\otimes\cdots}{W}{Z}\right)^{*}=\covlg{\bar{X}\otimes\bar{Y}\otimes\cdots}{\bar{W}}{\bar{Z}},
\]
 where the $\mathrm{rhs}$ is the complex conjugate counterpart of
the term inside the round parentheses of the $\mathrm{lhs}$ and we
have suppressed all the nesting structures and intermediate irreps.
When $\rep{Z}$ is the trivial singlet representation, $\Gamma_{Z}$
is the one-dimensional identity matrix and $\covlg{\bar{X}\otimes\bar{Y}\otimes\cdots}{\bar{W}}{\bar{Z}}$
is a c-number complex conjugate to $\covlg{X\otimes Y\otimes\cdots}{W}{Z}$.
Particularly, if the contraction $\covlg{X\otimes Y\otimes\cdots}{W}{1}$,
with $\rep{1}$ being the trivial singlet, is invariant under complex
conjugate operation of irreps, i.e. $\covlg{X\otimes Y\otimes\cdots}{W}{1}=\covlg{\bar{X}\otimes\bar{Y}\otimes\cdots}{\bar{W}}{1}$,
then the contraction is a real number.

\subsection{Remarks}

The above property has an implication in flavor physics models with
discrete flavor symmetries. To make the Lagrangian Hermitian, one
needs to add its complex conjugate to for each term of the Lagrangian.
For example, if a term like $\lambda\covlg{X}{\covlg{V}{W}{\bar{X}}}{1}$
is contained in the Lagrangian, then a counterpart term $\lambda^{*}\covlg{\bar{X}}{\covlg{\bar{V}}{\bar{W}}{X}}{1}$
is presumably contained in the Lagrangian as well. However, for general
CG coefficients, $\covlg{X}{\covlg{V}{W}{\bar{X}}}{1}$ and $\covlg{\bar{X}}{\covlg{\bar{V}}{\bar{W}}{X}}{1}$
are not necessary complex conjugate to each other. Therefore, the
coupling constant $\lambda$ in general has to be a complex number
with certain phase to make the Lagrangian real. With the embedding
factors defined under constraints (\ref{eq:CG_GammaCon}), the coefficient
$\lambda$ can always be a real number. 

One should not confuse the $\Gamma$ matrices with the unitary matrix
$\mathcal{U}$ of generalized charge-parity (CP) transformations\cite{CPChen:2014tpa,CPEcker:1981wv,CPEcker:1983hz,CPFeruglio:2012cw,CPHolthausen:2012dk,CPNeufeld:1987wa},
\begin{equation}
\phi_{i}\stackrel{CP}{\longrightarrow}\mathcal{U}_{i}\phi_{i}^{*}.\label{eq:Red_CP_Tran}
\end{equation}
\begin{equation}
\mathcal{U}_{i}\rho_{\rep{R}_{i}}\left(g\right)^{*}\mathcal{U}_{i}^{-1}=\rho_{\rep{R}_{i}}\left(u\left(g\right)\right),\quad\forall g\in G,\,\forall i,\label{eq:Red_CP}
\end{equation}
 where $u$ is an automorphism of the group. For a physical CP transformation,
$u$ should be class-inverting and involutory. If such an automorphism
exists, a model employing the group $G$ as flavor symmetry can be
invariant under transformation (\ref{eq:Red_CP_Tran}). The matrices
$\mathcal{U}_{i}$ are in general different with the $\Gamma$ matrices
defined as (\ref{eq:Red_Gamma_rho}). In fact, under the transformation
like $X\to\cpx{X}X^{*}$, we have 
\begin{equation}
X\to\bar{X},\quad\bar{X}\to\begin{cases}
-X & \text{pseudoreal }\rep{X}\\
X & \text{compex or real }\rep{X}
\end{cases}\label{eq:Red_Gamma_Tran}
\end{equation}
The Lagrangian is not invariant under (\ref{eq:Red_Gamma_Tran}) if
it contains contractions with odd number of pseudoreal irreps. But
there exist groups with pseudoreal irreps admitting a CP symmetry.
For example, the group $Q_{8}$, which has one 2-dimensional pseudoreal
irrep, admits a CP symmetry\cite{CPChen:2014tpa}. On the other hand,
even if a group does not have pseudoreal irreps, the transformation
$X\to\cpx{X}X^{*}$ is not necessary a CP transformation, since the
class-inverting involutory automorphism for such a transformation
might not exist. For example, the group $\subgb$, of which all irreps
are complex except the real trivial singlet, does not admit a physical
CP transformation\cite{CPChen:2014tpa} but the Lagrangian is invariant
under the transformation $X\to\cpx{X}X^{*}$.

\section{\label{sec:Calc_CGC}Procedure to calculate embedding factors}

In this section, we will describe the procedure to calculate embedding
factors. First, we remark that there are some existing methods to
calculate CG coefficients, for example, the Mathematica package Discrete\cite{Holthausen:2011vd},
which implements the algorithm of \olcite{van1978clebsch}, and the
method introduced by \olcite{Ludl:2009ft}. These methods work well
for low-dimensional irreps and groups with small order. However, when
it comes to CG coefficients of large discrete group or those involving
high-dimensional irreps, they are usually not effective\footnote{We do not know the order of groups or the dimension of irreps, beyond
which these methods become ineffective. Based on our testing, for
the group $\group$ of order $168$, the Discrete package did not
give any result for the calculation of CG coefficients of two six-dimensional
irreps, even after days of computation. The method of \olcite{Ludl:2009ft}
requires diagonalization of representation matrices by Mathematica,
which, based on our testing, failed to give any result for the six-dimensional
representation matrices of the group $\group$.}. Furthermore, our goal is to calculate the embedding factors, we
therefore introduce the following procedure. 

The step zero is to find representation matrices and CG coefficients
of its subgroup. Here, we assume that the subgroup are relatively
small and its representation matrices and CG coefficients are known
or easy to find. Moreover, the CG coefficients of the subgroup should
satisfy eq. (\ref{eq:Red_M_xaybzc}).
\begin{itemize}
\item Step I
\end{itemize}
Find representation matrices of $G$ in the subgroup basis. For simplicity,
we choose a basis that projection matrices are in the simplest form,
meaning that $P_{X\to j}^{\left(\mathrm{u}\right)}$ in (\ref{eq:SBCG_PXa})
are identity matrices\footnote{There is a special case that the projection matrix cannot be in the
trivial form. This occurs when two irreps of the group are identified
to the same irrep of the subgroup. For example, for the $\group$
group, both $\rep{3}$ and $\rep{\bar{3}}$ are identified to the
$\rep{3_{2}}$ of $S_{4}$. If we choose $P_{3\to3_{2}}^{\left(u\right)}$
to be a identify matrix, then $P_{\bar{3}\to3_{2}}^{\left(u\right)}$
cannot be identity matrix simultaneously. The key point here is to
make projection matrices as simple as possible.}. One can first find the representation matrices of low-dimensional
irreps then build the high-dimensional irreps from tensor products.
Usually the low-dimensional representation matrices, in a certain
basis, are already known in the literature or can be obtained from
the GAP\cite{GAP4}. We therefore focus on finding a similarity transformation
that transforms the matrices to the desired basis.

To find the similarity transformation, we need to diagonalize the
representation matrices, see Section \ref{sec:Gen_PSL27}. Entries
of these matrices are usually cyclotomic numbers, which are polynomials
of $n$-th roots of unity for certain fixed $n$. It is difficult
to find the eigenvectors of these matrices directly by Mathematica.
We developed a algorithm to perform arithmetic operation of cyclotomic
numbers. The details are discussed in Appendix \ref{sec:Tech_Detail}.
With the algorithm, we can find the eigenvector of a matrix for a
given eigenvalue, which, for low-dimensional irreps, usually can be
calculated directly by Mathematica.
\begin{itemize}
\item Step II
\end{itemize}
Write down the most general expression of a contraction $\covlg{X}{Y}{Z}$
in terms of subgroup contractions, as eq. (\ref{eq:SBCG_Def1}), with
undetermined embedding factors $\mathcal{E}_{c,ab}$ and then setup
equations for these coefficients. With eq. (\ref{eq:SBCG_Def1}),
we obtain the expression of $Z$ in terms of $\mathcal{E}_{c,ab}$
and bilinear forms of $X$ and $Y$. We then substitute the expression
of $Z$ into eq. (\ref{eq:SBCG_Def2}) with $g$ being generators
of $G$. If a generator is a member of the subgroup, then eq. (\ref{eq:SBCG_Def2})
is automatically satisfied. So we only need to substitute $g$ with
generators that are in the subgroup. By matching of the coefficients
of bilinear forms of $X$ and $Y$, we obtain homogeneous equations
with respect to the unknown variables $\mathcal{E}_{c,ab}$. In this
way, the number of equations we obtained are usually much more than
the number of unknown variables. Many of the equations are dependent
on others and hence redundant. 

Alternatively, instead of matching coefficients of bilinear forms,
we can generate the equations by replacing $X$ and $Y$ with some
constant vectors 
\begin{equation}
\left\{ X=V^{\left(p\right)},Y=W^{\left(p\right)}\right\} ,\quad p=1,2,\cdots,\label{eq:Proc_XY}
\end{equation}
 with each $p$ corresponding to one set of inputs. There are different
choices of the constant vectors $V^{\left(p\right)}$ and $W^{\left(p\right)}$.
A simple choice is that each vector has only one nonzero component,
i.e., 
\begin{equation}
V_{k}^{\left(p\right)}=\delta_{k,i_{p}},\,W_{k}^{\left(p\right)}=\delta_{k,j_{p}},\label{eq:Proc_VW}
\end{equation}
 where $\left\{ i_{p}\right\} $ and $\left\{ j_{p}\right\} $ are
two sets of appropriately chosen positive integers. In this way, we
can reduce the number of equations. Of course, we need to choose enough
number of $i_{p}$ and $j_{p}$ and there could still be redundant
equations and some of the equations are trivially $0=0$. If $\rep{X}\otimes\rep{Y}\to\rep{Z}$
has multiplicity $\mu_{Z}$, then there will be at most $N_{c}-\mu_{Z}$
independent equations, where $N_{c}$ is the number of the unknown
variables $\mathcal{M}_{c,ab}$.
\begin{itemize}
\item Step III
\end{itemize}
The third step is to solve for the unknown variables $\mathcal{E}_{c,ab}$.
In the solution of the homogeneous linear equations, there will be
$\mu_{Z}$ free variables and the other $N_{c}-\mu_{Z}$ variables
be expressed as linear combinations of these free variables. In principle,
we could solve these linear equations using standard methods. However,
as the coefficients of these linear equations are cyclotomic numbers,
which come from the matrices of group generators, the exact solutions
are usually involved. If we use Mathematica to solve the equations
directly, it usually cannot simplify the solution to appropriate forms. 

There are two ways to solve the issue. The first way is to use the
calculation technique of cyclotomic numbers. We can use the Gaussian
elimination algorithm with arithmetic operation of cyclotomic numbers
to solve the equations. To be efficient, the Gaussian elimination
procedure should apply to a set of independent equations, which can
be found by converting the coefficients of equations to floating numbers
and apply regular Gaussian elimination algorithm with certain error
tolerance.

The second way to solve the issue is to use a Mathematica programming
trick. Instead of solving the equations directly, we convert all the
coefficients to floating numbers and solve the equation numerically.
We then convert the float numbers back into exact numbers using the
Mathematica function RootApproximant\cite{Mathematica}. Because of
numerical instability, the dependency relations of the equations are
broken when the coefficients are converted to float numbers. Hence,
it is crucial to pick out the maximal set of independent equations
before solving the equations. Again, this can be done by the Gaussian
elimination algorithm with an appropriate error tolerance. The method
is less rigorous comparing to the first one. But we found it very
effective in the calculation of $\group$ CG coefficients.
\begin{itemize}
\item Step IV
\end{itemize}
The last step is to solve the constraint (\ref{eq:Red_M_cab}) and
orthonormalize the embedding factors. The constraint (\ref{eq:Red_M_cab})
simply means that the embedding factors of $\covlg{X}{Y}{Z}$ are
complex conjugate of those of $\covlg{\bar{X}}{\bar{Y}}{\bar{Z}}$.
Thus, care should be taken when $\covlg{X}{Y}{Z}$ is invariant under
the complex conjugate transformation of irreps, i.e., the Case I and
II of section \ref{sec:CGCs-in-subgroup}. The constraint under these
two cases can be solved as follows. Let $\left\{ \mathcal{E}_{A}\right\} $
be the sets of free parameters in the solution of step III, where
$A$ denotes subscripts of the form $\left(c,ab\right)$, then the
constraint (\ref{eq:Red_M_cab}) is translated into equations of the
form 
\begin{equation}
\sum_{A}\alpha_{A}\mathcal{E}_{A}=\sum_{A}\beta_{A}\mathcal{E}_{A}^{*},\label{eq:Proc_aPcP}
\end{equation}
 where $\left\{ \alpha_{A}\right\} $ and $\left\{ \beta_{A}\right\} $
are constant c-numbers. These equations can be solved by expressing
$\mathcal{E}_{A}$ in terms of real and imaginary parts, i.e., $\mathcal{E}_{A}=\mathcal{E}_{A}^{\left(r\right)}+i\mathcal{E}_{A}^{\left(i\right)}$.
Finally, if $\mu_{Z}>1$, we can use the Gram\textendash Schmidt process
to build $\mu_{Z}$ sets of orthonormal embedding factors. \\

We have shown the procedure to find embedding factors. Let us now
count the number of degrees of freedom (DOF) of embedding factors.
There are $\mu_{Z}$ free complex coefficients $\mathcal{E}_{c,ab}$
after solving the homogeneous linear equations in step III. For Case
I and II, the constraint (\ref{eq:Red_M_cab}) generates $\mu_{Z}$
independent equations as (\ref{eq:Proc_aPcP}), which reduce the $\mu_{Z}$
complex DOF to $\mu_{Z}$ real DOF. The normalization condition reduces
one more real DOF. Therefore, there are $\mu_{Z}$ sets of independent
embedding factors with $\mu_{Z}-1$ real DOF. For other cases, the
constraint (\ref{eq:Red_M_cab}) relates embedding factors of $\covlg{X}{Y}{Z}$
to those of $\covlg{\bar{X}}{\bar{Y}}{\bar{Z}}$. Then there are $2\mu_{Z}$
sets of embedding factors with $\mu_{Z}$ complex DOF. Again, the
normalization condition reduces one more real DOF. The $2\mu_{Z}$
sets of embedding factors therefore have $2\mu_{Z}-1$ real DOF. 

In the following sections, we will demonstrate the above procedure
with the group $\group$ and its subgroups $\subga$,, and $\subgb$.
We automate much of the procedures in Mathematica package files, which
can be found in \olcite{Chen:2017GitHub}. We note that the source
code can be adapted for different groups.

\section{\label{sec:Gen_PSL27}Representation matrices of $\group$}

In this section, we will find representation matrices of $\group$
in its subgroup bases. To begin with, let us first give a brief introduction
of the group and its subgroups. Much of the group theories can be
found in Appendix \ref{sec:Group-Theory}. 

$\group$, the largest discrete subgroup of $SU\left(3\right)$ of
order 168, is the projective special linear group of $\left(2\times2\right)$
matrices over $\mathbb{F}_{7}$, the finite Galois field of seven
elements. The generators of the group are defined as
\[
\left\langle A,B|A^{2}=B^{3}=\left(AB\right)^{7}=\left[A,B\right]^{4}=e\right\rangle ,
\]
 where $\left[A,B\right]\equiv A^{-1}B^{-1}AB$. It has six irreps
\cite{Luhn:2007yr}: the complex $\rep{3}$ and its conjugate, $\rep{\bar{3}}$,
as well as four reals, $\rep{1},\rep{6},\rep{7}$, and $\rep{8}$.
Two of its subgroups are $\subga$, generated by
\begin{equation}
a^{4}=b^{2}=\left(ab\right)^{3}=e.\label{eq:Gen_S4_presentation}
\end{equation}
and $\subgb$, generated by
\begin{equation}
c^{7}=d^{3}=1,\quad d^{-1}cd=c^{4}.\label{eq:Gen_T7_presentations}
\end{equation}
The $\subga$ and $\subgb$ generators can be expressed in terms of
$\mathcal{PSL}_{2}\left(7\right)$ generators as
\begin{equation}
a=\left[A,B\right],\quad b=\left(AB^{2}\left(AB\right)^{2}\right)^{2},\label{eq:Gen_Psl_S4_Relation}
\end{equation}
 and
\begin{equation}
c=AB,\quad d=AB\left(AB^{2}\right)^{2}\left(AB\right)^{2}\left(AB^{2}\right)^{2}.\label{eq:Gen_Psl_T7_Relation}
\end{equation}
 We note that these expression are not unique and they can be found
by GAP\cite{GAP4}. An example GAP code to find such relations can
be found in \olcite{Chen:2017GitHub}.

In the following subsection, we will choose bases of $\group$ irreps
so that the projection matrices $P_{X\to a}$ are in the form of eq.
(\ref{eq:SBCG_PXa}) with $P_{X\to a}^{\left(\mathrm{u}\right)}$
being identity matrices, except for a special case that we will see
shortly. In such bases, the subgroup irreps are contained in $\group$
irreps following their orders in the embedding relations shown as
table \ref{tab:PSL27-S4-T7}. For example, irreps of the subgroup
$\subgb$ are embedded in the $\rep{6}$ irrep as $\rep{6}=\rep{3}\oplus\rep{\bar{3}}$,
then the first three components of a sextet form a $\subgb$ triplet
and the last three components form a $\subgb$ anti-triplet. 

We will first find the representation matrices in the $\subga$ basis
and then the $\subgb$ basis. For both subgroups, we will first seek
for the representation matrices of $\rep{3}$ and $\rep{\bar{3}}$,
then build high-dimensional irreps from the tensor product of low-dimensional
irreps. We will denote contractions of $\group$ by square brackets
and those of $\subga$ and $\subgb$ by curly brackets.

\subsection{\label{subsec:PSL27-Gen-S4}In the $\subga$ basis}

To find representation matrices of $\group$ in the $\subga$ basis,
we first need to find representation matrices of $\subga$, which
usually can be obtained from the literature or the GAP, and then the
CG coefficients of $\subga$. But since $\subga$ has the subgroup
$A_{4}$, which is also a group popular in model building of flavor
physics, we seek for a matrix realization of $\subga$ in $A_{4}$
basis. $A_{4}$ has only one non-singlet irrep, $\rep{3}$, whose
matrix realization is given as eq. (\ref{eq:A4-generators}). The
procedures to find $S_{4}$ representation matrices in $A_{4}$ basis
are similar to what we will talk in this section but the calculation
is kind of trivial. So we simply give the results in Appendix \ref{sec:Group-Theory}
and focus on finding representation matrices of $\group$.

Although we require the projection matrices to be the simplest form,
there are still ambiguities in the representation matrices of $\group$
because of the phase ambiguity of subgroup representation spaces.
For complex representations $\rep{x}$ and $\rep{\bar{x}}$, if $x$
and $\bar{x}$ are vectors in their representation spaces satisfying
$\Gamma_{x}x^{*}=\bar{x}$, then $e^{i\theta}x$ and $e^{-i\theta}\bar{x}$
are also vectors of the representation spaces and satisfy the same
constraint; for real or pseudoreal representation $\rep{x}$, the
phase ambiguity becomes a sign ambiguity. Since $\subga$ is an ambivalent
group, whose irreps all are real or pseudoreal, the vectors of $\subga$
representation spaces have sign ambiguities. It then implies that,
consulting the embedding of $\subga$ irreps in $\group$ irreps as
table \ref{tab:PSL27-S4-T7}, representation matrices of $\group$
irreps $\rep{6}$, $\rep{7}$, and $\rep{8}$ are fixed up to similarity
transformation of diagonal sign-factor matrices,
\begin{equation}
O^{[\rep{R}]}\to\mathcal{S}^{[\rep{R}]}O^{[\rep{R}]}\left(\mathcal{S}^{[\rep{R}]}\right)^{\dagger},\quad O=A,\,B,\quad\rep{R}=\rep{6},\,\rep{7},\,\rep{8},\label{eq:S4_SMatrix}
\end{equation}
 where 
\begin{align*}
\mathcal{S}^{\left[\rep{6}\right]} & =\mathrm{diag}\left(\pm1,\pm\mathbf{I}_{2},\pm\mathbf{I}_{3}\right),\\
\mathcal{S}^{\left[\rep{7}\right]} & =\mathrm{diag}\left(\pm1,\pm\mathbf{I}_{3},\pm\mathbf{I}_{3}\right),\\
\mathcal{S}^{\left[\rep{8}\right]} & =\mathrm{diag}\left(\pm\mathbf{I}_{2},\pm\mathbf{I}_{3},\pm\mathbf{I}_{3}\right),
\end{align*}
 with $\mathbf{I}_{n}$ being the $n\times n$ identity matrix. Since
the triplet and anti-triplet decomposition are $\rep{3}=\rep{3_{2}}$
and $\rep{\bar{3}}=\rep{3_{2}}$, their representation matrices are
fixed because a similarity transformation of the above form does not
change the matrices. Under above transformations, the embedding factors
transform as eq. (\ref{eq:SBCG_Tran_M}) with $\phi^{\left(ab\to c\right)}=0$
and $\theta_{a,b,c}^{\left(X,Y,Z\right)}=0\text{ or }\pi$. In the
following results, the $\mathcal{S}$ matrices are chosen to be identity
matrices for simplicity. 

\subsubsection{The triplet representation}

Representation matrices of $\group$ generators already exist in the
literature. We will use the existing results and transform them to
the $\subga$ basis. In \olcite{Luhn:2007yr} the representation
matrices of $\group$ in triplet irrep are \begin{subequations}\label{eq:Gen_AB3_old}
\begin{align}
\tilde{A}^{[\rep{3}]} & =\frac{i}{\sqrt{7}}\begin{pmatrix}\eta^{2}-\eta^{5} & \eta-\eta^{6} & \eta^{4}-\eta^{3}\\
\eta-\eta^{6} & \eta^{4}-\eta^{3} & \eta^{2}-\eta^{5}\\
\eta^{4}-\eta^{3} & \eta^{2}-\eta^{5} & \eta-\eta^{6}
\end{pmatrix},\quad\eta=\exp\left(i2\pi/7\right)\label{eq:Gen_A3_old}\\
\tilde{B}^{[\rep{3}]} & =\frac{i}{\sqrt{7}}\begin{pmatrix}\eta^{3}-\eta^{6} & \eta^{3}-\eta & \eta-1\\
\eta^{2}-1 & \eta^{6}-\eta^{5} & \eta^{6}-\eta^{2}\\
\eta^{5}-\eta^{4} & \eta^{4}-1 & \eta^{5}-\eta^{3}
\end{pmatrix}.\label{eq:Gen_B3_old}
\end{align}
\end{subequations}Since $\rep{3}$ of $\group$ is identified to
$\rep{3_{2}}$ of $\subga$, the $S_{4}$ generators $a$ and $b$
in $\rep{3_{2}}$ irrep can be generated by $\tilde{A}^{[\rep{3}]}$
and $\tilde{B}^{[\rep{3}]}$. Using eq. (\ref{eq:Gen_Psl_S4_Relation}),
we have 
\[
\tilde{a}=\frac{i}{\sqrt{7}}\left(\begin{array}{ccc}
\eta^{5}-\eta & \eta^{5}-\eta^{3} & \eta^{3}-\eta^{2}\\
\eta^{6}-\eta^{4} & \eta^{3}-\eta^{2} & \eta^{3}-\eta^{6}\\
\eta^{6}-\eta^{5} & \eta^{5}-\eta & \eta^{6}-\eta^{4}
\end{array}\right),
\]
\[
\tilde{b}=-\frac{1}{7}\begin{pmatrix}\eta^{6}-\eta^{2} & \eta-\eta^{6} & \eta^{3}-\eta^{2}\\
\eta^{6}-\eta^{4} & \eta^{5}-\eta^{4} & \eta^{2}-\eta^{5}\\
\eta^{4}-\eta^{3} & \eta^{5}-\eta & \eta^{3}-\eta
\end{pmatrix}^{2}.
\]

We now want to find a unitary matrix $U$ that simultaneously transforms
$\tilde{a}$ to $a^{[\rep{3_{2}}]}$ and $\tilde{b}$ to $b^{[\rep{3_{2}}]}$,
where $a^{[\rep{3_{2}}]}$ and $b^{[\rep{3_{2}}]}$ are given as eqs.
(\ref{eq:S4_irr31}), \begin{subequations} 
\begin{align}
U^{\dagger}\tilde{a}U & =a^{[\rep{3_{2}}]},\label{eq:S4_Unitary_Transform_aa}\\
U^{\dagger}\tilde{b}U & =b^{[\rep{3_{2}}]}.\label{eq:S4_Unitary_Transform_bb}
\end{align}
\end{subequations} The matrix $U$ can be found as follows. Since
$\tilde{a}$ and $a^{[\rep{3_{2}}]}$ have the same eigenvalues $\left\{ 1,i,-i\right\} $,
there exist unitary matrices $U_{1}$ and $U_{2}$ such that 
\begin{equation}
U_{1}^{\dagger}\tilde{a}U_{1}=U_{2}^{\dagger}a^{[\rep{3_{2}}]}U_{2}=\mathrm{diag}\left(1,i,-i\right).\label{eq:S4_U1_U2}
\end{equation}
 Now the matrix $U$ can be written as 
\begin{equation}
U=U_{1}\begin{pmatrix}e^{i\theta_{1}}\\
 & e^{i\theta_{2}}\\
 &  & e^{i\theta_{3}}
\end{pmatrix}U_{2}^{\dagger}.\label{eq:S4_Unitary_Transform_U}
\end{equation}
Substituting above into eq. (\ref{eq:S4_Unitary_Transform_bb}), we
can solve for $\theta_{2}$ and $\theta_{3}$ in terms of $\theta_{1}$
and determine the matrix $U$ up to an irrelevant overall phase. To
diagonalize $\tilde{a}$, we can use the algorithm of cyclotomic number
calculation to find its eigenvectors.

Alternatively, an easier way to find $U$ is using the $A_{4}$ generators
$s=a^{2}$ and $t=ab$. In the desired basis, $t=\mathrm{diag}\left(1,\omega,\omega^{2}\right)$
is diagonal and, in the basis of (\ref{eq:Gen_AB3_old}), $\tilde{t}$
has a simple form 
\[
\tilde{t}=\tilde{a}\tilde{b}=\begin{pmatrix}\begin{array}{ccc}
0 & \eta^{4} & 0\\
0 & 0 & \eta\\
\eta^{2} & 0 & 0
\end{array}\end{pmatrix}.
\]
Replacing $\tilde{a}$ and $a^{[\rep{3_{2}}]}$ with $\tilde{t}$
and $t$ in eq. (\ref{eq:S4_U1_U2}) and repeating the calculation,
we find that $U_{2}$ is the identity matrix and\begin{subequations}\label{eq:Gen_U1phases}
\begin{equation}
U_{1}=\left(\begin{array}{ccc}
1 & 1 & 1\\
\eta^{3} & \omega\eta^{3} & \omega^{2}\eta^{3}\\
\eta^{2} & \omega^{2}\eta^{2} & \omega\eta^{2}
\end{array}\right),\quad\omega=\exp\left(i\frac{2\pi}{3}\right).
\end{equation}
Requiring that $U$ transforms $\tilde{s}=\tilde{a}^{2}$ to the matrix
$s$ of eq. (\ref{eq:A4-generators}), we obtain
\begin{align}
e^{i\left(\theta_{2}-\theta_{1}\right)} & =\frac{1}{28}\left[\left(-9-\sqrt{3}i\right)\eta^{5}-2\left(9+\sqrt{3}i\right)\eta^{4}\right]\nonumber \\
 & +\frac{1}{28}\left[\left(-6+4i\sqrt{3}\right)\eta^{3}+\left(6-4i\sqrt{3}\right)\eta^{2}\right]\nonumber \\
 & +\frac{1}{28}\left[\left(-3+9i\sqrt{3}\right)\eta+\sqrt{3}i-5\right],\\
e^{i\left(\theta_{3}-\theta_{1}\right)} & =\frac{1}{28}\left[\left(-9+\sqrt{3}i\right)\eta^{5}-2\left(9-\sqrt{3}i\right)\eta^{4}\right]\nonumber \\
 & +\frac{1}{28}\left[\left(-6-4i\sqrt{3}\right)\eta^{3}+\left(6+4i\sqrt{3}\right)\eta^{2}\right]\nonumber \\
 & +\frac{1}{28}\left[\left(-3-9i\sqrt{3}\right)\eta-\sqrt{3}i-5\right].
\end{align}
\end{subequations}Substituting eqs. (\ref{eq:Gen_U1phases}) and
$U_{2}=I$ into eq. (\ref{eq:S4_Unitary_Transform_U}), we obtain
the matrix $U$ in a complicated expression. Fortunately, applying
the unitary transformation with $U$ to $\tilde{A}^{\left[\rep{3}\right]}$
and $\tilde{B}^{\left[\rep{3}\right]}$ , we get simple expressions
of $A$ and $B$ in the desired basis\begin{subequations}

\begin{align}
A^{[\rep{3}]} & =\left(\begin{array}{ccc}
-\frac{1}{3} & \frac{2}{3}\omega & \frac{2}{3}\omega^{2}\\
\frac{2}{3}\omega^{2} & -\frac{1}{3} & \frac{2}{3}\omega\\
\frac{2}{3}\omega & \frac{2}{3}\omega^{2} & -\frac{1}{3}
\end{array}\right),\label{eq:S4_A3}\\
B^{[\rep{3}]} & =\left(\begin{array}{ccc}
\frac{2}{3} & -\frac{i\left(\sqrt{3}+\sqrt{7}\right)}{6}\omega^{2}\bar{b}_{7}^{2} & \frac{i\left(\sqrt{3}-\sqrt{7}\right)}{6}\omega\bar{b}_{7}^{2}\\
\frac{i\left(\sqrt{7}-\sqrt{3}\right)}{6}\omega b_{7}^{2} & -\frac{1}{3} & \frac{1+\sqrt{21}}{6}\omega^{2}\\
\frac{i\left(\sqrt{3}+\sqrt{7}\right)}{6}\omega^{2}b_{7}^{2} & \frac{1-\sqrt{21}}{6}\omega & -\frac{1}{3}
\end{array}\right),\label{eq:S4_B3}
\end{align}
 \end{subequations} where $b_{7}$ and $\bar{b}_{7}$ are pure phases
\begin{equation}
b_{7}=\frac{\eta+\eta^{2}+\eta^{4}}{\sqrt{2}}=\frac{-1+i\sqrt{7}}{2\sqrt{2}},\quad\bar{b}_{7}=b_{7}^{*}.
\end{equation}

The $\rep{\bar{3}}$ matrix realization is the complex conjugate of
$\rep{3}$. However, the projection matrix $P_{\bar{3}\to3_{2}}$
is not the identity matrix but equals the matrix $\Gamma_{3_{2}}$,
\[
P_{\bar{3}\to3_{2}}=\cpx{3_{2}}=\begin{pmatrix}1 & 0 & 0\\
0 & 0 & 1\\
0 & 1 & 0
\end{pmatrix}.
\]
It can be explained as follows. If $\rep{3_{2}}$ were a complex representation,
there would exist its complex conjugate $\rep{\bar{3}_{2}}$. The
decomposition of $\rep{3}$ of $\group$ to $\subga$ irreps would
be $\rep{3}=\rep{3_{2}}$, $\rep{\bar{3}}=\rep{\bar{3}_{2}}$ and
both of the projection matrices be the identity matrix. But now $\rep{3_{2}}$
and $\rep{\bar{3}_{2}}$ are equivalent and related by a similarity
transformation $\Gamma_{3_{2}}$. Therefore $\rep{3_{2}}$ and $\rep{\bar{3}}$
should also be related by the same similarity transformation, and
hence, $P_{\bar{3}\to3_{2}}=\Gamma_{3_{2}}$. 

\subsubsection{Sextet, Octet, and Septet Representations}

We now build representation matrices of high-dimensional irreps with
those of $\rep{3}$ and $\rep{\bar{3}}$. The generators in $\rep{6}$
irrep can be obtained from the tensor product $\rep{3}\otimes\rep{3}\to\rep{6}$
. The decompositions of $\group$ irreps into $\subga$ irreps 
\begin{equation}
\rep{6}=\rep{1_{0}}\oplus\rep{2}\oplus\rep{3_{1}},\quad\rep{3}=\rep{3_{2}},\label{eq:S4_embed6}
\end{equation}
and the tensor product of $\subga$ irreps
\[
\rep{3_{2}}\otimes\rep{3_{2}}\to\rep{1_{0}}\oplus\rep{2}\oplus\rep{3_{1}}
\]
determine the $\group$ contraction $\covp{3}{3}{6}$ to be 
\begin{equation}
\covp{3}{3}{6}=\begin{pmatrix}e^{i\theta_{1}}\covs{3_{2}}{3_{2}}{1_{0}}\\
e^{i\theta_{2}}\covs{3_{3}}{3_{2}}{2}\\
e^{i\theta_{3}}\covs{3_{2}}{3_{2}}{3_{1}}
\end{pmatrix},\label{eq:S4_CG336}
\end{equation}
 where the phases are to be determined. The generators of $\rep{6}$
can be extracted from the equations 
\[
O^{[\rep{6}]}\rep{6}=\left[O^{\left[\rep{3}\right]}\rep{3}\otimes O^{\left[\rep{3}\right]}\rep{3}\right]_{\rep{6}},\quad O=A,B.
\]
We then obtain matrices $A^{[\rep{6}]}$ and $B^{[\rep{6}]}$ with
unknown phases $\theta_{i}$. The phases $\theta_{i}$ can be determined
by the constraints 
\begin{equation}
\cpx{\rep{6}}O^{\left[\rep{6}\right]*}\cpx{\rep{6}}^{\dagger}=O^{\left[\rep{6}\right]},\quad O=A,B,\label{eq:S4_Gamma6}
\end{equation}
 where $\cpx{\rep{6}}$ can be determined by eq. (\ref{eq:Red_Gamma_Con3})
with the projection matrices in the simplest form. It turns out that
\[
\cpx{\rep{6}}=\left(1\right)\oplus\cpx{\rep{2}}^{\left(S_{4}\right)}\oplus\cpx{\rep{3_{1}}}^{\left(S_{4}\right)}
\]
 with $\Gamma^{\left(S_{4}\right)}$ given as eq. (\ref{eq:S4_Gamma}).
With above constraints, we find that $e^{i\theta_{2}}=-i\bar{b}_{7}^{2}e^{i\theta_{1}}$
and $e^{i\theta_{3}}=i\bar{b}_{7}e^{i\theta_{1}}$ with $\theta_{1}$
being a free unphysical phase. Choosing $\theta_{1}=0$, we obtain
the contraction 
\[
\covp{3}{3}{6}=\begin{pmatrix}\covs{3_{2}}{3_{2}}{1}\\
-i\bar{b}_{7}^{2}\covs{3_{3}}{3_{2}}{2}\\
i\bar{b}_{7}\covs{3_{2}}{3_{2}}{3_{1}}
\end{pmatrix},
\]
and the generators
\[
A^{[\rep{6}]}=\left(\begin{array}{cccccc}
1 & 0 & 0 & 0 & 0 & 0\\
0 & 1 & 0 & 0 & 0 & 0\\
0 & 0 & 1 & 0 & 0 & 0\\
0 & 0 & 0 & -\frac{1}{3} & \frac{2\omega}{3} & \frac{2\omega^{2}}{3}\\
0 & 0 & 0 & \frac{2\omega^{2}}{3} & -\frac{1}{3} & \frac{2\omega}{3}\\
0 & 0 & 0 & \frac{2\omega}{3} & \frac{2\omega^{2}}{3} & -\frac{1}{3}
\end{array}\right),
\]
 
\begin{multline*}
B^{[\rep{6}]}=\\
\left(\begin{array}{cccccc}
-\frac{1}{6} & -\frac{\sqrt{7}\omega}{6} & -\frac{\sqrt{7}\omega^{2}}{6} & \frac{\sqrt{7}}{6} & \frac{\sqrt{7}\omega}{6} & \frac{\sqrt{7}\omega^{2}}{6}\\
-\frac{\sqrt{7}}{6} & \frac{\omega}{3} & \frac{\omega^{2}}{3} & -\frac{2\omega^{2}}{3}-\frac{1}{6} & -\frac{\omega^{2}}{3} & -\frac{1}{3}\\
-\frac{\sqrt{7}}{6} & \frac{\omega}{3} & \frac{\omega^{2}}{3} & -\frac{2\omega}{3}-\frac{1}{6} & -\frac{1}{3} & -\frac{\omega}{3}\\
\frac{\sqrt{7}}{6} & -\frac{\omega}{6}-\frac{2}{3} & -\frac{\omega^{2}}{6}-\frac{2}{3} & \frac{1}{6} & \frac{\omega}{6} & \frac{\omega^{2}}{6}\\
\frac{\sqrt{7}}{6} & -\frac{\omega^{2}}{3} & -\frac{\omega}{3} & \frac{1}{6} & -\frac{\omega}{3} & \frac{2\omega^{2}}{3}\\
\frac{\sqrt{7}}{6} & -\frac{\omega^{2}}{3} & -\frac{\omega}{3} & \frac{1}{6} & \frac{2\omega}{3} & -\frac{\omega^{2}}{3}
\end{array}\right).
\end{multline*}

The representation matrices of $\rep{8}$ can be calculated with the
tensor product $\rep{3}\otimes\rep{\bar{3}}\to\rep{8}$ and the embedding
relations 
\[
\rep{8}=\rep{2}\oplus\rep{3_{1}}\oplus\rep{3_{2}},\quad\rep{3}=\rep{3_{2}},\quad\rep{\bar{3}}=\rep{3_{2}}.
\]
 Since the complex conjugate of $\rep{3}\otimes\rep{\bar{3}}\to\rep{8}$
is itself, the overall phase of the CG coefficients is fixed. By a
little algebra, we find that
\[
\covp{3}{\bar{3}}{8}=\begin{pmatrix}\covs{3_{2}}{3_{2}}{2}\\
\covs{3_{2}}{3_{2}}{3_{1}}\\
i\covs{3_{2}}{3_{2}}{3_{2}}
\end{pmatrix}.
\]
The generator $A^{[\rep{8}]}$ has a simple form 
\begin{align*}
A^{[\rep{8}]} & =\left(\begin{array}{cccccccc}
1 & 0 & 0 & 0 & 0 & 0 & 0 & 0\\
0 & 1 & 0 & 0 & 0 & 0 & 0 & 0\\
0 & 0 & -\frac{1}{3} & \frac{2\omega}{3} & \frac{2\omega^{2}}{3} & 0 & 0 & 0\\
0 & 0 & \frac{2\omega^{2}}{3} & -\frac{1}{3} & \frac{2\omega}{3} & 0 & 0 & 0\\
0 & 0 & \frac{2\omega}{3} & \frac{2\omega^{2}}{3} & -\frac{1}{3} & 0 & 0 & 0\\
0 & 0 & 0 & 0 & 0 & -\frac{1}{3} & \frac{2\omega}{3} & \frac{2\omega^{2}}{3}\\
0 & 0 & 0 & 0 & 0 & \frac{2\omega^{2}}{3} & -\frac{1}{3} & \frac{2\omega}{3}\\
0 & 0 & 0 & 0 & 0 & \frac{2\omega}{3} & \frac{2\omega^{2}}{3} & -\frac{1}{3}
\end{array}\right).\\
\end{align*}
The generator $B^{[\rep{8}]}$ is given by \begin{widetext}
\[
B^{[\rep{8}]}=\left(\begin{array}{cccccccc}
-\frac{\omega}{4} & -\frac{\omega^{2}}{4} & -\frac{i\omega^{2}}{\sqrt{6}} & \frac{i\left(3\omega^{2}-1\right)}{4\sqrt{6}} & \frac{2i-\sqrt{3}\omega^{2}}{4\sqrt{6}} & 0 & -\frac{1}{4}\sqrt{\frac{7}{2}}\omega & \frac{1}{4}\sqrt{\frac{7}{2}}\omega^{2}\\
-\frac{\omega}{4} & -\frac{\omega^{2}}{4} & \frac{i\omega}{\sqrt{6}} & -\frac{2i+\sqrt{3}\omega}{4\sqrt{6}} & \frac{i\left(1-3\omega\right)}{4\sqrt{6}} & 0 & \frac{1}{4}\sqrt{\frac{7}{2}}\omega & -\frac{1}{4}\sqrt{\frac{7}{2}}\omega^{2}\\
-\frac{i}{\sqrt{6}} & \frac{i}{\sqrt{6}} & \frac{1}{6} & \frac{\omega}{6} & \frac{\omega^{2}}{6} & \frac{\sqrt{7}}{6} & \frac{\sqrt{7}\omega}{6} & \frac{\sqrt{7}\omega^{2}}{6}\\
\frac{i\left(3\omega^{2}-1\right)}{4\sqrt{6}} & \frac{i\left(1-3\omega\right)}{4\sqrt{6}} & \frac{1}{6} & \frac{13\omega}{24} & -\frac{5\omega^{2}}{24} & \frac{\sqrt{7}\omega}{6} & \frac{i\sqrt{3}\omega^{2}+1}{24/\sqrt{7}} & \frac{3\omega+1}{24/\sqrt{7}}\\
\frac{i\left(3\omega^{2}-1\right)}{4\sqrt{6}} & \frac{i\left(1-3\omega\right)}{4\sqrt{6}} & \frac{1}{6} & -\frac{5\omega}{24} & \frac{13\omega^{2}}{24} & \frac{\sqrt{7}\omega^{2}}{6} & \frac{3\omega^{2}+1}{24/\sqrt{7}} & \frac{1-i\sqrt{3}\omega}{24/\sqrt{7}}\\
0 & 0 & -\frac{\sqrt{7}}{6} & -\frac{\sqrt{7}}{6}\omega^{2} & -\frac{\sqrt{7}\omega}{6} & -\frac{1}{2} & -\frac{i\omega^{2}}{2\sqrt{3}} & \frac{i\omega}{2\sqrt{3}}\\
\frac{1}{4}\sqrt{\frac{7}{2}}\omega & -\frac{1}{4}\sqrt{\frac{7}{2}}\omega^{2} & -\frac{\sqrt{7}}{6} & -\frac{i\sqrt{3}\omega^{2}+1}{24/\sqrt{7}} & \frac{i\sqrt{3}\omega-1}{24/\sqrt{7}} & -\frac{i\omega}{2\sqrt{3}} & \frac{3\sqrt{3}\omega^{2}+i}{8\sqrt{3}} & -\frac{\omega^{2}}{8}\\
-\frac{1}{4}\sqrt{\frac{7}{2}}\omega & \frac{1}{4}\sqrt{\frac{7}{2}}\omega^{2} & -\frac{\sqrt{7}}{6} & -\frac{i\sqrt{3}\omega^{2}+1}{24/\sqrt{7}} & \frac{i\sqrt{3}\omega-1}{24/\sqrt{7}} & \frac{i\omega^{2}}{2\sqrt{3}} & -\frac{\omega}{8} & \frac{3\sqrt{3}\omega-i}{8\sqrt{3}}
\end{array}\right).
\]

Finally, let us consider the $\rep{7}$ irrep. We can obtain the $\rep{7}$
irrep from $\rep{3}\otimes\rep{6}\to\rep{7}$. However, unlike the
$\rep{6}$ case, in which the contraction $\covp{3}{3}{6}$ is determined
up to phases solely by embedding relations and subgroup tensor products,
the absolute values of the embedding factors of $\covp{3}{6}{7}$
cannot be determined now. Instead, we have to first determine the
embedding factors of $\covp{3}{6}{\bar{3}}$ and $\covp{3}{6}{8}$,
and then obtain those of $\covp{3}{6}{7}$ by orthogonality. By a
straightforward calculation, we obtain
\[
P_{\rep{\bar{3}}\to\rep{3_{2}}}\covp{3}{6}{\bar{3}}=\frac{1}{\sqrt{6}}\covs{3_{2}}{1_{0}}{3_{2}}-\frac{i\bar{b}_{7}^{2}}{\sqrt{3}}\covs{3_{2}}{2}{3_{2}}+\frac{i\bar{b}_{7}}{\sqrt{2}}\covs{3_{2}}{3_{1}}{3_{2}},
\]
\[
\covp{3}{6}{8}=\begin{pmatrix}\covs{3_{2}}{3_{2}}{2}\\
-\sqrt{\frac{2}{3}}b_{7}\covs{3_{2}}{2}{3_{1}}+\frac{1}{\sqrt{3}}\covs{3_{2}}{3_{1}}{3_{1}}\\
\frac{2\bar{b}_{7}}{7}\covs{3_{2}}{1_{0}}{3_{2}}-\frac{\sqrt{2}i}{3}b_{7}\covs{3_{2}}{2}{3_{2}}+\frac{i}{\sqrt{3}}\covs{3_{2}}{3_{1}}{3_{2}}
\end{pmatrix},
\]
 up to overall phases. The orthogonality of embedding factors determines
$\covp{3}{6}{7}$ to be 
\[
\covp{3}{6}{7}=\begin{pmatrix}e^{i\theta_{1}}\covs{3_{2}}{3_{1}}{1_{1}}\\
e^{i\theta_{2}}\left(\frac{1}{\sqrt{3}}\covs{3_{2}}{2}{3_{1}}+\sqrt{\frac{2}{3}}\bar{b}_{7}\covs{3_{2}}{3_{1}}{3_{1}}\right)\\
e^{i\theta_{3}}\left(\frac{i}{3}\sqrt{\frac{7}{2}}b_{7}\covs{3_{2}}{1_{0}}{3_{2}}-\frac{2i}{3}\bar{b}_{7}^{3}\covs{3_{2}}{2}{3_{2}}-\frac{i}{\sqrt{6}}\covs{3_{2}}{3_{1}}{3_{2}}\right)
\end{pmatrix},
\]
 with $\theta_{i}$ being unknown phases. Applying a constraint in
analogy to (\ref{eq:S4_Gamma6}) with $\cpx{\rep{7}}=\left(1\right)\oplus\cpx{\rep{3_{1}}}\oplus\cpx{\rep{3_{2}}}$,
we arrive at $e^{i\left(\theta_{2}-\theta_{1}\right)}=\bar{b}_{7}^{3}$
and $e^{i\left(\theta_{3}-\theta_{1}\right)}=-i\bar{b}_{7}^{2}$,
which then give rise to 
\[
A^{[\rep{7}]}=\left(\begin{array}{ccccccc}
1 & 0 & 0 & 0 & 0 & 0 & 0\\
0 & -\frac{1}{3} & \frac{2\omega}{3} & \frac{2\omega^{2}}{3} & 0 & 0 & 0\\
0 & \frac{2\omega^{2}}{3} & -\frac{1}{3} & \frac{2\omega}{3} & 0 & 0 & 0\\
0 & \frac{2\omega}{3} & \frac{2\omega^{2}}{3} & -\frac{1}{3} & 0 & 0 & 0\\
0 & 0 & 0 & 0 & -\frac{1}{3} & \frac{2\omega}{3} & \frac{2\omega^{2}}{3}\\
0 & 0 & 0 & 0 & \frac{2\omega^{2}}{3} & -\frac{1}{3} & \frac{2\omega}{3}\\
0 & 0 & 0 & 0 & \frac{2\omega}{3} & \frac{2\omega^{2}}{3} & -\frac{1}{3}
\end{array}\right),
\]
\[
B^{[\rep{7}]}=\left(\begin{array}{ccccccc}
0 & -\frac{i\bar{b}_{7}^{3}}{\sqrt{3}} & \frac{i\omega\bar{b}_{7}^{3}}{2\sqrt{3}} & \frac{i\omega^{2}\bar{b}_{7}^{3}}{2\sqrt{3}} & 0 & -\frac{i\omega\bar{b}_{7}^{2}}{2} & \frac{i\omega^{2}\bar{b}_{7}^{2}}{2}\\
-\frac{ib_{7}^{3}}{\sqrt{3}} & -\frac{1}{3} & -\frac{\omega}{3} & -\frac{\omega^{2}}{3} & \frac{ib_{7}}{3} & \frac{i\omega b_{7}}{3} & \frac{i\omega^{2}b_{7}}{3}\\
\frac{ib_{7}^{3}}{2\sqrt{3}} & -\frac{1}{3} & -\frac{\omega}{3} & -\frac{\omega^{2}}{3} & \frac{ib_{7}\omega^{2}}{3}+\frac{b_{7}}{2\sqrt{3}} & \frac{i\omega^{2}b_{7}}{3} & \frac{ib_{7}}{3}\\
\frac{ib_{7}^{3}}{2\sqrt{3}} & -\frac{1}{3} & -\frac{\omega}{3} & -\frac{\omega^{2}}{3} & \frac{i\omega b_{7}}{3}-\frac{b_{7}}{2\sqrt{3}} & \frac{ib_{7}}{3} & \frac{i\omega b_{7}}{3}\\
0 & \frac{i\bar{b}_{7}}{3} & \frac{i\bar{b}_{7}}{3}+\frac{\omega\bar{b}_{7}}{2\sqrt{3}} & \frac{i\bar{b}_{7}}{3}-\frac{\omega^{2}\bar{b}_{7}}{2\sqrt{3}} & 0 & -\frac{i\omega}{2\sqrt{3}} & \frac{i\omega^{2}}{2\sqrt{3}}\\
-\frac{ib_{7}^{2}}{2} & \frac{i\bar{b}_{7}}{3} & \frac{i\omega^{2}\bar{b}_{7}}{3} & \frac{i\omega\bar{b}_{7}}{3} & -\frac{i}{2\sqrt{3}} & -\frac{i\omega}{\sqrt{3}} & 0\\
\frac{ib_{7}^{2}}{2} & \frac{i\bar{b}_{7}}{3} & \frac{i\omega^{2}\bar{b}_{7}}{3} & \frac{i\omega\bar{b}_{7}}{3} & \frac{i}{2\sqrt{3}} & 0 & \frac{i\omega^{2}}{\sqrt{3}}
\end{array}\right).
\]
\end{widetext}

\subsection{In the $\subgb$ basis}

We now find the representation matrices of $\group$ in the $\subgb$
basis. In fact, \olcite{Luhn:2007yr} has built representation matrices
of $\group$ with those of $\subgb$. Our calculation is similar to
the one of \olcite{Luhn:2007yr} but with slightly different bases.
We will take care of the phase ambiguities of the matrices which is
not covered in \olcite{Luhn:2007yr}.

Since all the $\subgb$ irreps are complex except the trivial singlet,
according to the same arguments leading to eq. (\ref{eq:S4_SMatrix}),
the representation matrices of $\group$ irreps have phase ambiguities.
We can obtain new representation matrices by performing similarity
transformations with diagonal pure phase matrices,
\[
O^{\left[\rep{R}\right]}\to\mathcal{U}^{\left[\rep{R}\right]}O^{\left[\rep{R}\right]}\left(\mathcal{U}^{\left[\rep{R}\right]}\right)^{\dagger},\quad O=A,\,B
\]
 where 
\begin{align*}
\mathcal{U}^{\left[\rep{6}\right]} & =\mathrm{diag}\left(e^{i\theta_{6}}\mathbf{I}_{3},e^{-i\theta_{6}}\mathbf{I}_{3}\right),\\
\mathcal{U}^{\left[\rep{7}\right]} & =\mathrm{diag}\left(\pm1,e^{i\theta_{7}}\mathbf{I}_{3},e^{-i\theta_{7}}\mathbf{I}_{3}\right),\\
\mathcal{U}^{\left[\rep{8}\right]} & =\mathrm{diag}\left(e^{i\theta_{81}},e^{-i\theta_{81}},e^{i\theta_{83}}\mathbf{I}_{3},e^{-i\theta_{83}}\mathbf{I}_{3}\right).
\end{align*}
Under these transformations, the embedding factors change according
to eq. (\ref{eq:SBCG_Tran_M}) with $\phi^{\left(ab\to c\right)}=0$
and $\theta_{a,b,c}^{\left(X,Y,Z\right)}$ being the angles in the
$\mathcal{U}$ matrices. In the following, we will set all the $\mathcal{U}$
matrices to be identity matrix for simplicity. 

The triplet representation is given by eq. (\ref{eq:Gen_AB3_old}),
in which the $\subgb$ generator $c=AB$ is a diagonal matrix
\[
(AB)^{[\rep{3}]}=\mathrm{diag}\left(\eta,\eta^{2},\eta^{4}\right)\equiv\mathcal{P}^{[3]}.
\]
Since $\rep{3}$ and $\rep{\bar{3}}$ are the only two non-singlet
irreps of $\subgb$, and in all singlets irreps, $AB$ are diagonal
matrices in all $\group$ irreps in the $\subgb$ basis
\begin{align*}
\left(AB\right)^{[\rep{6}]} & =\mathrm{diag}\left(\mathcal{P}^{\left[3\right]},\bar{\mathcal{P}}^{[3]}\right),\\
\left(AB\right)^{[\rep{7}]} & =\mathrm{diag}\left(1,\mathcal{P}^{[3]},\bar{\mathcal{P}}^{[3]}\right),\\
\left(AB\right)^{[\rep{8}]} & =\mathrm{diag}\left(\mathbf{I}_{2},\mathcal{P}^{[3]},\bar{\mathcal{P}}^{[3]}\right),
\end{align*}
 where $\bar{\mathcal{P}}^{[3]}=\left(\mathcal{P}^{[3]}\right)^{\dagger}$.
Given since $A^{2}=e$, the matrices of $B$ can be expressed as 
\[
B^{[\rep{R}]}=A^{[\rep{R}]}\left(AB\right)^{[\rep{R}]},\quad\rep{R}=\rep{3},\,\rep{\bar{3}},\,\rep{6},\,\rep{7},\,\rep{8}.
\]
 We therefore only need to find the matrices of $A$ in the following.

The $\rep{6}$ irrep can be extract from the tensor product $\rep{3}\otimes\rep{3}\to\rep{6}$.
Repeating the calculation for the $\subga$ case, up to a similarity
transformation of $\mathcal{U}^{[\rep{6}]}$, we obtain the matrix
of $A$ \begin{subequations}\label{eq:Gen_T7_A6}

\begin{equation}
A^{[\rep{6}]}=\frac{2}{7}\begin{pmatrix}\mathcal{R}_{1} & \sqrt{2}\mathcal{R}_{2}\\
\sqrt{2}\mathcal{R}_{2} & \mathcal{R}_{1}
\end{pmatrix},
\end{equation}
 where $\mathcal{R}_{1,2}$ are real $3\times3$ symmetric matrices
\begin{align}
\mathcal{R}_{1} & =\begin{pmatrix}1-c_{2} & 1-c_{1} & 1-c_{3}\\
1-c_{1} & 1-c_{3} & 1-c_{2}\\
1-c_{3} & 1-c_{2} & 1-c_{1}
\end{pmatrix},\\
\mathcal{R}_{2} & =\begin{pmatrix}c_{1}-c_{3} & c_{3}-c_{2} & c_{2}-c_{1}\\
c_{3}-c_{2} & c_{2}-c_{1} & c_{1}-c_{3}\\
c_{2}-c_{1} & c_{1}-c_{3} & c_{3}-c_{2}
\end{pmatrix}
\end{align}
 \end{subequations} with $c_{k}=\cos\frac{2k\pi}{7}$.

Similarly, the $\rep{8}$ irrep can be obtained from the tensor product
$\rep{3}\otimes\rep{\bar{3}}\to\rep{8}$,\begin{subequations}\label{eq:Gen_T7_a8}
\begin{equation}
A^{[\rep{8}]}=\frac{2}{7}\begin{pmatrix}\mathcal{C}_{1} & \mathcal{C}_{2} & \mathcal{C}_{2}\\
\mathcal{C}_{2}^{\dagger} & \mathcal{R}_{3} & \mathcal{R}_{1}\\
\mathcal{C}_{2}^{\dagger} & \mathcal{R}_{1} & \mathcal{R}_{3}
\end{pmatrix},
\end{equation}
 where $\mathcal{C}_{1}$ and $\mathcal{C}_{2}$ are complex matrices
\begin{align}
\mathcal{C}_{1} & =\frac{1}{2}\begin{pmatrix}0 & \omega\\
\omega^{2} & 0
\end{pmatrix}+i\sqrt{3}\begin{pmatrix}0 & c_{2}-c_{3}\omega^{2}\\
-c_{2}+c_{3}\omega & 0
\end{pmatrix},\\
\mathcal{C}_{2} & =\frac{i}{2}\begin{pmatrix}-\omega^{2} & -1 & -\omega\\
\omega & 1 & \omega^{2}
\end{pmatrix}\nonumber \\
 & +\sqrt{3}\begin{pmatrix}c_{2}\omega-c_{3} & c_{2}\omega^{2}-c_{3}\omega & c_{2}-\omega^{2}c_{3}\\
c_{2}\omega^{2}-c_{3} & c_{2}\omega-c_{3}\omega^{2} & \left(c_{2}-\omega c_{3}\right)
\end{pmatrix},
\end{align}
and $\mathcal{R}_{3}$ is a real $3\times3$ symmetric matrix 
\begin{equation}
\mathcal{R}_{3}=\begin{pmatrix}c_{2}-c_{1} & c_{1}-c_{3} & c_{3}-c_{2}\\
c_{1}-c_{3} & c_{3}-c_{2} & c_{2}-c_{1}\\
c_{3}-c_{2} & c_{2}-c_{1} & c_{1}-c_{3}
\end{pmatrix}.
\end{equation}
\end{subequations}Note that $\mathcal{R}_{3}$ can be obtained by
applying the cyclic permutation $c_{1}\to c_{2}\to c_{3}\to c_{1}$
on $\mathcal{R}_{2}$. 

The $\rep{7}$ irrep can be extracted from the tensor product $\rep{3}\otimes\rep{6}\to\rep{\bar{3}}\oplus\rep{8}\oplus\rep{7}$.
We first find the embedding factors of $\rep{3}\otimes\rep{6}\to\rep{\bar{3}}$
and $\rep{3}\otimes\rep{6}\to\rep{8}$, then embedding factors of
$\rep{3}\otimes\rep{6}\to\rep{7}$ are fixed up to phases. The generator
$A^{[\rep{7}]}$ is \begin{widetext}
\[
A^{[\rep{7}]}=\frac{2}{7}\begin{pmatrix}-\frac{1}{2} & \sqrt{2} & \sqrt{2} & \sqrt{2} & \sqrt{2} & \sqrt{2} & \sqrt{2}\\
\sqrt{2} & 2c_{1}+c_{2} & c_{1}+2c_{3} & 2c_{2}+c_{3} & 2c_{2}+1 & 2c_{1}+1 & 2c_{3}+1\\
\sqrt{2} & c_{1}+2c_{3} & 2c_{2}+c_{3} & 2c_{1}+c_{2} & 2c_{1}+1 & 2c_{3}+1 & 2c_{2}+1\\
\sqrt{2} & 2c_{2}+c_{3} & 2c_{1}+c_{2} & c_{1}+2c_{3} & 2c_{3}+1 & 2c_{2}+1 & 2c_{1}+1\\
\sqrt{2} & 2c_{2}+1 & 2c_{1}+1 & 2c_{3}+1 & 2c_{1}+c_{2} & c_{1}+2c_{3} & 2c_{2}+c_{3}\\
\sqrt{2} & 2c_{1}+1 & 2c_{3}+1 & 2c_{2}+1 & c_{1}+2c_{3} & 2c_{2}+c_{3} & 2c_{1}+c_{2}\\
\sqrt{2} & 2c_{3}+1 & 2c_{2}+1 & 2c_{1}+1 & 2c_{2}+c_{3} & 2c_{1}+c_{2} & c_{1}+2c_{3}
\end{pmatrix}.
\]
\end{widetext}

To conclude this section, let us discuss the relation between our
bases with those of \olcite{Luhn:2007yr}. The triplet and anti-triplet
are the same. For the sextet, octet and septet irreps, the similarity
transformations between our bases and those of \olcite{Luhn:2007yr}
are
\[
\tilde{O}^{[\rep{R}]}=U^{[\rep{R}]}O^{[\rep{R}]}U^{[\rep{R}]^{\dagger}},\quad\rep{R}=\rep{6},\,\rep{7},\,\rep{8},\quad O=A,\,B,
\]
 where $\tilde{O}$ denotes the generators in the basis of \olcite{Luhn:2007yr}
and the $U$ matrices are
\[
U^{[\rep{6}]}=\begin{pmatrix}0 & 1 & 0 & 0 & 0 & 0\\
0 & 0 & 1 & 0 & 0 & 0\\
1 & 0 & 0 & 0 & 0 & 0\\
0 & 0 & 0 & 0 & 0 & 1\\
0 & 0 & 0 & 0 & 1 & 0\\
0 & 0 & 0 & 1 & 0 & 0
\end{pmatrix},
\]
\[
U^{[\rep{7}]}=\left(\begin{array}{ccccccc}
1 & 0 & 0 & 0 & 0 & 0 & 0\\
0 & 1 & 0 & 0 & 0 & 0 & 0\\
0 & 0 & 1 & 0 & 0 & 0 & 0\\
0 & 0 & 0 & 0 & 0 & 0 & 1\\
0 & 0 & 0 & 1 & 0 & 0 & 0\\
0 & 0 & 0 & 0 & 0 & 1 & 0\\
0 & 0 & 0 & 0 & 1 & 0 & 0
\end{array}\right),
\]
\[
U^{[\rep{8}]}=\frac{1}{\sqrt{2}}\begin{pmatrix}0 & 0 & 1 & 0 & 0 & 1 & 0 & 0\\
0 & 0 & i & 0 & 0 & -i & 0 & 0\\
\frac{1-\omega}{\sqrt{3}} & \frac{1-\omega^{2}}{\sqrt{3}} & 0 & 0 & 0 & 0 & 0 & 0\\
0 & 0 & 0 & 0 & 1 & 0 & 0 & 1\\
0 & 0 & 0 & 0 & -i & 0 & 0 & i\\
0 & 0 & 0 & 1 & 0 & 0 & 1 & 0\\
0 & 0 & 0 & i & 0 & 0 & -i & 0\\
-\omega^{2} & -\omega & 0 & 0 & 0 & 0 & 0 & 0
\end{pmatrix}.
\]

\section{\label{sec:Calc_PSL27_CGC}$\group$ embedding factors}

With the generator matrices, we can now find the embedding factors
of $\group$. To demonstrate the procedure, we will take the tensor
product $\rep{6}\otimes\rep{6}\to\rep{6}_{s}^{\left(1\right)}\oplus\rep{6}_{s}^{\left(2\right)}$
as an example. Here the superscripts $\left(1\right)$ and $\left(2\right)$
indicates the multiplicity of irreps in the tensor product and the
subscript $s$ ($a$) indicates the symmetric (antisymmetric) part. 

We first consider the embedding factors for $\group$ and the $\subga$
subgroup. Using the embedding relation $\rep{6}=\rep{1_{0}}\oplus\rep{2}\oplus\rep{3_{1}}$,
we can write down the most general expression of the $\covp{6}{6}{6_{s}}$
built out of $S_{4}$ contractions,\begin{subequations}\label{eq:Ex_SixSix2Six}
\begin{align}
\rep{1_{0}} & =\mathcal{E}_{1,11}\covs{1_{0}}{1_{0}}{1_{0}}+\mathcal{E}_{1,22}\covs{2}{2}{1_{0}}\nonumber \\
 & +\mathcal{E}_{1,33}\covs{3_{1}}{3_{1}}{1_{0}},\\
\rep{2} & =\frac{\mathcal{E}_{2,12}}{\sqrt{2}}\left(\covs{1_{0}}{2}{2}+\covs{2}{1_{0}}{2}\right)\nonumber \\
 & +\mathcal{E}_{2,22}\covs{2}{2}{2}+\mathcal{E}_{2,33}\covs{3_{1}}{3_{1}}{2},\\
\rep{3_{1}} & =\frac{\mathcal{E}_{3,13}}{\sqrt{2}}\left(\covs{1_{0}}{3_{1}}{3_{1}}+\covs{3_{1}}{1_{0}}{3_{1}}\right)\nonumber \\
 & +\frac{\mathcal{E}_{3,23}}{\sqrt{2}}\left(\covs{2}{3_{1}}{3_{1}}+\covs{3_{1}}{2}{3_{1}}\right)\nonumber \\
 & +\mathcal{E}_{3,33}\covs{3_{1}}{3_{1}}{3_{1}},
\end{align}
\end{subequations} where $\mathcal{E}_{c,ab}$ are embedding factors
to be determined. In the above expressions, the notation $\covs{x}{y}{z}$
denotes the contraction of two $\subga$ irreps, where $\rep{x}$
components is embedded in the first $\rep{6}$ and $\rep{y}$ embedded
in the second $\rep{6}$. 

Since $\rep{6}\otimes\rep{6}\to\rep{6}_{s}$ has multiplicity $2$
and there are nine unknowns, we can setup seven homogeneous linear
equations by feeding values into $X$, $Y$, and $g$ of eq. (\ref{eq:SBCG_Def2}).
To do this, we choose $g$ to be the generator $B$ and $X$, $Y$
to be constant vectors, whose only one nonzero components are specified
by integer lists $\left\{ i_{p}\right\} $ and $\left\{ j_{p}\right\} $,
see eqs. (\ref{eq:Proc_XY}) and (\ref{eq:Proc_VW}). It turns out
that choosing $\left\{ i_{p}\right\} =\left\{ 1,3\right\} $ and $\left\{ j_{p}\right\} =\left\{ 1,4\right\} $
for $X$ and $Y$ is enough to generate the equations. By solving
these equations, we express all coefficients in terms of $\mathcal{E}_{1,11}$
and $\mathcal{M}_{1,22}$,\begin{subequations}\label{eq:Ex_sol_M}
\begin{align}
\mathcal{E}_{1,33} & =-\frac{\mathcal{E}_{1,11}}{\sqrt{3}}-\frac{2\mathcal{E}_{1,22}}{\sqrt{3}},\quad\mathcal{E}_{2,12}=\mathcal{E}_{1,22},\\
\mathcal{E}_{2,22} & =-\frac{2\mathcal{E}_{1,11}}{\sqrt{7}}-\frac{3\mathcal{E}_{1,22}}{\sqrt{14}},\\
\mathcal{E}_{2,33} & =\frac{\mathcal{E}_{1,22}}{\sqrt{42}}-\frac{4\mathcal{E}_{1,11}}{\sqrt{21}},\\
\mathcal{E}_{3,13} & =-\frac{\sqrt{2}}{3}\mathcal{E}_{1,11}-\frac{2\mathcal{E}_{1,22}}{3},\\
\mathcal{E}_{3,23} & =\frac{1}{3}\sqrt{\frac{2}{7}}\mathcal{E}_{1,22}-\frac{8\mathcal{E}_{1,11}}{3\sqrt{7}},\\
\mathcal{E}_{3,33} & =\sqrt{\frac{2}{21}}\mathcal{E}_{1,11}+\frac{5\mathcal{E}_{1,22}}{\sqrt{21}}.
\end{align}
\end{subequations}Since all the irreps of $\subga$ are real, the
constraints $\mathcal{E}_{c,ab}=\mathcal{E}_{\bar{c},\bar{a}\bar{b}}$
are satisfied if $\mathcal{E}_{1,11}$ and $\mathcal{E}_{1,22}$ are
real. To build two orthonormal sets of embedding factors, we consider
the embedding factors of the $\rep{1_{0}}$ component. We need to
find two orthonormal vectors $\left\{ \mathcal{E}_{1,11},\mathcal{E}_{1,22},\mathcal{E}_{1,33}\right\} $
corresponding to two independent solutions. By the Gram\textendash Schmidt
process, we find the solutions $\left(\mathcal{E}_{1,11},\mathcal{E}_{1,22}\right)=\left(\frac{\sqrt{3}}{2},0\right)$
and $\left(\mathcal{E}_{1,11},\mathcal{E}_{1,22}\right)=\left(-\frac{1}{2\sqrt{3}},\sqrt{\frac{2}{3}}\right)$
can generate such two vectors. Substituting the two solutions into
(\ref{eq:Ex_sol_M}) and (\ref{eq:Ex_SixSix2Six}), we obtain the
embedding factors of $\rep{6}\otimes\rep{6}\to\rep{6}_{s}^{\left(1\right)}\oplus\rep{6}_{s}^{\left(2\right)}$
shown in Appendix \ref{sec:App_CGC_PSL27_S4}.

The calculation of embedding factors in $\subgb$ basis is similar.
According to the decomposition $\rep{6}=\rep{3}\oplus\rep{\bar{3}}$
and tensor products of $\subgb$, $\covp{6}{6}{6_{s}}$ can be written
as
\begin{align*}
\rep{3} & =\mathcal{E}_{3,33}\covt{3}{3}{3}+\mathcal{E}_{3,\bar{3}\bar{3}}\covt{\bar{3}}{\bar{3}}{3_{s}}\\
 & +\frac{\mathcal{E}_{3,3\bar{3}}}{\sqrt{2}}\left(\covt{3}{\bar{3}}{3}+\covt{\bar{3}}{3}{3}\right),\\
\rep{\bar{3}} & =\mathcal{E}_{\bar{3},33}\covt{3}{3}{\bar{3}_{s}}+\mathcal{E}_{\bar{3},\bar{3}\bar{3}}\covt{\bar{3}_{s}}{\bar{3}}{\bar{3}}\\
 & +\frac{\mathcal{E}_{\bar{3},3\bar{3}}}{\sqrt{2}}\left(\covt{3}{\bar{3}}{\bar{3}}+\covt{\bar{3}}{3}{\bar{3}}\right),
\end{align*}
 where the curly brackets denote $\subgb$ contractions. The homogeneous
linear equations can be generated in the same way as $\subga$ case
with input $\left\{ i_{p}\right\} =\left\{ 1,3\right\} $ and $\left\{ j_{p}\right\} =\{2,4\}$
and generator $A$. Solving the linear equations yields
\begin{align*}
\mathcal{E}_{3,\bar{3}\bar{3}} & =-\sqrt{2}\mathcal{E}_{3,33}-\frac{\mathcal{E}_{3,\bar{3}\bar{3}}}{\sqrt{2}},\\
\mathcal{E}_{\bar{3},33} & =-\mathcal{E}_{3,33}-\mathcal{E}_{3,\bar{3}\bar{3}},\\
\mathcal{E}_{\bar{3},\bar{3}\bar{3}} & =\sqrt{2}\mathcal{E}_{3,33},\quad\mathcal{E}_{\bar{3},\bar{3}\bar{3}}=\frac{1}{\sqrt{2}}\mathcal{E}_{3,\bar{3}\bar{3}}.
\end{align*}
Now the constraints $\mathcal{E}_{c,ab}=\mathcal{E}_{\bar{c},\bar{a}\bar{b}}$
can be solved by expressing $\mathcal{E}_{3,33}$ and $\mathcal{E}_{3,\bar{3}\bar{3}}$
in terms of real and imaginary parts. It turns out that the solution
is $\mathcal{E}_{3,\bar{3}\bar{3}}=\mathcal{E}_{3,33}^{*}$. The Gram\textendash Schmidt
process then give rise to two sets of independent embedding factors
shown in Appendix \ref{sec:App_CGC_PSL27_T7}.

\section{Conclusion}

We have introduced the embedding factors, which express CG coefficients
of a discrete group in terms of CG coefficients of its subgroup. embedding
factors are fixed up to phase ambiguities and invariant under basis
transformations of irreps of the group and the subgroup. Their phase
ambiguities are reduced by a phase convention defined as (\ref{eq:CG_GammaCon}).
Particularly, the phase ambiguities are reduced to sign ambiguities
in the Case I and II of Section \ref{sec:RedPhaseAmbiguity}. We also
obtained complete sets of embedding factors of the group $\group$
in the bases of its subgroup $\subga$ and $\subgb$.

The work can be extended in several directions. One direction is to
apply the method to other discrete groups. To give a few examples,
the group $\Sigma\left(360\phi\right)$ of order 1080 has subgroups
$\Sigma\left(60\right)\simeq A_{5}$ of order $60$ and $\Sigma\left(36\phi\right)$
of order 108; the group $\Sigma\left(216\phi\right)$ of order 648
has the subgroup tree $\Sigma\left(216\phi\right)\supset\Sigma\left(72\phi\right)\supset\Sigma\left(36\phi\right)$,
where $\Sigma\left(72\phi\right)$ is of order 216 and $\Sigma\left(36\phi\right)$
of order 108. More subgroup examples can be found in \olcite{Merle:2011vy}.

Another possible direction is to simplify the procedure to find embedding
factors. As we have shown, embedding factors are basis independent
and fixed up to possible phases. It seems unnecessary to find representation
matrices of the group and subgroups in order to calculate these coefficients,
at least for their absolute values. A natural guess is that embedding
factors have something to do with the coset structure of the group.
The procedure to calculate these coefficients would be much simpler
if there exists a method to determine them without knowing the representation
matrices and CG coefficients of subgroups.
\begin{acknowledgments}
We thank Pierre Ramond for helpful discussions and comments on this
work. We are also very grateful to Jue Zhang for helpful comments
and careful reading of this manuscript. This research was supported
in part by the Department of Energy under Grant No. DE-SC0010296.
\end{acknowledgments}

\appendix
\numberwithin{equation}{section}
\counterwithin{table}{section}

\section{\label{sec:Group-Theory}Group Theory}

\subsection{\label{subsec:A4Group} The $A_{4}$ group}

The $A_{4}$ group is generated by two elements $s$ and $t$ which
fulfill 
\begin{equation}
s^{2}=t^{3}=\left(st\right)^{3}=e.\label{eq:A4_presentation}
\end{equation}
The group has four irreps: one trivial singlet representation $\rep{1}$,
two nontrivial singlets $\rep{1^{\prime}}$ and $\rep{\bar{1}^{\prime}}$
that are complex conjugate to each other, and one triplet $\rep{3}$.
The character table of $A_{4}$ is shown as table \ref{tab:A4_character}.
The $\rep{3}$ representation of the group can be written as 
\begin{equation}
s=\frac{1}{3}\begin{pmatrix}-1 & 2 & 2\\
2 & -1 & 2\\
2 & 2 & -1
\end{pmatrix},\quad t=\begin{pmatrix}1\\
 & \omega\\
 &  & \omega^{2}
\end{pmatrix},\quad\omega=e^{2\pi i/3}.\label{eq:A4-generators}
\end{equation}

\begin{center}
\begin{table}
\begin{centering}
\begin{tabular}{|c|c|c|c|c|}
\hline 
$A_{4}$ & $C_{1}$ & $4C_{2}^{\left[3\right]}$ & $3C_{4}^{\left[2\right]}$ & $4C_{3}^{\left[3\right]}$\tabularnewline
\hline 
\hline 
$\rep{1}$ & $1$ & $1$ & $1$ & $1$\tabularnewline
\hline 
$\rep{1^{\prime}}$ & $1$ & $\omega$ & $1$ & $\omega^{2}$\tabularnewline
\hline 
$\rep{\bar{1}^{\prime}}$ & $1$ & $\omega^{2}$ & $1$ & $\omega$\tabularnewline
\hline 
$\rep{3}$ & $3$ & $0$ & $-1$ & $0$\tabularnewline
\hline 
\end{tabular}
\par\end{centering}
\caption{\label{tab:A4_character}Character table of $A_{4}$, where $\omega=\exp\left(2\pi i/3\right)$.}
\end{table}
\par\end{center}

Let $\rep{x}$, $\rep{y}$ be two triplets, $\rep{s^{\prime}}$ and
$\rep{\bar{s}^{\prime}}$ be $\rep{1^{\prime}}$ and $\rep{\bar{1}^{\prime}}$
respectively. The the CG coefficients of $A_{4}$ in (\ref{eq:A4-generators})
representation are 
\begin{align*}
\left(\rep{x}\otimes\rep{y}\right)_{\rep{1}} & =\frac{1}{\sqrt{3}}\left(x_{1}y_{1}+x_{2}y_{3}+x_{3}y_{2}\right),\\
\left(\rep{x}\otimes\rep{y}\right)_{\rep{1^{\prime}}} & =\frac{1}{\sqrt{3}}\left(x_{3}y_{3}+x_{1}y_{2}+x_{2}y_{1}\right),\\
\left(\rep{x}\otimes\rep{y}\right)_{\rep{\bar{1}^{\prime}}} & =\frac{1}{\sqrt{3}}\left(x_{2}y_{2}+x_{1}y_{3}+x_{3}y_{1}\right),\\
\left(\rep{x}\otimes\rep{y}\right)_{\rep{3}_{s}} & =\frac{1}{\sqrt{6}}\begin{pmatrix}2x_{1}y_{1}-x_{2}y_{3}-x_{3}y_{2}\\
2x_{3}y_{3}-x_{1}y_{2}-x_{2}y_{1}\\
2x_{2}y_{2}-x_{1}y_{3}-x_{3}y_{1}
\end{pmatrix},\\
\left(\rep{x}\otimes\rep{y}\right)_{\rep{3}_{a}} & =\frac{1}{\sqrt{6}}\begin{pmatrix}x_{2}y_{3}-x_{3}y_{2}\\
x_{1}y_{2}-x_{2}y_{1}\\
x_{1}y_{3}-x_{3}y_{1}
\end{pmatrix},
\end{align*}
\[
\left(\rep{x}\otimes\rep{s^{\prime}}\right)_{\rep{3}}=s^{\prime}\begin{pmatrix}x_{3}\\
x_{1}\\
x_{2}
\end{pmatrix},\quad\left(\rep{x}\otimes\rep{s^{\prime\prime}}\right)_{\rep{3}}=s^{\prime\prime}\begin{pmatrix}x_{2}\\
x_{3}\\
x_{1}
\end{pmatrix}.
\]
 The subscript $s$ and $a$ in above expressions denote symmetric
and antisymmetric parts of tensor products respectively. The rest
CG coefficients of $A_{4}$ are trivial.

\subsection{\label{subsec:S4Group}The $S_{4}$ group}

The $S_{4}$ group is generated by two elements $a$ and $b$ which
fulfill
\begin{equation}
a^{4}=b^{2}=\left(ab\right)^{3}=e.\label{eq:S4_presentation}
\end{equation}
As $A_{4}$ is a subgroup of $S_{4}$, $s$ and $t$ in (\ref{eq:A4_presentation})
can be expressed as 
\begin{equation}
s=a^{2},\quad t=ab.\label{eq:App_S4A4_Gen_Relation}
\end{equation}
$S_{4}$ has five irreps: one trivial singlet $\rep{1}$, one nontrivial
singlet $\rep{1_{1}}$, one doublet $\rep{2}$, and two triplets $\rep{3_{1}}$
and $\rep{3_{2}}$. The character table of $S_{4}$ is shown as table
\ref{tab:S4_character}. The irreps of $S_{4}$ can be decomposed
into irreps of $A_{4}$:
\begin{align*}
\rep{1_{1}} & =\rep{1},\quad\rep{2}=\rep{1^{\prime}}\oplus\rep{\bar{1}^{\prime}},\\
\rep{3_{1}} & =\rep{3},\quad\rep{3_{2}}=\rep{3}.
\end{align*}

\begin{center}
\begin{table}
\begin{centering}
\begin{tabular}{|c|c|c|c|c|c|}
\hline 
$S_{4}$ & $C_{1}$ & $6C_{2}^{[2]}$ & $3C_{3}^{[2]}$ & $8C_{4}^{[3]}$ & $6C_{5}^{[4]}$\tabularnewline
\hline 
\hline 
$\rep{1_{0}}$ & 1 & 1 & 1 & 1 & 1\tabularnewline
\hline 
$\rep{1_{1}}$ & 1 & -1 & 1 & 1 & -1\tabularnewline
\hline 
$\rep{2}$ & 2 & 0 & 2 & -1 & 0\tabularnewline
\hline 
$\rep{3_{1}}$ & 3 & 1 & -1 & 0 & 1\tabularnewline
\hline 
$\rep{3_{2}}$ & 3 & -1 & -1 & 0 & 1\tabularnewline
\hline 
\end{tabular}
\par\end{centering}
\centering{}\caption{\label{tab:S4_character}Character table of $S_{4}$}
\end{table}
\par\end{center}

The non-singlet irreps of $S_{4}$ can be expressed as 
\begin{equation}
\rep{2}:\quad a=\begin{pmatrix}0 & 1\\
1 & 0
\end{pmatrix},\quad b=\left(\begin{array}{cc}
0 & \omega^{2}\\
\omega & 0
\end{array}\right),\label{eq:S4_irr2}
\end{equation}
\begin{subequations}\label{eq:S4_irr31}
\begin{align}
\rep{3_{1}}:\quad a^{[\rep{3_{1}}]} & =\frac{1}{3}\left(\begin{array}{ccc}
-1 & 2\omega^{2} & 2\omega\\
2\omega^{2} & 2\omega & -1\\
2\omega & -1 & 2\omega^{2}
\end{array}\right),\\
b^{[\rep{3_{1}}]} & =\frac{1}{3}\left(\begin{array}{ccc}
-1 & 2\omega^{2} & 2\omega\\
2\omega & 2 & -\omega^{2}\\
2\omega^{2} & -\omega & 2
\end{array}\right),\\
a^{[\rep{3_{2}}]} & =-a^{[\rep{3_{1}}]},\quad b^{[\rep{3_{2}}]}=-b^{[\rep{3_{1}}]}.
\end{align}
\end{subequations} Their $\Gamma$ matrices, see eq. (\ref{eq:Red_Gamma_rho})
for the definition, are 
\begin{equation}
\Gamma_{\rep{2}}=\begin{pmatrix}0 & 1\\
1 & 0
\end{pmatrix},\quad\Gamma_{\rep{3_{1}}}=\Gamma_{\rep{3_{2}}}=\begin{pmatrix}1 & 0 & 0\\
0 & 0 & 1\\
0 & 1 & 0
\end{pmatrix}.\label{eq:S4_Gamma}
\end{equation}
We remark that the matrices of (\ref{eq:S4_irr2}) and (\ref{eq:S4_irr31})
are in the $A_{4}$ basis, meaning that, in the $\rep{2}=\rep{1^{\prime}}\oplus\rep{\bar{1}^{\prime}}$
representation (\ref{eq:S4_irr2}), $A_{4}$ elements are diagonal
matrices, and in the $\rep{3_{1}}=\rep{3}$ and $\rep{3_{2}}=\rep{3}$
representations, $A_{4}$ elements are generated by the same matrices
as the ones in (\ref{eq:A4-generators}). To obtain these matrices,
we use GAP to obtain an arbitrary matrix realization of $S_{4}$ and
then perform similarity transformations to transform them to the desired
bases. The procedures are described in Section \ref{sec:Gen_PSL27}.

Using the CG coefficients of $A_{4}$, we obtain the CG coefficients
of two triplets of $S_{4}$. In the following CG coefficients, $\rep{x}$
and $\rep{y}$ are triplets, $\rep{z}=\left(\rep{s_{z}^{\prime}},\rep{\bar{s}_{z}^{\prime}}\right)$
and $\rep{w}=\left(\rep{s_{w}^{\prime}},\rep{\bar{s}_{w}^{\prime}}\right)$
are doublets, where $\rep{s^{\prime}}$ and $\rep{\bar{s}^{\prime}}$
are nontrivial singlets of $A_{4}$, and $\rep{s}$ is a trivial singlet
of $S_{4}$ and $A_{4}$.
\begin{itemize}
\item $\rep{3_{1}}\otimes\rep{3_{1}}\to\left(\rep{1_{0}}+\rep{2}+\rep{3_{1}}\right)_{s}+\rep{3_{2}}_{a}$:
\begin{align*}
\covs{x}{y}{1_{0}} & =\cova{x}{y}{1_{0}}\\
\covs{x}{y}{2} & =\begin{pmatrix}\cova{x}{y}{1^{\prime}}\\
\cova{x}{y}{\bar{1}^{\prime}}
\end{pmatrix}\\
\covs{x}{y}{3_{1}} & =\cova{x}{y}{3_{s}}\\
\covs{x}{y}{3_{2}} & =i\cova{x}{y}{3_{a}}
\end{align*}
\item $\rep{3_{2}}\otimes\rep{3_{2}}\to\left(\rep{1_{0}}+\rep{2}+\rep{3_{1}}\right)_{s}+\rep{3_{2}}_{a}$:
\begin{align*}
\covs{x}{y}{1_{0}} & =\cova{x}{y}{1_{0}}\\
\covs{x}{y}{2} & =\begin{pmatrix}\cova{x}{y}{1^{\prime}}\\
\cova{x}{y}{\bar{1}^{\prime}}
\end{pmatrix}\\
\covs{x}{y}{3_{1}} & =\cova{x}{y}{3_{s}}\\
\covs{x}{y}{3_{2}} & =i\cova{x}{y}{3_{a}}
\end{align*}
\item $\rep{3_{1}}\otimes\rep{3_{2}}\to\rep{1_{1}}+\rep{2}+\rep{3_{1}}+\rep{3_{2}}$:
\begin{align*}
\covs{x}{y}{1_{1}} & =\cova{x}{y}{1}\\
\covs{x}{y}{2} & =\begin{pmatrix}i\cova{x}{y}{1^{\prime}}\\
-i\cova{x}{y}{\bar{1}^{\prime}}
\end{pmatrix}\\
\covs{x}{y}{3_{1}} & =i\cova{x}{y}{3_{a}}\\
\covs{x}{y}{3_{2}} & =\cova{x}{y}{3_{s}}
\end{align*}
\item $\rep{3_{1}}\otimes\rep{2}\to\rep{3_{1}}+\rep{3_{2}}$: 
\begin{align*}
\covs{x}{z}{3_{1}} & =\frac{1}{\sqrt{2}}\cova{x}{s_{z}^{\prime}}{3}+\frac{1}{\sqrt{2}}\cova{x}{\bar{s}_{z}^{\prime}}{3}\\
\covs{x}{z}{3_{2}} & =\frac{i}{\sqrt{2}}\cova{x}{s_{z}^{\prime}}{3}-\frac{i}{\sqrt{2}}\cova{x}{\bar{s}_{z}^{\prime}}{3}
\end{align*}
\item $\rep{3_{2}}\otimes\rep{2}\to\rep{3_{1}}+\rep{3_{2}}$:
\begin{align*}
\covs{x}{z}{3_{1}} & =\frac{i}{\sqrt{2}}\cova{x}{s_{z}^{\prime}}{3}-\frac{i}{\sqrt{2}}\cova{x}{\bar{s}_{z}^{\prime}}{3}\\
\covs{x}{z}{3_{2}} & =\frac{1}{\sqrt{2}}\cova{x}{s_{z}^{\prime}}{3}+\frac{1}{\sqrt{2}}\cova{x}{\bar{s}_{z}^{\prime}}{3}
\end{align*}
\item $\rep{2}\otimes\rep{2}\to\left(\rep{1_{0}}+\rep{2}\right)_{s}+\rep{1_{1}}_{a}$:
\begin{align*}
\covs{z}{w}{1_{0}} & =\frac{1}{\sqrt{2}}\cova{s_{z}^{\prime}}{\bar{s}_{w}^{\prime}}{1}+\frac{1}{\sqrt{2}}\cova{\bar{s}_{z}^{\prime}}{s_{w}^{\prime}}{1}\\
\covs{z}{w}{1_{1}} & =\frac{i}{\sqrt{2}}\cova{s_{z}^{\prime}}{\bar{s}_{w}^{\prime}}{1}-\frac{i}{\sqrt{2}}\cova{\bar{s}_{z}^{\prime}}{s_{w}^{\prime}}{1}\\
\covs{z}{w}{2} & =\begin{pmatrix}\cova{\bar{s}_{z}^{\prime}}{\bar{s}_{w}^{\prime}}{1^{\prime}}\\
\cova{s_{z}^{\prime}}{s_{w}^{\prime}}{\bar{1}^{\prime}}
\end{pmatrix}
\end{align*}
\item $\rep{2}\otimes\rep{1_{1}}\to\rep{2}$:
\[
\covs{z}{s}{2}=\begin{pmatrix}i\cova{s_{z}^{\prime}}{s}{1^{\prime}}\\
-i\cova{\bar{s}_{z}^{\prime}}{s}{\bar{1}^{\prime}}
\end{pmatrix}.
\]
 
\end{itemize}
Note that the above CG coefficients in $A_{4}$ basis satisfy the
constraints of eqs. (\ref{eq:CG_GammaCon}).

\subsection{The $\subgb$ group}

The Frobenius group of order 21 is the smallest finite non-Abelian
subgroup of $SU\left(3\right)$. It contains elements of order three
and seven, with the presentation 
\begin{equation}
\left\langle c,d\big|c^{7}=d^{3}=1,d^{-1}cd=c^{4}\right\rangle .\label{eq:T7_presentation}
\end{equation}
Its irreps are, a real singlet, one complex triplet $\rep{3}$, a
complex singlet, $\rep{1^{\prime}}$, and their inequivalent conjugates,
$\rep{\bar{3}}$, and $\rep{\bar{1}^{\prime}}$. Their Kronecker products
are 
\begin{align*}
\rep{1^{\prime}}\otimes\rep{1^{\prime}} & =\rep{\bar{1}^{\prime}},\quad\quad\quad\quad\quad\;\,\rep{1^{\prime}}\otimes\rep{\bar{1}^{\prime}}=\rep{1}\\
\rep{3}\otimes\rep{1^{\prime}} & =\rep{3},\quad\quad\quad\quad\quad\quad\rep{3}\otimes\rep{\bar{1}^{\prime}}=\rep{3}\\
\rep{3}\otimes\rep{3}\; & =\left(\rep{3}+\rep{\bar{3}}\right)_{s}+\rep{\bar{3}}_{a},\;\;\,\rep{3}\otimes\rep{\bar{3}}\;=\rep{1}+\rep{1^{\prime}}+\rep{\bar{1}^{\prime}}+\rep{3}+\rep{\bar{3}}.
\end{align*}

The character table of $\subgb$ is shown as table \ref{tab:Character-T7}.

\begin{table}
\begin{centering}
\begin{tabular}{|c|c|c|c|c|c|}
\hline 
$\subgb$ & $C_{1}$ & $7C_{2}^{[3]}$ & $7C_{3}^{[2]}$ & $3C_{4}^{[7]}$ & $3C_{5}^{[7]}$\tabularnewline
\hline 
\hline 
$\rep{1}$ & 1 & 1 & 1 & 1 & 1\tabularnewline
\hline 
$\rep{1^{\prime}}$ & 1 & $\omega$ & $\omega^{2}$ & 1 & 1\tabularnewline
\hline 
$\rep{\bar{1}^{\prime}}$ & 1 & $\omega^{2}$ & $\omega$ & 1 & 1\tabularnewline
\hline 
$\rep{3_{1}}$ & 3 & 0 & 0 & $b_{7}$ & $\bar{b}_{7}$\tabularnewline
\hline 
$\rep{3_{2}}$ & 3 & -1 & -1 & $\bar{b}_{7}$ & $b_{7}$\tabularnewline
\hline 
\end{tabular}
\par\end{centering}
\caption{\label{tab:Character-T7}Character table of $\subgb$}

\end{table}

The CG coefficients of $\subgb$ are as follows\cite{Kile:2013gla}.
\begin{align*}
\covt{1^{\prime}}{\bar{3}}{3} & =\begin{pmatrix}\ket{1^{\prime}}\ket{1}\\
\omega\ket{1^{\prime}}\ket{2}\\
\omega^{2}\ket{1^{\prime}}\ket{3}
\end{pmatrix},\quad\covt{3}{3}{3}=\begin{pmatrix}\ket{3}\ket{3^{\prime}}\\
\ket{1}\ket{1^{\prime}}\\
\ket{2}\ket{2^{\prime}}
\end{pmatrix}\\
\covt{3}{3}{\bar{3}_{s}} & =\begin{pmatrix}\frac{1}{\sqrt{2}}\left(\ket{3}\ket{2^{\prime}}+\ket{2}\ket{3^{\prime}}\right)\\
\frac{1}{\sqrt{2}}\left(\ket{1}\ket{3^{\prime}}+\ket{3}\ket{1^{\prime}}\right)\\
\frac{1}{\sqrt{2}}\left(\ket{2}\ket{1^{\prime}}+\ket{1}\ket{2^{\prime}}\right)
\end{pmatrix},\\
\covt{3}{3}{\bar{3}_{s}} & =\begin{pmatrix}\frac{1}{\sqrt{2}}\left(\ket{3}\ket{2^{\prime}}-\ket{2}\ket{3^{\prime}}\right)\\
\frac{1}{\sqrt{2}}\left(\ket{1}\ket{3^{\prime}}-\ket{3}\ket{1^{\prime}}\right)\\
\frac{1}{\sqrt{2}}\left(\ket{2}\ket{1^{\prime}}-\ket{1}\ket{2^{\prime}}\right)
\end{pmatrix},\\
\covt{3}{\bar{3}}{3} & =\begin{pmatrix}\ket{2}\ket{\bar{1}}\\
\ket{3}\ket{\bar{2}}\\
\ket{1}\ket{\bar{3}}
\end{pmatrix},\quad\covt{3}{\bar{3}}{\bar{3}}=\begin{pmatrix}\ket{1}\ket{\bar{2}}\\
\ket{2}\ket{\bar{3}}\\
\ket{3}\ket{\bar{1}}
\end{pmatrix},\\
\covt{3}{\bar{3}}{1^{\prime}} & =\frac{1}{\sqrt{3}}\left(\ket{1}\ket{\bar{1}}+\omega^{2}\ket{2}\ket{\bar{2}}+\omega\ket{3}\ket{\bar{3}}\right),\\
\covt{3}{\bar{3}}{1^{\prime}} & =\frac{1}{\sqrt{3}}\left(\ket{1}\ket{\bar{1}}+\omega\ket{2}\ket{\bar{2}}+\omega^{2}\ket{3}\ket{\bar{3}}\right).
\end{align*}
 Since irreps of $\subgb$ are complex except the trivial singlet,
the rest CG coefficients can be obtained by taking complex conjugate
of above CG coefficients.

\subsection{The $\group$ group}

$\mathcal{PSL}_{2}\left(7\right)$ irreps can be decomposed into $S_{4}$
irreps as Table \ref{tab:PSL27-S4-T7}. 

\begin{table}
\begin{centering}
\begin{tabular}{|c|}
\hline 
$\group\supset\subga$\tabularnewline
\hline 
\hline 
$\rep{3}=\rep{3_{2}}$\tabularnewline
$\rep{\bar{3}}=\rep{3_{2}}$\tabularnewline
$\rep{6}=\rep{1}+\rep{2}+\rep{3_{1}}$\tabularnewline
$\rep{7}=\rep{1_{1}}+\rep{3_{1}}+\rep{3_{2}}$\tabularnewline
$\rep{8}=\rep{2}+\rep{3_{1}}+\rep{3_{2}}$\tabularnewline
\hline 
\end{tabular} %
\begin{tabular}{|c|}
\hline 
$\group\supset\subgb$\tabularnewline
\hline 
\hline 
$\rep{3}=\rep{3}$\tabularnewline
$\rep{\bar{3}}=\rep{\bar{3}}$\tabularnewline
$\rep{6}=\rep{3}+\rep{\bar{3}}$\tabularnewline
$\rep{7}=\rep{1}+\rep{3}+\rep{\bar{3}}$\tabularnewline
$\rep{8}=\rep{1^{\prime}}+\rep{\bar{1}^{\prime}}+\rep{3}+\rep{\bar{3}}$\tabularnewline
\hline 
\end{tabular} 
\par\end{centering}
\caption{\label{tab:PSL27-S4-T7}Embedding of $\subga$ and $\subgb$ in $\group$.}
\end{table}

The tensor products of $\group$ are as table \ref{tab:PSL27-Kronecker}.

\begin{table}
\begin{centering}
\begin{tabular}{|l|}
\hline 
$\group$ Tensor Products\tabularnewline
\hline 
\hline 
$\rep{3}\otimes\rep{3}=\rep{\bar{3}}_{a}+\rep{6}_{s}$\tabularnewline
$\rep{3}\otimes\rep{\bar{3}}=\rep{1}+\rep{8}$\tabularnewline
$\rep{3}\otimes\rep{6}=\rep{\bar{3}}+\rep{7}+\rep{8}$\tabularnewline
$\rep{\bar{3}}\otimes\rep{6}=\rep{3}+\rep{7}+\rep{8}$\tabularnewline
$\rep{3}\otimes\rep{7}=\rep{6}+\rep{7}+\rep{8}$\tabularnewline
$\rep{\bar{3}}\otimes\rep{7}=\rep{6}+\rep{7}+\rep{8}$\tabularnewline
$\rep{3}\otimes\rep{8}=\rep{3}+\rep{6}+\rep{7}+\rep{8}$\tabularnewline
$\rep{\bar{3}}\otimes\rep{8}=\rep{\bar{3}}+\rep{6}+\rep{7}+\rep{8}$\tabularnewline
$\rep{6}\otimes\rep{6}=\left(\rep{1}+\rep{6}+\rep{6}+\rep{8}\right)_{s}+\left(\rep{7}+\rep{8}\right)_{a}$\tabularnewline
$\rep{6}\otimes\rep{7}=\rep{3}+\rep{\bar{3}}+\rep{6}+\rep{7}+\rep{7}+\rep{8}+\rep{8}$\tabularnewline
$\rep{6}\otimes\rep{8}=\rep{3}+\rep{\bar{3}}+\rep{6}+\rep{6}+\rep{7}+\rep{7}+\rep{8}+\rep{8}$\tabularnewline
$\rep{7}\otimes\rep{7}=\left(\rep{1}+\rep{6}+\rep{6}+\rep{7}+\rep{8}\right)_{s}+\left(\rep{3}+\rep{\bar{3}}+\rep{7}+\rep{8}\right)_{a}$\tabularnewline
$\rep{7}\otimes\rep{8}=\rep{3}+\rep{\bar{3}}+\rep{6}+\rep{6}+\rep{7}+\rep{7}+\rep{8}+\rep{8}+\rep{8}$\tabularnewline
$\rep{8}\otimes\rep{8}=\left(\rep{1}+\rep{6}+\rep{6}+\rep{7}+\rep{8}+\rep{8}\right)_{s}+\left(\rep{3}+\rep{\bar{3}}+\rep{7}+\rep{7}+\rep{8}\right)_{a}$\tabularnewline
\hline 
\end{tabular}
\par\end{centering}
\caption{\label{tab:PSL27-Kronecker}Tensor products of $\group$.}
\end{table}

\section{\label{sec:Tech_Detail}Cyclotomic Numbers}

Cyclotomic numbers are elements of the cyclotomic field. The $n$th
cyclotomic field contains all $n$th roots of unity and the numbers
which can be expressed as a polynomials of $n$th roots of unity with
rational coefficients. The general form of an element in the field
is 
\begin{equation}
f=\sum_{k=0}^{n-1}q_{k}\exp\left(2k\pi i/n\right),\quad q_{k}\in\mathbb{Q}.\label{eq:CN_def}
\end{equation}
 Discrete group characters are usually cyclotomic numbers for fixed
$n$. We therefore need to perform arithmetic operators over cyclotomic
numbers in order to calculate CG coefficients of discrete groups.
Here we introduce the algorithm we used for calculations involving
cyclotomic numbers.

Cyclotomic fields are closed in the arithmetic of addition, multiplication,
and division. The addition and multiplication operators are trivially
the operators of polynomials. So we only need to discuss the division.
If we can find the reciprocal of a cyclotomic number (\ref{eq:CN_def}),
then the division operator becomes a multiplication operator and the
problem is solved.

Let us now consider how to find reciprocal of (\ref{eq:CN_def}).
For convenience, we now write $\exp\left(2\pi i/n\right)$ as $r_{n}$.
Let $g=\sum p_{k}r_{n}^{k}$ be the reciprocal of $f$. Collecting
$r_{n}^{k}$ terms in the product of $f$ and $g$, we have \begin{subequations}\label{eq:CN_equations}
\begin{align}
0 & =\sum_{i}q_{\left[k\right]-i}p_{i},\quad k>0,\,\left[k\right]=k,\,\text{or }n+k,\\
1 & =q_{0}p_{0}+\sum_{i}q_{n-i}p_{i}.
\end{align}
\end{subequations} We therefore have $n$ linear equations for $n$
unknown variables $p_{k}$. It seems enough to solve the equations.
However, these equations are in general not independent and therefore
have no solution. The root cause is that the expression (\ref{eq:CN_def})
is not unique. 

One of the ambiguities comes from that $r_{n}^{k}$ are not independent,
e.g,
\[
\sum_{k=0}^{n-1}r_{n}^{k}=0.
\]
We can always eliminate $r_{n}^{n-1}$ from the expression (\ref{eq:CN_def}).
Is it the only dependent relation among $r_{n}^{k}$? The answer is
no. If $n$ has a factor $1<p<n$, then
\begin{equation}
\sum_{k=0}^{n/p}r_{n}^{kp}=0\label{eq:CN_dependent_Relation}
\end{equation}
is also a dependent relation. If $n$ has $s$ positive factors (not
including $n$ itself), we can setup $s$ equations in the form of
(\ref{eq:CN_dependent_Relation}) and solve for $r_{n}^{k_{1}},\cdots,r_{n}^{k_{s}}$
in terms of the remaining $n-s$ $n$th root of unities. By doing
this, there are at most $n-s$ terms in (\ref{eq:CN_def}) and the
number of equations in (\ref{eq:CN_equations}) is $n-s$. It turns
out that, after doing this, the equations are always solvable and
has only one solution. 

\section{\label{sec:App_CGC_PSL27_S4}CG Coefficients of $\group$ in $\subga$
Basis}

The notations used in this appendix are as follows. If $\rep{X}$,
$\rep{Y}$, $\rep{Z}$ are decomposed to subgroup irreps as 
\[
\rep{X}=\bigoplus\rep{x_{a}},\quad\rep{Y}=\bigoplus\rep{y_{b}},\quad\rep{Z}=\bigoplus\rep{z_{c}}
\]
then the CG coefficients for $\rep{X}\otimes\rep{Y}\to\rep{Z}$ will
be written as a list of expressions of the form 
\begin{equation}
\rep{z_{c}}=\sum_{a,b}\mathcal{E}_{c,ab}\subcgs{x_{a}}{y_{b}}{z_{c}},\label{eq:App_CGList}
\end{equation}
 where $\mathcal{E}_{c,ab}$ are embedding factors and $\subcgs{x_{a}}{y_{b}}{z_{c}}$
are contractions of subgroup irreps. Note that, for each term $\subcgs{x_{a}}{y_{b}}{z_{c}}$
on the $\mathrm{rhs}$ of eq. (\ref{eq:App_CGList}), the first subgroup
irrep $\rep{x_{a}}$ always comes from $\rep{X}$ and second subgroup
irrep $\rep{y_{b}}$ always comes from $\rep{Y}$.

In the following, $b_{7}$ is a pure phase complex constant
\[
b_{7}=\frac{\eta+\eta^{2}+\eta^{4}}{\sqrt{2}}=\frac{-1+i\sqrt{7}}{2\sqrt{2}},\quad\eta=e^{2\pi i/7}.
\]

\paragraph*{\cgEqFontsize $\tsprod{3}{3}\to\rep{6}_{s}+\rep{\bar{3}}_{a}$}
\begin{itemize}
\item $\tsprodx{3}{3}\to\rep{\bar{3}}_{a}$: \bfl 
\begin{align*}
\rep{3}_{2} & =\subcgs{3_{2}}{3_{2}}{3_{2}}
\end{align*}
\efl 
\item $\tsprodx{3}{3}\to\rep{6}_{s}$: \bfl 
\begin{align*}
\rep{1}_{0} & =\subcgs{3_{2}}{3_{2}}{1_{0}}
\end{align*}
\begin{align*}
\rep{2} & =-\frac{i}{b_{7}^{2}}\subcgs{3_{2}}{3_{2}}{2}
\end{align*}
\begin{align*}
\rep{3}_{1} & =\frac{i}{b_{7}}\subcgs{3_{2}}{3_{2}}{3_{1}}
\end{align*}
\efl 
\end{itemize}

\paragraph*{\cgEqFontsize $\tsprod{3}{\bar{3}}\to\rep{1}+\rep{8}$}
\begin{itemize}
\item $\tsprodx{3}{\bar{3}}\to\rep{1}$: \bfl 
\begin{align*}
\rep{1}_{0} & =\subcgs{3_{2}}{3_{2}}{1_{0}}
\end{align*}
\efl 
\item $\tsprodx{3}{\bar{3}}\to\rep{8}$: \bfl 
\begin{align*}
\rep{2} & =\subcgs{3_{2}}{3_{2}}{2}
\end{align*}
\begin{align*}
\rep{3}_{1} & =\subcgs{3_{2}}{3_{2}}{3_{1}}
\end{align*}
\begin{align*}
\rep{3}_{2} & =i\subcgs{3_{2}}{3_{2}}{3_{2}}
\end{align*}
\efl 
\end{itemize}

\paragraph*{\cgEqFontsize $\tsprod{3}{6}\to\rep{\bar{3}}+\rep{7}+\rep{8}$}
\begin{itemize}
\item $\tsprodx{3}{6}\to\rep{\bar{3}}$: \bfl 
\begin{align*}
\rep{3}_{2} & =\frac{1}{\sqrt{6}}\subcgs{3_{2}}{1_{0}}{3_{2}}-\frac{i}{\sqrt{3}b_{7}^{2}}\subcgs{3_{2}}{2}{3_{2}}\\
 & +\frac{i}{\sqrt{2}b_{7}}\subcgs{3_{2}}{3_{1}}{3_{2}}
\end{align*}
\efl 
\item $\tsprodx{3}{6}\to\rep{7}$: \bfl 
\begin{align*}
\rep{1}_{1} & =\subcgs{3_{2}}{3_{1}}{1_{1}}
\end{align*}
\begin{align*}
\rep{3}_{1} & =\frac{1}{\sqrt{3}b_{7}^{3}}\subcgs{3_{2}}{2}{3_{1}}+\frac{\sqrt{\frac{2}{3}}}{b_{7}^{4}}\subcgs{3_{2}}{3_{1}}{3_{1}}
\end{align*}
\begin{align*}
\rep{3}_{2} & =\frac{\sqrt{\frac{7}{2}}}{3b_{7}}\subcgs{3_{2}}{1_{0}}{3_{2}}+\frac{2}{3b_{7}^{5}}\subcgs{3_{2}}{2}{3_{2}}\\
 & -\frac{1}{\sqrt{6}b_{7}^{2}}\subcgs{3_{2}}{3_{1}}{3_{2}}
\end{align*}
\efl 
\item $\tsprodx{3}{6}\to\rep{8}$: \bfl 
\begin{align*}
\rep{2} & =\subcgs{3_{2}}{3_{1}}{2}
\end{align*}
\begin{align*}
\rep{3}_{1} & =-\sqrt{\frac{2}{3}}b_{7}\subcgs{3_{2}}{2}{3_{1}}+\frac{1}{\sqrt{3}}\subcgs{3_{2}}{3_{1}}{3_{1}}
\end{align*}
\begin{align*}
\rep{3}_{2} & =\frac{2}{3b_{7}}\subcgs{3_{2}}{1_{0}}{3_{2}}-\frac{1}{3}i\sqrt{2}b_{7}\subcgs{3_{2}}{2}{3_{2}}\\
 & +\frac{i}{\sqrt{3}}\subcgs{3_{2}}{3_{1}}{3_{2}}
\end{align*}
\efl 
\end{itemize}

\paragraph*{\cgEqFontsize $\tsprod{3}{7}\to\rep{6}+\rep{7}+\rep{8}$}
\begin{itemize}
\item $\tsprodx{3}{7}\to\rep{6}$: \bfl 
\begin{align*}
\rep{1}_{0} & =\subcgs{3_{2}}{3_{2}}{1_{0}}
\end{align*}
\begin{align*}
\rep{2} & =-\frac{\sqrt{\frac{3}{7}}}{b_{7}^{2}}\subcgs{3_{2}}{3_{1}}{2}+\frac{2}{\sqrt{7}b_{7}^{4}}\subcgs{3_{2}}{3_{2}}{2}
\end{align*}
\begin{align*}
\rep{3}_{1} & =\sqrt{\frac{2}{7}}b_{7}\subcgs{3_{2}}{1_{1}}{3_{1}}-\frac{2}{\sqrt{7}b_{7}^{3}}\subcgs{3_{2}}{3_{1}}{3_{1}}\\
 & -\frac{1}{\sqrt{7}b_{7}}\subcgs{3_{2}}{3_{2}}{3_{1}}
\end{align*}
\efl 
\item $\tsprodx{3}{7}\to\rep{7}$: \bfl 
\begin{align*}
\rep{1}_{1} & =\subcgs{3_{2}}{3_{1}}{1_{1}}
\end{align*}
\begin{align*}
\rep{3}_{1} & =-\frac{1}{\sqrt{3}}\subcgs{3_{2}}{1_{1}}{3_{1}}-\frac{1}{\sqrt{3}b_{7}}\subcgs{3_{2}}{3_{1}}{3_{1}}\\
 & -\frac{1}{\sqrt{3}b_{7}}\subcgs{3_{2}}{3_{2}}{3_{1}}
\end{align*}
\begin{align*}
\rep{3}_{2} & =\frac{1}{\sqrt{3}b_{7}}\subcgs{3_{2}}{3_{1}}{3_{2}}+\frac{\sqrt{\frac{2}{3}}}{b_{7}^{2}}\subcgs{3_{2}}{3_{2}}{3_{2}}
\end{align*}
\efl 
\item $\tsprodx{3}{7}\to\rep{8}$: \bfl 
\begin{align*}
\rep{2} & =\frac{2}{\sqrt{7}}\subcgs{3_{2}}{3_{1}}{2}+\frac{\sqrt{\frac{3}{7}}}{b_{7}^{2}}\subcgs{3_{2}}{3_{2}}{2}
\end{align*}
\begin{align*}
\rep{3}_{1} & =\frac{2\sqrt{\frac{2}{21}}}{b_{7}^{4}}\subcgs{3_{2}}{1_{1}}{3_{1}}+\frac{\sqrt{\frac{2}{21}}}{b_{7}}\subcgs{3_{2}}{3_{1}}{3_{1}}\\
 & +\frac{13\sqrt{7}+7i}{28\sqrt{3}}\subcgs{3_{2}}{3_{2}}{3_{1}}
\end{align*}
\begin{align*}
\rep{3}_{2} & =\frac{\sqrt{\frac{2}{3}}}{b_{7}}\subcgs{3_{2}}{3_{1}}{3_{2}}-\frac{1}{\sqrt{3}b_{7}^{2}}\subcgs{3_{2}}{3_{2}}{3_{2}}
\end{align*}
\efl 
\end{itemize}

\paragraph*{\cgEqFontsize $\tsprod{3}{8}\to\rep{3}+\rep{6}+\rep{7}+\rep{8}$}
\begin{itemize}
\item $\tsprodx{3}{8}\to\rep{3}$: \bfl 
\begin{align*}
\rep{3}_{2} & =\frac{1}{2}\subcgs{3_{2}}{2}{3_{2}}+\frac{\sqrt{\frac{3}{2}}}{2}\subcgs{3_{2}}{3_{1}}{3_{2}}\\
 & -\frac{1}{2}i\sqrt{\frac{3}{2}}\subcgs{3_{2}}{3_{2}}{3_{2}}
\end{align*}
\efl 
\item $\tsprodx{3}{8}\to\rep{6}$: \bfl 
\begin{align*}
\rep{1}_{0} & =\subcgs{3_{2}}{3_{2}}{1_{0}}
\end{align*}
\begin{align*}
\rep{2} & =\frac{1}{2}\sqrt{3}b_{7}^{2}\subcgs{3_{2}}{3_{1}}{2}-\frac{1}{2}ib_{7}^{2}\subcgs{3_{2}}{3_{2}}{2}
\end{align*}
\begin{align*}
\rep{3}_{1} & =-\frac{b_{7}}{\sqrt{2}}\subcgs{3_{2}}{2}{3_{1}}-\frac{b_{7}}{2}\subcgs{3_{2}}{3_{1}}{3_{1}}\\
 & +\frac{ib_{7}}{2}\subcgs{3_{2}}{3_{2}}{3_{1}}
\end{align*}
\efl 
\item $\tsprodx{3}{8}\to\rep{7}$: \bfl 
\begin{align*}
\rep{1}_{1} & =\subcgs{3_{2}}{3_{1}}{1_{1}}
\end{align*}
\begin{align*}
\rep{3}_{1} & =-\frac{b_{7}^{4}}{\sqrt{3}}\subcgs{3_{2}}{2}{3_{1}}-\frac{b_{7}^{3}}{2\sqrt{3}}\subcgs{3_{2}}{3_{1}}{3_{1}}\\
 & +\frac{1}{2}\sqrt{\frac{7}{3}}b_{7}^{3}\subcgs{3_{2}}{3_{2}}{3_{1}}
\end{align*}
\begin{align*}
\rep{3}_{2} & =\frac{b_{7}^{2}}{2}\subcgs{3_{2}}{2}{3_{2}}+\frac{i\left(5\sqrt{7}+i\right)}{8\sqrt{6}}\subcgs{3_{2}}{3_{1}}{3_{2}}\\
 & +\frac{1}{2}\sqrt{\frac{7}{6}}b_{7}^{2}\subcgs{3_{2}}{3_{2}}{3_{2}}
\end{align*}
\efl 
\item $\tsprodx{3}{8}\to\rep{8}$: \bfl 
\begin{align*}
\rep{2} & =\frac{1}{2}\subcgs{3_{2}}{3_{1}}{2}+\frac{i\sqrt{3}}{2}\subcgs{3_{2}}{3_{2}}{2}
\end{align*}
\begin{align*}
\rep{3}_{1} & =\frac{1}{\sqrt{6}}\subcgs{3_{2}}{2}{3_{1}}-\sqrt{\frac{2}{3}}b_{7}^{3}\subcgs{3_{2}}{3_{1}}{3_{1}}\\
 & +\frac{ib_{7}}{\sqrt{6}}\subcgs{3_{2}}{3_{2}}{3_{1}}
\end{align*}
\begin{align*}
\rep{3}_{2} & =-\frac{i}{\sqrt{2}}\subcgs{3_{2}}{2}{3_{2}}-\frac{ib_{7}}{\sqrt{6}}\subcgs{3_{2}}{3_{1}}{3_{2}}\\
 & -\frac{b_{7}^{2}}{\sqrt{3}}\subcgs{3_{2}}{3_{2}}{3_{2}}
\end{align*}
\efl 
\end{itemize}

\paragraph*{\cgEqFontsize $\tsprod{6}{6}\to\left(\rep{1}+\rep{6}^{(1)}+\rep{6}^{(2)}+\rep{8}\right)_{s}+\left(\rep{7}+\rep{8}\right)_{a}$}
\begin{itemize}
\item $\tsprodx{6}{6}\to\rep{1}_{s}$: \bfl 
\begin{align*}
\rep{1}_{0} & =\frac{1}{\sqrt{6}}\subcgs{1_{0}}{1_{0}}{1_{0}}+\frac{1}{\sqrt{3}}\subcgs{2}{2}{1_{0}}\\
 & +\frac{1}{\sqrt{2}}\subcgs{3_{1}}{3_{1}}{1_{0}}
\end{align*}
\efl 
\item $\tsprodx{6}{6}\to\rep{6}_{s}^{(1)}$: \bfl 
\begin{align*}
\rep{1}_{0} & =\frac{\sqrt{3}}{2}\subcgs{1_{0}}{1_{0}}{1_{0}}-\frac{1}{2}\subcgs{3_{1}}{3_{1}}{1_{0}}
\end{align*}
\begin{align*}
\rep{2} & =-\sqrt{\frac{3}{7}}\subcgs{2}{2}{2}-\frac{2}{\sqrt{7}}\subcgs{3_{1}}{3_{1}}{2}
\end{align*}
\begin{align*}
\rep{3}_{1} & =-\frac{1}{2\sqrt{3}}\left(\subcgs{1_{0}}{3_{1}}{3_{1}}+\subcgs{3_{1}}{1_{0}}{3_{1}}\right)\\
 & -2\sqrt{\frac{2}{21}}\left(\subcgs{2}{3_{1}}{3_{1}}+\subcgs{3_{1}}{2}{3_{1}}\right)\\
 & +\frac{1}{\sqrt{14}}\subcgs{3_{1}}{3_{1}}{3_{1}}
\end{align*}
\efl 
\item $\tsprodx{6}{6}\to\rep{6}_{s}^{(2)}$: \bfl 
\begin{align*}
\rep{1}_{0} & =-\frac{1}{2\sqrt{3}}\subcgs{1_{0}}{1_{0}}{1_{0}}+\sqrt{\frac{2}{3}}\subcgs{2}{2}{1_{0}}\\
 & -\frac{1}{2}\subcgs{3_{1}}{3_{1}}{1_{0}}
\end{align*}
\begin{align*}
\rep{2} & =\frac{1}{\sqrt{3}}\left(\subcgs{1_{0}}{2}{2}+\subcgs{2}{1_{0}}{2}\right)\\
 & -\frac{2}{\sqrt{21}}\subcgs{2}{2}{2}+\frac{1}{\sqrt{7}}\subcgs{3_{1}}{3_{1}}{2}
\end{align*}
\begin{align*}
\rep{3}_{1} & =-\frac{1}{2\sqrt{3}}\left(\subcgs{1_{0}}{3_{1}}{3_{1}}+\subcgs{3_{1}}{1_{0}}{3_{1}}\right)\\
 & +\sqrt{\frac{2}{21}}\left(\subcgs{2}{3_{1}}{3_{1}}+\subcgs{3_{1}}{2}{3_{1}}\right)\\
 & +\frac{3}{\sqrt{14}}\subcgs{3_{1}}{3_{1}}{3_{1}}
\end{align*}
\efl 
\item $\tsprodx{6}{6}\to\rep{8}_{s}$: \bfl 
\begin{align*}
\rep{2} & =\frac{1}{\sqrt{6}}\left(\subcgs{1_{0}}{2}{2}+\subcgs{2}{1_{0}}{2}\right)\\
 & +2\sqrt{\frac{2}{21}}\subcgs{2}{2}{2}-\sqrt{\frac{2}{7}}\subcgs{3_{1}}{3_{1}}{2}
\end{align*}
\begin{align*}
\rep{3}_{1} & =\frac{1}{\sqrt{3}}\left(\subcgs{1_{0}}{3_{1}}{3_{1}}+\subcgs{3_{1}}{1_{0}}{3_{1}}\right)\\
 & -\frac{1}{\sqrt{42}}\left(\subcgs{2}{3_{1}}{3_{1}}+\subcgs{3_{1}}{2}{3_{1}}\right)\\
 & +\sqrt{\frac{2}{7}}\subcgs{3_{1}}{3_{1}}{3_{1}}
\end{align*}
\begin{align*}
\rep{3}_{2} & =-\frac{1}{\sqrt{2}}\left(\subcgs{2}{3_{1}}{3_{2}}+\subcgs{3_{1}}{2}{3_{2}}\right)
\end{align*}
\efl 
\item $\tsprodx{6}{6}\to\rep{7}_{a}$: \bfl 
\begin{align*}
\rep{1}_{1} & =\subcgs{2}{2}{1_{1}}
\end{align*}
\begin{align*}
\rep{3}_{1} & =\frac{\sqrt{\frac{7}{2}}}{3}\left(\subcgs{1_{0}}{3_{1}}{3_{1}}-\subcgs{3_{1}}{1_{0}}{3_{1}}\right)\\
 & +\frac{1}{3}\left(\subcgs{2}{3_{1}}{3_{1}}-\subcgs{3_{1}}{2}{3_{1}}\right)
\end{align*}
\begin{align*}
\rep{3}_{2} & =\frac{1}{\sqrt{3}}\left(\subcgs{2}{3_{1}}{3_{2}}-\subcgs{3_{1}}{2}{3_{2}}\right)\\
 & -\frac{1}{\sqrt{3}}\subcgs{3_{1}}{3_{1}}{3_{2}}
\end{align*}
\efl 
\item $\tsprodx{6}{6}\to\rep{8}_{a}$: \bfl 
\begin{align*}
\rep{2} & =\frac{1}{\sqrt{2}}\left(\subcgs{1_{0}}{2}{2}-\subcgs{2}{1_{0}}{2}\right)
\end{align*}
\begin{align*}
\rep{3}_{1} & =-\frac{1}{3}\left(\subcgs{1_{0}}{3_{1}}{3_{1}}-\subcgs{3_{1}}{1_{0}}{3_{1}}\right)\\
 & +\frac{\sqrt{\frac{7}{2}}}{3}\left(\subcgs{2}{3_{1}}{3_{1}}-\subcgs{3_{1}}{2}{3_{1}}\right)
\end{align*}
\begin{align*}
\rep{3}_{2} & =-\frac{1}{\sqrt{6}}\left(\subcgs{2}{3_{1}}{3_{2}}-\subcgs{3_{1}}{2}{3_{2}}\right)\\
 & -\sqrt{\frac{2}{3}}\subcgs{3_{1}}{3_{1}}{3_{2}}
\end{align*}
\efl 
\end{itemize}

\paragraph*{\cgEqFontsize $\tsprod{6}{7}\to\rep{3}+\rep{\bar{3}}+\rep{6}+\rep{7}^{(1)}+\rep{7}^{(2)}+\rep{8}^{(1)}+\rep{8}^{(2)}$}
\begin{itemize}
\item $\tsprodx{6}{7}\to\rep{3}$: \bfl 
\begin{align*}
\rep{3}_{2} & =\frac{1}{\sqrt{6}}\subcgs{1_{0}}{3_{2}}{3_{2}}+\frac{b_{7}^{2}}{\sqrt{7}}\subcgs{2}{3_{1}}{3_{2}}\\
 & +\frac{2b_{7}^{4}}{\sqrt{21}}\subcgs{2}{3_{2}}{3_{2}}+\frac{1}{\sqrt{7}b_{7}}\subcgs{3_{1}}{1_{1}}{3_{2}}\\
 & +\sqrt{\frac{2}{7}}b_{7}^{3}\subcgs{3_{1}}{3_{1}}{3_{2}}-\frac{b_{7}}{\sqrt{14}}\subcgs{3_{1}}{3_{2}}{3_{2}}
\end{align*}
\efl 
\item $\tsprodx{6}{7}\to\rep{\bar{3}}$: \bfl 
\begin{align*}
\rep{3}_{2} & =\frac{1}{\sqrt{6}}\subcgs{1_{0}}{3_{2}}{3_{2}}+\frac{1}{\sqrt{7}b_{7}^{2}}\subcgs{2}{3_{1}}{3_{2}}\\
 & +\frac{2}{\sqrt{21}b_{7}^{4}}\subcgs{2}{3_{2}}{3_{2}}+\frac{b_{7}}{\sqrt{7}}\subcgs{3_{1}}{1_{1}}{3_{2}}\\
 & +\frac{\sqrt{\frac{2}{7}}}{b_{7}^{3}}\subcgs{3_{1}}{3_{1}}{3_{2}}-\frac{1}{\sqrt{14}b_{7}}\subcgs{3_{1}}{3_{2}}{3_{2}}
\end{align*}
\efl 
\item $\tsprodx{6}{7}\to\rep{6}$: \bfl 
\begin{align*}
\rep{1}_{0} & =\subcgs{3_{1}}{3_{1}}{1_{0}}
\end{align*}
\begin{align*}
\rep{2} & =-\sqrt{\frac{3}{7}}\subcgs{2}{1_{1}}{2}+\frac{1}{\sqrt{7}}\subcgs{3_{1}}{3_{1}}{2}\\
 & -\sqrt{\frac{3}{7}}\subcgs{3_{1}}{3_{2}}{2}
\end{align*}
\begin{align*}
\rep{3}_{1} & =-\frac{1}{\sqrt{3}}\subcgs{1_{0}}{3_{1}}{3_{1}}-\sqrt{\frac{2}{21}}\subcgs{2}{3_{1}}{3_{1}}\\
 & -\sqrt{\frac{2}{7}}\subcgs{2}{3_{2}}{3_{1}}-\sqrt{\frac{2}{7}}\subcgs{3_{1}}{3_{2}}{3_{1}}
\end{align*}
\efl 
\item $\tsprodx{6}{7}\to\rep{7}^{(1)}$: \bfl 
\begin{align*}
\rep{1}_{1} & =\subcgs{1_{0}}{1_{1}}{1_{1}}
\end{align*}
\begin{align*}
\rep{3}_{1} & =-\frac{1}{6}\subcgs{1_{0}}{3_{1}}{3_{1}}+\frac{\sqrt{\frac{7}{2}}}{6}\subcgs{2}{3_{1}}{3_{1}}\\
 & +\frac{\sqrt{\frac{7}{6}}}{2}\subcgs{2}{3_{2}}{3_{1}}+\frac{\sqrt{\frac{7}{6}}}{2}\subcgs{3_{1}}{3_{1}}{3_{1}}\\
 & -\frac{\sqrt{\frac{7}{6}}}{2}\subcgs{3_{1}}{3_{2}}{3_{1}}
\end{align*}
\begin{align*}
\rep{3}_{2} & =-\frac{1}{6}\subcgs{1_{0}}{3_{2}}{3_{2}}+\frac{\sqrt{\frac{7}{6}}}{2}\subcgs{2}{3_{1}}{3_{2}}\\
 & +\frac{\sqrt{\frac{7}{2}}}{6}\subcgs{2}{3_{2}}{3_{2}}+\frac{\sqrt{\frac{7}{6}}}{2}\subcgs{3_{1}}{3_{1}}{3_{2}}\\
 & -\frac{\sqrt{\frac{7}{6}}}{2}\subcgs{3_{1}}{3_{2}}{3_{2}}
\end{align*}
\efl 
\item $\tsprodx{6}{7}\to\rep{7}^{(2)}$: \bfl 
\begin{align*}
\rep{1}_{1} & =\subcgs{3_{1}}{3_{2}}{1_{1}}
\end{align*}
\begin{align*}
\rep{3}_{1} & =-\frac{\sqrt{7}}{6}\subcgs{1_{0}}{3_{1}}{3_{1}}+\frac{7}{6\sqrt{2}}\subcgs{2}{3_{1}}{3_{1}}\\
 & -\frac{1}{2\sqrt{6}}\subcgs{2}{3_{2}}{3_{1}}-\frac{1}{2\sqrt{6}}\subcgs{3_{1}}{3_{1}}{3_{1}}\\
 & +\frac{1}{2\sqrt{6}}\subcgs{3_{1}}{3_{2}}{3_{1}}
\end{align*}
\begin{align*}
\rep{3}_{2} & =\frac{\sqrt{7}}{6}\subcgs{1_{0}}{3_{2}}{3_{2}}-\frac{1}{2\sqrt{6}}\subcgs{2}{3_{1}}{3_{2}}\\
 & -\frac{1}{6\sqrt{2}}\subcgs{2}{3_{2}}{3_{2}}+\frac{1}{\sqrt{3}}\subcgs{3_{1}}{1_{1}}{3_{2}}\\
 & -\frac{1}{2\sqrt{6}}\subcgs{3_{1}}{3_{1}}{3_{2}}-\frac{\sqrt{\frac{3}{2}}}{2}\subcgs{3_{1}}{3_{2}}{3_{2}}
\end{align*}
\efl 
\item $\tsprodx{6}{7}\to\rep{8}^{(1)}$: \bfl 
\begin{align*}
\rep{2} & =\frac{1}{\sqrt{2}}\subcgs{2}{1_{1}}{2}-\frac{1}{\sqrt{2}}\subcgs{3_{1}}{3_{2}}{2}
\end{align*}
\begin{align*}
\rep{3}_{1} & =\frac{1}{\sqrt{6}}\subcgs{2}{3_{2}}{3_{1}}-\sqrt{\frac{2}{3}}\subcgs{3_{1}}{3_{1}}{3_{1}}\\
 & -\frac{1}{\sqrt{6}}\subcgs{3_{1}}{3_{2}}{3_{1}}
\end{align*}
\begin{align*}
\rep{3}_{2} & =-\frac{2}{3}\subcgs{1_{0}}{3_{2}}{3_{2}}-2\sqrt{\frac{2}{21}}\subcgs{2}{3_{1}}{3_{2}}\\
 & -\frac{1}{3\sqrt{14}}\subcgs{2}{3_{2}}{3_{2}}+\frac{1}{\sqrt{21}}\subcgs{3_{1}}{1_{1}}{3_{2}}\\
 & +\sqrt{\frac{2}{21}}\subcgs{3_{1}}{3_{1}}{3_{2}}-\frac{1}{\sqrt{42}}\subcgs{3_{1}}{3_{2}}{3_{2}}
\end{align*}
\efl 
\item $\tsprodx{6}{7}\to\rep{8}^{(2)}$: \bfl 
\begin{align*}
\rep{2} & =\frac{1}{\sqrt{14}}\subcgs{2}{1_{1}}{2}+\sqrt{\frac{6}{7}}\subcgs{3_{1}}{3_{1}}{2}\\
 & +\frac{1}{\sqrt{14}}\subcgs{3_{1}}{3_{2}}{2}
\end{align*}
\begin{align*}
\rep{3}_{1} & =\frac{2}{3}\subcgs{1_{0}}{3_{1}}{3_{1}}+\frac{2\sqrt{\frac{2}{7}}}{3}\subcgs{2}{3_{1}}{3_{1}}\\
 & -\sqrt{\frac{3}{14}}\subcgs{2}{3_{2}}{3_{1}}-\sqrt{\frac{3}{14}}\subcgs{3_{1}}{3_{2}}{3_{1}}
\end{align*}
\begin{align*}
\rep{3}_{2} & =-\frac{1}{\sqrt{2}}\subcgs{2}{3_{2}}{3_{2}}-\frac{1}{\sqrt{3}}\subcgs{3_{1}}{1_{1}}{3_{2}}\\
 & -\frac{1}{\sqrt{6}}\subcgs{3_{1}}{3_{2}}{3_{2}}
\end{align*}
\efl 
\end{itemize}

\paragraph*{\cgEqFontsize $\tsprod{6}{8}\to\rep{3}+\rep{\bar{3}}+\rep{6}^{(1)}+\rep{6}^{(2)}+\rep{7}^{(1)}+\rep{7}^{(2)}+\rep{8}^{(1)}+\rep{8}^{(2)}$}
\begin{itemize}
\item $\tsprodx{6}{8}\to\rep{3}$: \bfl 
\begin{align*}
\rep{3}_{2} & =\frac{1}{\sqrt{6}}\subcgs{1_{0}}{3_{2}}{3_{2}}-\frac{1}{2b_{7}^{2}}\subcgs{2}{3_{1}}{3_{2}}\\
 & +\frac{i}{2\sqrt{3}b_{7}^{2}}\subcgs{2}{3_{2}}{3_{2}}-\frac{1}{2b_{7}}\subcgs{3_{1}}{2}{3_{2}}\\
 & +\frac{1}{2\sqrt{2}b_{7}}\subcgs{3_{1}}{3_{1}}{3_{2}}-\frac{i}{2\sqrt{2}b_{7}}\subcgs{3_{1}}{3_{2}}{3_{2}}
\end{align*}
\efl 
\item $\tsprodx{6}{8}\to\rep{\bar{3}}$: \bfl 
\begin{align*}
\rep{3}_{2} & =\frac{1}{\sqrt{6}}\subcgs{1_{0}}{3_{2}}{3_{2}}-\frac{b_{7}^{2}}{2}\subcgs{2}{3_{1}}{3_{2}}\\
 & -\frac{ib_{7}^{2}}{2\sqrt{3}}\subcgs{2}{3_{2}}{3_{2}}-\frac{b_{7}}{2}\subcgs{3_{1}}{2}{3_{2}}\\
 & +\frac{b_{7}}{2\sqrt{2}}\subcgs{3_{1}}{3_{1}}{3_{2}}+\frac{ib_{7}}{2\sqrt{2}}\subcgs{3_{1}}{3_{2}}{3_{2}}
\end{align*}
\efl 
\item $\tsprodx{6}{8}\to\rep{6}^{(1)}$: \bfl 
\begin{align*}
\rep{1}_{0} & =\subcgs{2}{2}{1_{0}}
\end{align*}
\begin{align*}
\rep{2} & =-\frac{1}{2\sqrt{2}}\subcgs{1_{0}}{2}{2}+\frac{1}{\sqrt{14}}\subcgs{2}{2}{2}\\
 & +\frac{3\sqrt{\frac{3}{7}}}{4}\subcgs{3_{1}}{3_{1}}{2}+\frac{3}{4}\subcgs{3_{1}}{3_{2}}{2}
\end{align*}
\begin{align*}
\rep{3}_{1} & =\frac{1}{2}\subcgs{1_{0}}{3_{1}}{3_{1}}-\sqrt{\frac{2}{7}}\subcgs{2}{3_{1}}{3_{1}}\\
 & -\frac{1}{2\sqrt{7}}\subcgs{3_{1}}{2}{3_{1}}+\frac{\sqrt{\frac{3}{14}}}{2}\subcgs{3_{1}}{3_{1}}{3_{1}}\\
 & -\frac{\sqrt{\frac{3}{2}}}{2}\subcgs{3_{1}}{3_{2}}{3_{1}}
\end{align*}
\efl 
\item $\tsprodx{6}{8}\to\rep{6}^{(2)}$: \bfl 
\begin{align*}
\rep{1}_{0} & =\subcgs{3_{1}}{3_{1}}{1_{0}}
\end{align*}
\begin{align*}
\rep{2} & =\frac{\sqrt{\frac{3}{2}}}{2}\subcgs{1_{0}}{2}{2}+\sqrt{\frac{3}{14}}\subcgs{2}{2}{2}\\
 & -\frac{5}{4\sqrt{7}}\subcgs{3_{1}}{3_{1}}{2}+\frac{\sqrt{3}}{4}\subcgs{3_{1}}{3_{2}}{2}
\end{align*}
\begin{align*}
\rep{3}_{1} & =\frac{1}{2\sqrt{3}}\subcgs{1_{0}}{3_{1}}{3_{1}}+\frac{1}{\sqrt{42}}\subcgs{2}{3_{1}}{3_{1}}\\
 & -\frac{1}{\sqrt{2}}\subcgs{2}{3_{2}}{3_{1}}-\frac{\sqrt{\frac{3}{7}}}{2}\subcgs{3_{1}}{2}{3_{1}}\\
 & +\frac{3}{2\sqrt{14}}\subcgs{3_{1}}{3_{1}}{3_{1}}+\frac{1}{2\sqrt{2}}\subcgs{3_{1}}{3_{2}}{3_{1}}
\end{align*}
\efl 
\item $\tsprodx{6}{8}\to\rep{7}^{(1)}$: \bfl 
\begin{align*}
\rep{1}_{1} & =\subcgs{2}{2}{1_{1}}
\end{align*}
\begin{align*}
\rep{3}_{1} & =\frac{\sqrt{7}}{12}\subcgs{1_{0}}{3_{1}}{3_{1}}+\frac{1}{6\sqrt{2}}\subcgs{2}{3_{1}}{3_{1}}\\
 & -\frac{\sqrt{\frac{7}{6}}}{2}\subcgs{2}{3_{2}}{3_{1}}+\frac{1}{4}\subcgs{3_{1}}{2}{3_{1}}\\
 & -\frac{7}{4\sqrt{6}}\subcgs{3_{1}}{3_{1}}{3_{1}}-\frac{\sqrt{\frac{7}{6}}}{4}\subcgs{3_{1}}{3_{2}}{3_{1}}
\end{align*}
\begin{align*}
\rep{3}_{2} & =-\frac{7}{12}\subcgs{1_{0}}{3_{2}}{3_{2}}+\frac{1}{2\sqrt{6}}\subcgs{2}{3_{1}}{3_{2}}\\
 & -\frac{\sqrt{\frac{7}{2}}}{6}\subcgs{2}{3_{2}}{3_{2}}+\frac{\sqrt{3}}{4}\subcgs{3_{1}}{2}{3_{2}}\\
 & +\frac{5}{4\sqrt{6}}\subcgs{3_{1}}{3_{1}}{3_{2}}-\frac{\sqrt{\frac{7}{6}}}{4}\subcgs{3_{1}}{3_{2}}{3_{2}}
\end{align*}
\efl 
\item $\tsprodx{6}{8}\to\rep{7}^{(2)}$: \bfl 
\begin{align*}
\rep{1}_{1} & =\subcgs{3_{1}}{3_{2}}{1_{1}}
\end{align*}
\begin{align*}
\rep{3}_{1} & =-\frac{7}{12}\subcgs{1_{0}}{3_{1}}{3_{1}}-\frac{\sqrt{\frac{7}{2}}}{6}\subcgs{2}{3_{1}}{3_{1}}\\
 & -\frac{1}{2\sqrt{6}}\subcgs{2}{3_{2}}{3_{1}}-\frac{\sqrt{7}}{4}\subcgs{3_{1}}{2}{3_{1}}\\
 & -\frac{\sqrt{\frac{7}{6}}}{4}\subcgs{3_{1}}{3_{1}}{3_{1}}-\frac{1}{4\sqrt{6}}\subcgs{3_{1}}{3_{2}}{3_{1}}
\end{align*}
\begin{align*}
\rep{3}_{2} & =-\frac{\sqrt{7}}{12}\subcgs{1_{0}}{3_{2}}{3_{2}}+\frac{\sqrt{\frac{7}{6}}}{2}\subcgs{2}{3_{1}}{3_{2}}\\
 & +\frac{5}{6\sqrt{2}}\subcgs{2}{3_{2}}{3_{2}}+\frac{\sqrt{\frac{7}{3}}}{4}\subcgs{3_{1}}{2}{3_{2}}\\
 & -\frac{\sqrt{\frac{7}{6}}}{4}\subcgs{3_{1}}{3_{1}}{3_{2}}+\frac{\sqrt{\frac{3}{2}}}{4}\subcgs{3_{1}}{3_{2}}{3_{2}}
\end{align*}
\efl 
\item $\tsprodx{6}{8}\to\rep{8}^{(1)}$: \bfl 
\begin{align*}
\rep{2} & =\frac{1}{2}\subcgs{1_{0}}{2}{2}-\frac{2}{\sqrt{7}}\subcgs{2}{2}{2}\\
 & +\frac{\sqrt{\frac{3}{14}}}{2}\subcgs{3_{1}}{3_{1}}{2}+\frac{1}{2\sqrt{2}}\subcgs{3_{1}}{3_{2}}{2}
\end{align*}
\begin{align*}
\rep{3}_{1} & =\frac{1}{6}\subcgs{1_{0}}{3_{1}}{3_{1}}-\frac{4\sqrt{\frac{2}{7}}}{3}\subcgs{2}{3_{1}}{3_{1}}\\
 & +\frac{1}{2\sqrt{7}}\subcgs{3_{1}}{2}{3_{1}}-\frac{\sqrt{\frac{3}{14}}}{2}\subcgs{3_{1}}{3_{1}}{3_{1}}\\
 & +\frac{\sqrt{\frac{3}{2}}}{2}\subcgs{3_{1}}{3_{2}}{3_{1}}
\end{align*}
\begin{align*}
\rep{3}_{2} & =-\frac{1}{2}\subcgs{1_{0}}{3_{2}}{3_{2}}-\frac{1}{2\sqrt{3}}\subcgs{3_{1}}{2}{3_{2}}\\
 & -\frac{\sqrt{\frac{3}{2}}}{2}\subcgs{3_{1}}{3_{1}}{3_{2}}-\frac{\sqrt{\frac{7}{6}}}{2}\subcgs{3_{1}}{3_{2}}{3_{2}}
\end{align*}
\efl 
\item $\tsprodx{6}{8}\to\rep{8}^{(2)}$: \bfl 
\begin{align*}
\rep{2} & =-\frac{1}{2}\subcgs{1_{0}}{2}{2}-\frac{1}{\sqrt{7}}\subcgs{2}{2}{2}\\
 & -\frac{3\sqrt{\frac{3}{14}}}{2}\subcgs{3_{1}}{3_{1}}{2}+\frac{1}{2\sqrt{2}}\subcgs{3_{1}}{3_{2}}{2}
\end{align*}
\begin{align*}
\rep{3}_{1} & =\frac{1}{2}\subcgs{1_{0}}{3_{1}}{3_{1}}+\frac{1}{\sqrt{14}}\subcgs{2}{3_{1}}{3_{1}}\\
 & +\frac{1}{\sqrt{6}}\subcgs{2}{3_{2}}{3_{1}}-\frac{3}{2\sqrt{7}}\subcgs{3_{1}}{2}{3_{1}}\\
 & -\frac{5}{2\sqrt{42}}\subcgs{3_{1}}{3_{1}}{3_{1}}+\frac{1}{2\sqrt{6}}\subcgs{3_{1}}{3_{2}}{3_{1}}
\end{align*}
\begin{align*}
\rep{3}_{2} & =-\frac{1}{6}\subcgs{1_{0}}{3_{2}}{3_{2}}+\frac{1}{\sqrt{6}}\subcgs{2}{3_{1}}{3_{2}}\\
 & -\frac{\sqrt{\frac{7}{2}}}{3}\subcgs{2}{3_{2}}{3_{2}}-\frac{1}{2\sqrt{3}}\subcgs{3_{1}}{2}{3_{2}}\\
 & -\frac{1}{2\sqrt{6}}\subcgs{3_{1}}{3_{1}}{3_{2}}+\frac{\sqrt{\frac{7}{6}}}{2}\subcgs{3_{1}}{3_{2}}{3_{2}}
\end{align*}
\efl 
\end{itemize}

\paragraph*{\cgEqFontsize $\tsprod{7}{7}\to\left(\rep{1}+\rep{6}^{(1)}+\rep{6}^{(2)}+\rep{7}+\rep{8}\right)_{s}+\left(\rep{3}+\rep{\bar{3}}+\rep{7}+\rep{8}\right)_{a}$}
\begin{itemize}
\item $\tsprodx{7}{7}\to\rep{1}_{s}$: \bfl 
\begin{align*}
\rep{1}_{0} & =\frac{1}{\sqrt{7}}\subcgs{1_{1}}{1_{1}}{1_{0}}+\sqrt{\frac{3}{7}}\subcgs{3_{1}}{3_{1}}{1_{0}}\\
 & +\sqrt{\frac{3}{7}}\subcgs{3_{2}}{3_{2}}{1_{0}}
\end{align*}
\efl 
\item $\tsprodx{7}{7}\to\rep{6}_{s}^{(1)}$: \bfl 
\begin{align*}
\rep{1}_{0} & =\frac{\sqrt{3}}{2}\subcgs{1_{1}}{1_{1}}{1_{0}}-\frac{1}{2}\subcgs{3_{2}}{3_{2}}{1_{0}}
\end{align*}
\begin{align*}
\rep{2} & =-\sqrt{\frac{3}{7}}\subcgs{3_{1}}{3_{2}}{2}-\sqrt{\frac{3}{7}}\subcgs{3_{2}}{3_{1}}{2}\\
 & +\frac{1}{\sqrt{7}}\subcgs{3_{2}}{3_{2}}{2}
\end{align*}
\begin{align*}
\rep{3}_{1} & =-\frac{1}{2\sqrt{7}}\left(\subcgs{1_{1}}{3_{2}}{3_{1}}+\subcgs{3_{2}}{1_{1}}{3_{1}}\right)\\
 & +\sqrt{\frac{2}{7}}\subcgs{3_{1}}{3_{1}}{3_{1}}+\sqrt{\frac{2}{7}}\subcgs{3_{1}}{3_{2}}{3_{1}}\\
 & -\sqrt{\frac{2}{7}}\subcgs{3_{2}}{3_{1}}{3_{1}}-\frac{1}{\sqrt{14}}\subcgs{3_{2}}{3_{2}}{3_{1}}
\end{align*}
\efl 
\item $\tsprodx{7}{7}\to\rep{6}_{s}^{(2)}$: \bfl 
\begin{align*}
\rep{1}_{0} & =-\frac{\sqrt{\frac{3}{7}}}{2}\subcgs{1_{1}}{1_{1}}{1_{0}}+\frac{2}{\sqrt{7}}\subcgs{3_{1}}{3_{1}}{1_{0}}\\
 & -\frac{3}{2\sqrt{7}}\subcgs{3_{2}}{3_{2}}{1_{0}}
\end{align*}
\begin{align*}
\rep{2} & =-\subcgs{3_{1}}{3_{1}}{2}
\end{align*}
\begin{align*}
\rep{3}_{1} & =-\frac{1}{2}\left(\subcgs{1_{1}}{3_{2}}{3_{1}}+\subcgs{3_{2}}{1_{1}}{3_{1}}\right)\\
 & +\frac{1}{\sqrt{2}}\subcgs{3_{2}}{3_{2}}{3_{1}}
\end{align*}
\efl 
\item $\tsprodx{7}{7}\to\rep{7}_{s}$: \bfl 
\begin{align*}
\rep{1}_{1} & =\frac{1}{\sqrt{2}}\subcgs{3_{1}}{3_{2}}{1_{1}}+\frac{1}{\sqrt{2}}\subcgs{3_{2}}{3_{1}}{1_{1}}
\end{align*}
\begin{align*}
\rep{3}_{1} & =\frac{1}{\sqrt{6}}\left(\subcgs{1_{1}}{3_{2}}{3_{1}}+\subcgs{3_{2}}{1_{1}}{3_{1}}\right)\\
 & +\frac{1}{\sqrt{3}}\subcgs{3_{1}}{3_{1}}{3_{1}}+\frac{1}{\sqrt{3}}\subcgs{3_{2}}{3_{2}}{3_{1}}
\end{align*}
\begin{align*}
\rep{3}_{2} & =\frac{1}{\sqrt{6}}\left(\subcgs{1_{1}}{3_{1}}{3_{2}}+\subcgs{3_{1}}{1_{1}}{3_{2}}\right)\\
 & +\frac{1}{\sqrt{3}}\subcgs{3_{1}}{3_{2}}{3_{2}}+\frac{1}{\sqrt{3}}\subcgs{3_{2}}{3_{1}}{3_{2}}
\end{align*}
\efl 
\item $\tsprodx{7}{7}\to\rep{8}_{s}$: \bfl 
\begin{align*}
\rep{2} & =\frac{1}{\sqrt{14}}\subcgs{3_{1}}{3_{2}}{2}+\frac{1}{\sqrt{14}}\subcgs{3_{2}}{3_{1}}{2}\\
 & +\sqrt{\frac{6}{7}}\subcgs{3_{2}}{3_{2}}{2}
\end{align*}
\begin{align*}
\rep{3}_{1} & =\frac{1}{\sqrt{21}}\left(\subcgs{1_{1}}{3_{2}}{3_{1}}+\subcgs{3_{2}}{1_{1}}{3_{1}}\right)\\
 & -2\sqrt{\frac{2}{21}}\subcgs{3_{1}}{3_{1}}{3_{1}}+\sqrt{\frac{3}{14}}\subcgs{3_{1}}{3_{2}}{3_{1}}\\
 & -\sqrt{\frac{3}{14}}\subcgs{3_{2}}{3_{1}}{3_{1}}+\sqrt{\frac{2}{21}}\subcgs{3_{2}}{3_{2}}{3_{1}}
\end{align*}
\begin{align*}
\rep{3}_{2} & =-\frac{1}{\sqrt{3}}\left(\subcgs{1_{1}}{3_{1}}{3_{2}}+\subcgs{3_{1}}{1_{1}}{3_{2}}\right)\\
 & +\frac{1}{\sqrt{6}}\subcgs{3_{1}}{3_{2}}{3_{2}}+\frac{1}{\sqrt{6}}\subcgs{3_{2}}{3_{1}}{3_{2}}
\end{align*}
\efl 
\item $\tsprodx{7}{7}\to\rep{3}_{a}$: \bfl 
\begin{align*}
\rep{3}_{2} & =\frac{1}{\sqrt{7}}\left(\subcgs{1_{1}}{3_{1}}{3_{2}}-\subcgs{3_{1}}{1_{1}}{3_{2}}\right)\\
 & +\frac{b_{7}}{\sqrt{7}}\subcgs{3_{1}}{3_{1}}{3_{2}}-\frac{b_{7}}{\sqrt{7}}\subcgs{3_{1}}{3_{2}}{3_{2}}\\
 & +\frac{b_{7}}{\sqrt{7}}\subcgs{3_{2}}{3_{1}}{3_{2}}-\sqrt{\frac{2}{7}}b_{7}^{2}\subcgs{3_{2}}{3_{2}}{3_{2}}
\end{align*}
\efl 
\item $\tsprodx{7}{7}\to\rep{\bar{3}}_{a}$: \bfl 
\begin{align*}
\rep{3}_{2} & =\frac{1}{\sqrt{7}}\left(\subcgs{1_{1}}{3_{1}}{3_{2}}-\subcgs{3_{1}}{1_{1}}{3_{2}}\right)\\
 & +\frac{1}{\sqrt{7}b_{7}}\subcgs{3_{1}}{3_{1}}{3_{2}}-\frac{1}{\sqrt{7}b_{7}}\subcgs{3_{1}}{3_{2}}{3_{2}}\\
 & +\frac{1}{\sqrt{7}b_{7}}\subcgs{3_{2}}{3_{1}}{3_{2}}-\frac{\sqrt{\frac{2}{7}}}{b_{7}^{2}}\subcgs{3_{2}}{3_{2}}{3_{2}}
\end{align*}
\efl 
\item $\tsprodx{7}{7}\to\rep{7}_{a}$: \bfl 
\begin{align*}
\rep{1}_{1} & =\frac{1}{\sqrt{2}}\subcgs{3_{1}}{3_{2}}{1_{1}}-\frac{1}{\sqrt{2}}\subcgs{3_{2}}{3_{1}}{1_{1}}
\end{align*}
\begin{align*}
\rep{3}_{1} & =-\frac{1}{\sqrt{6}}\left(\subcgs{1_{1}}{3_{2}}{3_{1}}-\subcgs{3_{2}}{1_{1}}{3_{1}}\right)\\
 & +\frac{1}{\sqrt{3}}\subcgs{3_{1}}{3_{2}}{3_{1}}+\frac{1}{\sqrt{3}}\subcgs{3_{2}}{3_{1}}{3_{1}}
\end{align*}
\begin{align*}
\rep{3}_{2} & =\frac{1}{\sqrt{6}}\left(\subcgs{1_{1}}{3_{1}}{3_{2}}-\subcgs{3_{1}}{1_{1}}{3_{2}}\right)\\
 & +\frac{1}{\sqrt{3}}\subcgs{3_{1}}{3_{1}}{3_{2}}-\frac{1}{\sqrt{3}}\subcgs{3_{2}}{3_{2}}{3_{2}}
\end{align*}
\efl 
\item $\tsprodx{7}{7}\to\rep{8}_{a}$: \bfl 
\begin{align*}
\rep{2} & =\frac{1}{\sqrt{2}}\subcgs{3_{1}}{3_{2}}{2}-\frac{1}{\sqrt{2}}\subcgs{3_{2}}{3_{1}}{2}
\end{align*}
\begin{align*}
\rep{3}_{1} & =-\frac{1}{\sqrt{3}}\left(\subcgs{1_{1}}{3_{2}}{3_{1}}-\subcgs{3_{2}}{1_{1}}{3_{1}}\right)\\
 & -\frac{1}{\sqrt{6}}\subcgs{3_{1}}{3_{2}}{3_{1}}-\frac{1}{\sqrt{6}}\subcgs{3_{2}}{3_{1}}{3_{1}}
\end{align*}
\begin{align*}
\rep{3}_{2} & =-\frac{1}{\sqrt{21}}\left(\subcgs{1_{1}}{3_{1}}{3_{2}}-\subcgs{3_{1}}{1_{1}}{3_{2}}\right)\\
 & +2\sqrt{\frac{2}{21}}\subcgs{3_{1}}{3_{1}}{3_{2}}+\sqrt{\frac{3}{14}}\subcgs{3_{1}}{3_{2}}{3_{2}}\\
 & -\sqrt{\frac{3}{14}}\subcgs{3_{2}}{3_{1}}{3_{2}}+\sqrt{\frac{2}{21}}\subcgs{3_{2}}{3_{2}}{3_{2}}
\end{align*}
\efl 
\end{itemize}

\paragraph*{\cgEqFontsize $\tsprod{7}{8}\to\rep{3}+\rep{\bar{3}}+\rep{6}^{(1)}+\rep{6}^{(2)}+\rep{7}^{(1)}+\rep{7}^{(2)}+\rep{8}^{(1)}+\rep{8}^{(2)}+\rep{8}^{(3)}$}
\begin{itemize}
\item $\tsprodx{7}{8}\to\rep{3}$: \bfl 
\begin{align*}
\rep{3}_{2} & =\frac{1}{\sqrt{7}}\subcgs{1_{1}}{3_{1}}{3_{2}}-\frac{1}{\sqrt{7}b_{7}^{4}}\subcgs{3_{1}}{2}{3_{2}}\\
 & +\frac{1}{2\sqrt{7}b_{7}^{3}}\subcgs{3_{1}}{3_{1}}{3_{2}}+\frac{1}{2b_{7}^{3}}\subcgs{3_{1}}{3_{2}}{3_{2}}\\
 & +\frac{\sqrt{\frac{3}{7}}}{2b_{7}^{2}}\subcgs{3_{2}}{2}{3_{2}}-\frac{\sqrt{7}+35i}{56\sqrt{2}}\subcgs{3_{2}}{3_{1}}{3_{2}}\\
 & -\frac{1}{2\sqrt{2}b_{7}^{2}}\subcgs{3_{2}}{3_{2}}{3_{2}}
\end{align*}
\efl 
\item $\tsprodx{7}{8}\to\rep{\bar{3}}$: \bfl 
\begin{align*}
\rep{3}_{2} & =\frac{1}{\sqrt{7}}\subcgs{1_{1}}{3_{1}}{3_{2}}-\frac{b_{7}^{4}}{\sqrt{7}}\subcgs{3_{1}}{2}{3_{2}}\\
 & +\frac{b_{7}^{3}}{2\sqrt{7}}\subcgs{3_{1}}{3_{1}}{3_{2}}+\frac{b_{7}^{3}}{2}\subcgs{3_{1}}{3_{2}}{3_{2}}\\
 & +\frac{1}{2}\sqrt{\frac{3}{7}}b_{7}^{2}\subcgs{3_{2}}{2}{3_{2}}-\frac{\sqrt{7}-35i}{56\sqrt{2}}\subcgs{3_{2}}{3_{1}}{3_{2}}\\
 & -\frac{b_{7}^{2}}{2\sqrt{2}}\subcgs{3_{2}}{3_{2}}{3_{2}}
\end{align*}
\efl 
\item $\tsprodx{7}{8}\to\rep{6}^{(1)}$: \bfl 
\begin{align*}
\rep{1}_{0} & =\subcgs{3_{1}}{3_{1}}{1_{0}}
\end{align*}
\begin{align*}
\rep{2} & =-\frac{\sqrt{\frac{3}{14}}}{2}\subcgs{1_{1}}{2}{2}+\frac{1}{\sqrt{7}}\subcgs{3_{1}}{3_{1}}{2}\\
 & +\frac{3\sqrt{\frac{3}{7}}}{4}\subcgs{3_{2}}{3_{1}}{2}-\frac{3}{4}\subcgs{3_{2}}{3_{2}}{2}
\end{align*}
\begin{align*}
\rep{3}_{1} & =-\frac{1}{2}\subcgs{1_{1}}{3_{2}}{3_{1}}+\sqrt{\frac{3}{7}}\subcgs{3_{1}}{2}{3_{1}}\\
 & -\frac{1}{2\sqrt{7}}\subcgs{3_{2}}{2}{3_{1}}-\frac{3}{2\sqrt{14}}\subcgs{3_{2}}{3_{1}}{3_{1}}\\
 & -\frac{1}{2\sqrt{2}}\subcgs{3_{2}}{3_{2}}{3_{1}}
\end{align*}
\efl 
\item $\tsprodx{7}{8}\to\rep{6}^{(2)}$: \bfl 
\begin{align*}
\rep{1}_{0} & =\subcgs{3_{2}}{3_{2}}{1_{0}}
\end{align*}
\begin{align*}
\rep{2} & =\frac{\sqrt{\frac{3}{2}}}{2}\subcgs{1_{1}}{2}{2}-\sqrt{\frac{3}{7}}\subcgs{3_{1}}{3_{2}}{2}\\
 & +\frac{\sqrt{3}}{4}\subcgs{3_{2}}{3_{1}}{2}+\frac{1}{4\sqrt{7}}\subcgs{3_{2}}{3_{2}}{2}
\end{align*}
\begin{align*}
\rep{3}_{1} & =-\frac{1}{2\sqrt{7}}\subcgs{1_{1}}{3_{2}}{3_{1}}+\frac{1}{\sqrt{2}}\subcgs{3_{1}}{3_{1}}{3_{1}}\\
 & -\frac{1}{\sqrt{14}}\subcgs{3_{1}}{3_{2}}{3_{1}}-\frac{1}{2}\subcgs{3_{2}}{2}{3_{1}}\\
 & +\frac{1}{2\sqrt{2}}\subcgs{3_{2}}{3_{1}}{3_{1}}+\frac{1}{2\sqrt{14}}\subcgs{3_{2}}{3_{2}}{3_{1}}
\end{align*}
\efl 
\item $\tsprodx{7}{8}\to\rep{7}^{(1)}$: \bfl 
\begin{align*}
\rep{1}_{1} & =\subcgs{3_{1}}{3_{2}}{1_{1}}
\end{align*}
\begin{align*}
\rep{3}_{1} & =\frac{\sqrt{3}}{4}\subcgs{1_{1}}{3_{2}}{3_{1}}+\frac{\sqrt{\frac{7}{6}}}{2}\subcgs{3_{1}}{3_{1}}{3_{1}}\\
 & -\frac{1}{2\sqrt{6}}\subcgs{3_{1}}{3_{2}}{3_{1}}+\frac{\sqrt{\frac{7}{3}}}{4}\subcgs{3_{2}}{2}{3_{1}}\\
 & -\frac{\sqrt{\frac{7}{6}}}{4}\subcgs{3_{2}}{3_{1}}{3_{1}}-\frac{5}{4\sqrt{6}}\subcgs{3_{2}}{3_{2}}{3_{1}}
\end{align*}
\begin{align*}
\rep{3}_{2} & =-\frac{\sqrt{\frac{7}{3}}}{4}\subcgs{1_{1}}{3_{1}}{3_{2}}+\frac{\sqrt{\frac{7}{6}}}{2}\subcgs{3_{1}}{3_{1}}{3_{2}}\\
 & -\frac{1}{2\sqrt{6}}\subcgs{3_{1}}{3_{2}}{3_{2}}-\frac{\sqrt{7}}{4}\subcgs{3_{2}}{2}{3_{2}}\\
 & -\frac{\sqrt{\frac{7}{6}}}{4}\subcgs{3_{2}}{3_{1}}{3_{2}}-\frac{1}{4\sqrt{6}}\subcgs{3_{2}}{3_{2}}{3_{2}}
\end{align*}
\efl 
\item $\tsprodx{7}{8}\to\rep{7}^{(2)}$: \bfl 
\begin{align*}
\rep{1}_{1} & =\subcgs{3_{2}}{3_{1}}{1_{1}}
\end{align*}
\begin{align*}
\rep{3}_{1} & =-\frac{\sqrt{\frac{7}{3}}}{4}\subcgs{1_{1}}{3_{2}}{3_{1}}-\frac{1}{2\sqrt{6}}\subcgs{3_{1}}{3_{1}}{3_{1}}\\
 & -\frac{\sqrt{\frac{7}{6}}}{2}\subcgs{3_{1}}{3_{2}}{3_{1}}+\frac{\sqrt{3}}{4}\subcgs{3_{2}}{2}{3_{1}}\\
 & +\frac{5}{4\sqrt{6}}\subcgs{3_{2}}{3_{1}}{3_{1}}-\frac{\sqrt{\frac{7}{6}}}{4}\subcgs{3_{2}}{3_{2}}{3_{1}}
\end{align*}
\begin{align*}
\rep{3}_{2} & =-\frac{\sqrt{3}}{4}\subcgs{1_{1}}{3_{1}}{3_{2}}-\frac{1}{\sqrt{3}}\subcgs{3_{1}}{2}{3_{2}}\\
 & +\frac{1}{2\sqrt{6}}\subcgs{3_{1}}{3_{1}}{3_{2}}+\frac{\sqrt{\frac{7}{6}}}{2}\subcgs{3_{1}}{3_{2}}{3_{2}}\\
 & +\frac{1}{4}\subcgs{3_{2}}{2}{3_{2}}+\frac{1}{4\sqrt{6}}\subcgs{3_{2}}{3_{1}}{3_{2}}\\
 & -\frac{\sqrt{\frac{7}{6}}}{4}\subcgs{3_{2}}{3_{2}}{3_{2}}
\end{align*}
\efl 
\item $\tsprodx{7}{8}\to\rep{8}^{(1)}$: \bfl 
\begin{align*}
\rep{2} & =\frac{3}{4\sqrt{2}}\subcgs{1_{1}}{2}{2}-\frac{5}{8}\subcgs{3_{2}}{3_{1}}{2}\\
 & -\frac{\sqrt{21}}{8}\subcgs{3_{2}}{3_{2}}{2}
\end{align*}
\begin{align*}
\rep{3}_{1} & =-\frac{\sqrt{\frac{7}{6}}}{4}\subcgs{1_{1}}{3_{2}}{3_{1}}+\frac{1}{2\sqrt{3}}\subcgs{3_{1}}{3_{1}}{3_{1}}\\
 & +\frac{\sqrt{\frac{7}{3}}}{2}\subcgs{3_{1}}{3_{2}}{3_{1}}+\frac{\sqrt{\frac{3}{2}}}{4}\subcgs{3_{2}}{2}{3_{1}}\\
 & +\frac{5}{8\sqrt{3}}\subcgs{3_{2}}{3_{1}}{3_{1}}-\frac{\sqrt{\frac{7}{3}}}{8}\subcgs{3_{2}}{3_{2}}{3_{1}}
\end{align*}
\begin{align*}
\rep{3}_{2} & =-\frac{3\sqrt{\frac{3}{14}}}{4}\subcgs{1_{1}}{3_{1}}{3_{2}}+2\sqrt{\frac{2}{21}}\subcgs{3_{1}}{2}{3_{2}}\\
 & +\frac{1}{2\sqrt{21}}\subcgs{3_{1}}{3_{1}}{3_{2}}+\frac{1}{2\sqrt{3}}\subcgs{3_{1}}{3_{2}}{3_{2}}\\
 & -\frac{1}{4\sqrt{14}}\subcgs{3_{2}}{2}{3_{2}}+\frac{23}{8\sqrt{21}}\subcgs{3_{2}}{3_{1}}{3_{2}}\\
 & +\frac{1}{8\sqrt{3}}\subcgs{3_{2}}{3_{2}}{3_{2}}
\end{align*}
\efl 
\item $\tsprodx{7}{8}\to\rep{8}^{(2)}$: \bfl 
\begin{align*}
\rep{2} & =\frac{1}{4\sqrt{7}}\subcgs{1_{1}}{2}{2}+\sqrt{\frac{6}{7}}\subcgs{3_{1}}{3_{1}}{2}\\
 & -\frac{3}{4\sqrt{14}}\subcgs{3_{2}}{3_{1}}{2}+\frac{\sqrt{\frac{3}{2}}}{4}\subcgs{3_{2}}{3_{2}}{2}
\end{align*}
\begin{align*}
\rep{3}_{1} & =-\frac{\sqrt{3}}{4}\subcgs{1_{1}}{3_{2}}{3_{1}}-\frac{2}{\sqrt{7}}\subcgs{3_{1}}{2}{3_{1}}\\
 & -\frac{\sqrt{\frac{3}{7}}}{4}\subcgs{3_{2}}{2}{3_{1}}-\frac{3\sqrt{\frac{3}{14}}}{4}\subcgs{3_{2}}{3_{1}}{3_{1}}\\
 & -\frac{\sqrt{\frac{3}{2}}}{4}\subcgs{3_{2}}{3_{2}}{3_{1}}
\end{align*}
\begin{align*}
\rep{3}_{2} & =\frac{\sqrt{3}}{4}\subcgs{1_{1}}{3_{1}}{3_{2}}-\frac{1}{4}\subcgs{3_{2}}{2}{3_{2}}\\
 & +\frac{\sqrt{\frac{3}{2}}}{4}\subcgs{3_{2}}{3_{1}}{3_{2}}-\frac{\sqrt{\frac{21}{2}}}{4}\subcgs{3_{2}}{3_{2}}{3_{2}}
\end{align*}
\efl 
\item $\tsprodx{7}{8}\to\rep{8}^{(3)}$: \bfl 
\begin{align*}
\rep{2} & =\frac{3}{4\sqrt{2}}\subcgs{1_{1}}{2}{2}+\frac{2}{\sqrt{7}}\subcgs{3_{1}}{3_{2}}{2}\\
 & +\frac{3}{8}\subcgs{3_{2}}{3_{1}}{2}+\frac{\sqrt{\frac{3}{7}}}{8}\subcgs{3_{2}}{3_{2}}{2}
\end{align*}
\begin{align*}
\rep{3}_{1} & =\frac{3\sqrt{\frac{3}{14}}}{4}\subcgs{1_{1}}{3_{2}}{3_{1}}-\frac{1}{2\sqrt{3}}\subcgs{3_{1}}{3_{1}}{3_{1}}\\
 & +\frac{1}{2\sqrt{21}}\subcgs{3_{1}}{3_{2}}{3_{1}}-\frac{5}{4\sqrt{6}}\subcgs{3_{2}}{2}{3_{1}}\\
 & +\frac{5}{8\sqrt{3}}\subcgs{3_{2}}{3_{1}}{3_{1}}-\frac{23}{8\sqrt{21}}\subcgs{3_{2}}{3_{2}}{3_{1}}
\end{align*}
\begin{align*}
\rep{3}_{2} & =\frac{\sqrt{\frac{7}{6}}}{4}\subcgs{1_{1}}{3_{1}}{3_{2}}+\frac{\sqrt{\frac{7}{3}}}{2}\subcgs{3_{1}}{3_{1}}{3_{2}}\\
 & -\frac{1}{2\sqrt{3}}\subcgs{3_{1}}{3_{2}}{3_{2}}+\frac{\sqrt{\frac{7}{2}}}{4}\subcgs{3_{2}}{2}{3_{2}}\\
 & +\frac{\sqrt{\frac{7}{3}}}{8}\subcgs{3_{2}}{3_{1}}{3_{2}}+\frac{1}{8\sqrt{3}}\subcgs{3_{2}}{3_{2}}{3_{2}}
\end{align*}
\efl 
\end{itemize}

\paragraph*{\cgEqFontsize $\tsprod{8}{8}\to\left(\rep{1}+\rep{6}^{(1)}+\rep{6}^{(2)}+\rep{7}+\rep{8}^{(1)}+\rep{8}^{(2)}\right)_{s}+\left(\rep{3}+\rep{\bar{3}}+\rep{7}^{(1)}+\rep{7}^{(2)}+\rep{8}\right)_{a}$}
\begin{itemize}
\item $\tsprodx{8}{8}\to\rep{1}_{s}$: \bfl 
\begin{align*}
\rep{1}_{0} & =\frac{1}{2}\subcgs{2}{2}{1_{0}}+\frac{\sqrt{\frac{3}{2}}}{2}\subcgs{3_{1}}{3_{1}}{1_{0}}\\
 & +\frac{\sqrt{\frac{3}{2}}}{2}\subcgs{3_{2}}{3_{2}}{1_{0}}
\end{align*}
\efl 
\item $\tsprodx{8}{8}\to\rep{6}_{s}^{(1)}$: \bfl 
\begin{align*}
\rep{1}_{0} & =\sqrt{\frac{3}{5}}\subcgs{2}{2}{1_{0}}-\sqrt{\frac{2}{5}}\subcgs{3_{2}}{3_{2}}{1_{0}}
\end{align*}
\begin{align*}
\rep{2} & =-\frac{\sqrt{\frac{15}{14}}}{2}\subcgs{2}{2}{2}-\frac{27}{4\sqrt{70}}\subcgs{3_{1}}{3_{1}}{2}\\
 & +\frac{\sqrt{\frac{3}{10}}}{4}\subcgs{3_{1}}{3_{2}}{2}+\frac{\sqrt{\frac{3}{10}}}{4}\subcgs{3_{2}}{3_{1}}{2}\\
 & +\frac{\sqrt{\frac{7}{10}}}{4}\subcgs{3_{2}}{3_{2}}{2}
\end{align*}
\begin{align*}
\rep{3}_{1} & =\frac{3\sqrt{\frac{3}{70}}}{2}\left(\subcgs{2}{3_{1}}{3_{1}}+\subcgs{3_{1}}{2}{3_{1}}\right)\\
 & -\frac{1}{2\sqrt{10}}\left(\subcgs{2}{3_{2}}{3_{1}}+\subcgs{3_{2}}{2}{3_{1}}\right)\\
 & -\frac{1}{2\sqrt{35}}\subcgs{3_{1}}{3_{1}}{3_{1}}-\frac{1}{\sqrt{5}}\subcgs{3_{1}}{3_{2}}{3_{1}}\\
 & +\frac{1}{\sqrt{5}}\subcgs{3_{2}}{3_{1}}{3_{1}}-\frac{\sqrt{\frac{7}{5}}}{2}\subcgs{3_{2}}{3_{2}}{3_{1}}
\end{align*}
\efl 
\item $\tsprodx{8}{8}\to\rep{6}_{s}^{(2)}$: \bfl 
\begin{align*}
\rep{1}_{0} & =-\frac{\sqrt{\frac{3}{5}}}{2}\subcgs{2}{2}{1_{0}}+\frac{\sqrt{\frac{5}{2}}}{2}\subcgs{3_{1}}{3_{1}}{1_{0}}\\
 & -\frac{3}{2\sqrt{10}}\subcgs{3_{2}}{3_{2}}{1_{0}}
\end{align*}
\begin{align*}
\rep{2} & =-\frac{\sqrt{\frac{15}{14}}}{2}\subcgs{2}{2}{2}+\frac{1}{4\sqrt{70}}\subcgs{3_{1}}{3_{1}}{2}\\
 & -\frac{3\sqrt{\frac{3}{10}}}{4}\subcgs{3_{1}}{3_{2}}{2}-\frac{3\sqrt{\frac{3}{10}}}{4}\subcgs{3_{2}}{3_{1}}{2}\\
 & -\frac{3\sqrt{\frac{7}{10}}}{4}\subcgs{3_{2}}{3_{2}}{2}
\end{align*}
\begin{align*}
\rep{3}_{1} & =-\sqrt{\frac{6}{35}}\left(\subcgs{2}{3_{1}}{3_{1}}+\subcgs{3_{1}}{2}{3_{1}}\right)\\
 & -\frac{1}{\sqrt{10}}\left(\subcgs{2}{3_{2}}{3_{1}}+\subcgs{3_{2}}{2}{3_{1}}\right)\\
 & -\frac{9}{4\sqrt{35}}\subcgs{3_{1}}{3_{1}}{3_{1}}-\frac{3}{4\sqrt{5}}\subcgs{3_{1}}{3_{2}}{3_{1}}\\
 & +\frac{3}{4\sqrt{5}}\subcgs{3_{2}}{3_{1}}{3_{1}}+\frac{\sqrt{\frac{7}{5}}}{4}\subcgs{3_{2}}{3_{2}}{3_{1}}
\end{align*}
\efl 
\item $\tsprodx{8}{8}\to\rep{7}_{s}$: \bfl 
\begin{align*}
\rep{1}_{1} & =\frac{1}{\sqrt{2}}\subcgs{3_{1}}{3_{2}}{1_{1}}+\frac{1}{\sqrt{2}}\subcgs{3_{2}}{3_{1}}{1_{1}}
\end{align*}
\begin{align*}
\rep{3}_{1} & =-\frac{1}{\sqrt{6}}\left(\subcgs{2}{3_{2}}{3_{1}}+\subcgs{3_{2}}{2}{3_{1}}\right)\\
 & -\frac{\sqrt{\frac{7}{3}}}{4}\subcgs{3_{1}}{3_{1}}{3_{1}}+\frac{\sqrt{3}}{4}\subcgs{3_{1}}{3_{2}}{3_{1}}\\
 & -\frac{\sqrt{3}}{4}\subcgs{3_{2}}{3_{1}}{3_{1}}-\frac{\sqrt{\frac{7}{3}}}{4}\subcgs{3_{2}}{3_{2}}{3_{1}}
\end{align*}
\begin{align*}
\rep{3}_{2} & =-\frac{\sqrt{\frac{7}{6}}}{2}\left(\subcgs{2}{3_{1}}{3_{2}}+\subcgs{3_{1}}{2}{3_{2}}\right)\\
 & +\frac{1}{2\sqrt{2}}\left(\subcgs{2}{3_{2}}{3_{2}}+\subcgs{3_{2}}{2}{3_{2}}\right)\\
 & -\frac{1}{2\sqrt{3}}\subcgs{3_{1}}{3_{2}}{3_{2}}-\frac{1}{2\sqrt{3}}\subcgs{3_{2}}{3_{1}}{3_{2}}
\end{align*}
\efl 
\item $\tsprodx{8}{8}\to\rep{8}_{s}^{(1)}$: \bfl 
\begin{align*}
\rep{2} & =\sqrt{\frac{2}{5}}\subcgs{2}{2}{2}-\sqrt{\frac{3}{10}}\subcgs{3_{1}}{3_{1}}{2}\\
 & -\sqrt{\frac{3}{10}}\subcgs{3_{2}}{3_{2}}{2}
\end{align*}
\begin{align*}
\rep{3}_{1} & =-\frac{1}{\sqrt{5}}\left(\subcgs{2}{3_{1}}{3_{1}}+\subcgs{3_{1}}{2}{3_{1}}\right)\\
 & +\sqrt{\frac{3}{10}}\subcgs{3_{1}}{3_{1}}{3_{1}}-\sqrt{\frac{3}{10}}\subcgs{3_{2}}{3_{2}}{3_{1}}
\end{align*}
\begin{align*}
\rep{3}_{2} & =-\frac{1}{\sqrt{5}}\left(\subcgs{2}{3_{2}}{3_{2}}+\subcgs{3_{2}}{2}{3_{2}}\right)\\
 & -\sqrt{\frac{3}{10}}\subcgs{3_{1}}{3_{2}}{3_{2}}-\sqrt{\frac{3}{10}}\subcgs{3_{2}}{3_{1}}{3_{2}}
\end{align*}
\efl 
\item $\tsprodx{8}{8}\to\rep{8}_{s}^{(2)}$: \bfl 
\begin{align*}
\rep{2} & =-\frac{3}{2\sqrt{35}}\subcgs{2}{2}{2}+\frac{3\sqrt{\frac{3}{35}}}{4}\subcgs{3_{1}}{3_{1}}{2}\\
 & +\frac{\sqrt{5}}{4}\subcgs{3_{1}}{3_{2}}{2}+\frac{\sqrt{5}}{4}\subcgs{3_{2}}{3_{1}}{2}\\
 & -\frac{\sqrt{\frac{21}{5}}}{4}\subcgs{3_{2}}{3_{2}}{2}
\end{align*}
\begin{align*}
\rep{3}_{1} & =\frac{3}{2\sqrt{70}}\left(\subcgs{2}{3_{1}}{3_{1}}+\subcgs{3_{1}}{2}{3_{1}}\right)\\
 & -\frac{\sqrt{\frac{5}{6}}}{2}\left(\subcgs{2}{3_{2}}{3_{1}}+\subcgs{3_{2}}{2}{3_{1}}\right)\\
 & +\frac{13}{2\sqrt{105}}\subcgs{3_{1}}{3_{1}}{3_{1}}+\frac{\sqrt{\frac{7}{15}}}{2}\subcgs{3_{2}}{3_{2}}{3_{1}}
\end{align*}
\begin{align*}
\rep{3}_{2} & =-\frac{\sqrt{\frac{5}{6}}}{2}\left(\subcgs{2}{3_{1}}{3_{2}}+\subcgs{3_{1}}{2}{3_{2}}\right)\\
 & -\frac{\sqrt{\frac{7}{10}}}{2}\left(\subcgs{2}{3_{2}}{3_{2}}+\subcgs{3_{2}}{2}{3_{2}}\right)\\
 & +\frac{\sqrt{\frac{7}{15}}}{2}\subcgs{3_{1}}{3_{2}}{3_{2}}+\frac{\sqrt{\frac{7}{15}}}{2}\subcgs{3_{2}}{3_{1}}{3_{2}}
\end{align*}
\efl 
\item $\tsprodx{8}{8}\to\rep{3}_{a}$: \bfl 
\begin{align*}
\rep{3}_{2} & =\frac{1}{4}\left(\subcgs{2}{3_{1}}{3_{2}}-\subcgs{3_{1}}{2}{3_{2}}\right)\\
 & +\frac{i\sqrt{3}}{4}\left(\subcgs{2}{3_{2}}{3_{2}}-\subcgs{3_{2}}{2}{3_{2}}\right)\\
 & -\frac{1}{2b_{7}^{3}}\subcgs{3_{1}}{3_{1}}{3_{2}}+\frac{i}{4b_{7}}\subcgs{3_{1}}{3_{2}}{3_{2}}\\
 & -\frac{i}{4b_{7}}\subcgs{3_{2}}{3_{1}}{3_{2}}-\frac{1}{2\sqrt{2}b_{7}^{2}}\subcgs{3_{2}}{3_{2}}{3_{2}}
\end{align*}
\efl 
\item $\tsprodx{8}{8}\to\rep{\bar{3}}_{a}$: \bfl 
\begin{align*}
\rep{3}_{2} & =\frac{1}{4}\left(\subcgs{2}{3_{1}}{3_{2}}-\subcgs{3_{1}}{2}{3_{2}}\right)\\
 & -\frac{i\sqrt{3}}{4}\left(\subcgs{2}{3_{2}}{3_{2}}-\subcgs{3_{2}}{2}{3_{2}}\right)\\
 & -\frac{b_{7}^{3}}{2}\subcgs{3_{1}}{3_{1}}{3_{2}}-\frac{ib_{7}}{4}\subcgs{3_{1}}{3_{2}}{3_{2}}\\
 & +\frac{ib_{7}}{4}\subcgs{3_{2}}{3_{1}}{3_{2}}-\frac{b_{7}^{2}}{2\sqrt{2}}\subcgs{3_{2}}{3_{2}}{3_{2}}
\end{align*}
\efl 
\item $\tsprodx{8}{8}\to\rep{7}_{a}^{(1)}$: \bfl 
\begin{align*}
\rep{1}_{1} & =\subcgs{2}{2}{1_{1}}
\end{align*}
\begin{align*}
\rep{3}_{1} & =-\frac{1}{8\sqrt{2}}\left(\subcgs{2}{3_{1}}{3_{1}}-\subcgs{3_{1}}{2}{3_{1}}\right)\\
 & +\frac{\sqrt{\frac{21}{2}}}{8}\left(\subcgs{2}{3_{2}}{3_{1}}-\subcgs{3_{2}}{2}{3_{1}}\right)\\
 & +\frac{\sqrt{21}}{8}\subcgs{3_{1}}{3_{2}}{3_{1}}+\frac{\sqrt{21}}{8}\subcgs{3_{2}}{3_{1}}{3_{1}}
\end{align*}
\begin{align*}
\rep{3}_{2} & =-\frac{\sqrt{\frac{3}{2}}}{8}\left(\subcgs{2}{3_{1}}{3_{2}}-\subcgs{3_{1}}{2}{3_{2}}\right)\\
 & +\frac{\sqrt{\frac{7}{2}}}{8}\left(\subcgs{2}{3_{2}}{3_{2}}-\subcgs{3_{2}}{2}{3_{2}}\right)\\
 & +\frac{\sqrt{3}}{4}\subcgs{3_{1}}{3_{1}}{3_{2}}+\frac{\sqrt{21}}{8}\subcgs{3_{1}}{3_{2}}{3_{2}}\\
 & -\frac{\sqrt{21}}{8}\subcgs{3_{2}}{3_{1}}{3_{2}}
\end{align*}
\efl 
\item $\tsprodx{8}{8}\to\rep{7}_{a}^{(2)}$: \bfl 
\begin{align*}
\rep{1}_{1} & =\frac{1}{\sqrt{2}}\subcgs{3_{1}}{3_{2}}{1_{1}}-\frac{1}{\sqrt{2}}\subcgs{3_{2}}{3_{1}}{1_{1}}
\end{align*}
\begin{align*}
\rep{3}_{1} & =-\frac{3\sqrt{\frac{7}{2}}}{8}\left(\subcgs{2}{3_{1}}{3_{1}}-\subcgs{3_{1}}{2}{3_{1}}\right)\\
 & -\frac{1}{8\sqrt{6}}\left(\subcgs{2}{3_{2}}{3_{1}}-\subcgs{3_{2}}{2}{3_{1}}\right)\\
 & -\frac{1}{8\sqrt{3}}\subcgs{3_{1}}{3_{2}}{3_{1}}-\frac{1}{8\sqrt{3}}\subcgs{3_{2}}{3_{1}}{3_{1}}
\end{align*}
\begin{align*}
\rep{3}_{2} & =-\frac{\sqrt{\frac{7}{6}}}{8}\left(\subcgs{2}{3_{1}}{3_{2}}-\subcgs{3_{1}}{2}{3_{2}}\right)\\
 & -\frac{3}{8\sqrt{2}}\left(\subcgs{2}{3_{2}}{3_{2}}-\subcgs{3_{2}}{2}{3_{2}}\right)\\
 & -\frac{\sqrt{\frac{7}{3}}}{4}\subcgs{3_{1}}{3_{1}}{3_{2}}+\frac{\sqrt{3}}{8}\subcgs{3_{1}}{3_{2}}{3_{2}}\\
 & -\frac{\sqrt{3}}{8}\subcgs{3_{2}}{3_{1}}{3_{2}}-\frac{\sqrt{\frac{7}{3}}}{2}\subcgs{3_{2}}{3_{2}}{3_{2}}
\end{align*}
\efl 
\item $\tsprodx{8}{8}\to\rep{8}_{a}$: \bfl 
\begin{align*}
\rep{2} & =\frac{1}{\sqrt{2}}\subcgs{3_{1}}{3_{2}}{2}-\frac{1}{\sqrt{2}}\subcgs{3_{2}}{3_{1}}{2}
\end{align*}
\begin{align*}
\rep{3}_{1} & =\frac{1}{\sqrt{3}}\left(\subcgs{2}{3_{2}}{3_{1}}-\subcgs{3_{2}}{2}{3_{1}}\right)\\
 & -\frac{1}{\sqrt{6}}\subcgs{3_{1}}{3_{2}}{3_{1}}-\frac{1}{\sqrt{6}}\subcgs{3_{2}}{3_{1}}{3_{1}}
\end{align*}
\begin{align*}
\rep{3}_{2} & =-\frac{1}{\sqrt{3}}\left(\subcgs{2}{3_{1}}{3_{2}}-\subcgs{3_{1}}{2}{3_{2}}\right)\\
 & -\frac{1}{\sqrt{6}}\subcgs{3_{1}}{3_{1}}{3_{2}}+\frac{1}{\sqrt{6}}\subcgs{3_{2}}{3_{2}}{3_{2}}
\end{align*}
\efl 
\end{itemize}

\section{\label{sec:App_CGC_PSL27_T7}CG Coefficients of $\group$ in $\subgb$
Basis}

The notations used in this appendix are the same as those of Appendix
\ref{sec:App_CGC_PSL27_S4}. Two constant angles will be used in the
following results
\[
\alpha=\arctan\frac{\sqrt{3}}{2},\quad\beta=\arctan\left(3\sqrt{3}+2\sqrt{6}\right).
\]

\paragraph*{\cgEqFontsize $\tsprod{3}{3}\to\rep{6}_{s}+\rep{\bar{3}}_{a}$}
\begin{itemize}
\item $\tsprodx{3}{3}\to\rep{\bar{3}}_{a}$: \bfl 
\begin{align*}
\rep{\bar{3}} & =\subcgt{3}{3}{\bar{3}_{a}}
\end{align*}
\efl 
\item $\tsprodx{3}{3}\to\rep{6}_{s}$: \bfl 
\begin{align*}
\rep{3} & =\subcgt{3}{3}{3}
\end{align*}
\begin{align*}
\rep{\bar{3}} & =\subcgt{3}{3}{\bar{3}_{s}}
\end{align*}
\efl 
\end{itemize}

\paragraph*{\cgEqFontsize $\tsprod{3}{\bar{3}}\to\rep{1}+\rep{8}$}
\begin{itemize}
\item $\tsprodx{3}{\bar{3}}\to\rep{1}$: \bfl 
\begin{align*}
\rep{1} & =\subcgt{3}{\bar{3}}{1}
\end{align*}
\efl 
\item $\tsprodx{3}{\bar{3}}\to\rep{8}$: \bfl 
\begin{align*}
\rep{1^{\prime}} & =\subcgt{3}{\bar{3}}{1^{\prime}}
\end{align*}
\begin{align*}
\rep{\bar{1}^{\prime}} & =\subcgt{3}{\bar{3}}{\bar{1}^{\prime}}
\end{align*}
\begin{align*}
\rep{3} & =\subcgt{3}{\bar{3}}{3}
\end{align*}
\begin{align*}
\rep{\bar{3}} & =\subcgt{3}{\bar{3}}{\bar{3}}
\end{align*}
\efl 
\end{itemize}

\paragraph*{\cgEqFontsize $\tsprod{3}{6}\to\rep{\bar{3}}+\rep{7}+\rep{8}$}
\begin{itemize}
\item $\tsprodx{3}{6}\to\rep{\bar{3}}$: \bfl 
\begin{align*}
\rep{\bar{3}} & =\frac{1}{\sqrt{2}}\subcgt{3}{3}{\bar{3}_{s}}+\frac{1}{\sqrt{2}}\subcgt{3}{\bar{3}}{\bar{3}}
\end{align*}
\efl 
\item $\tsprodx{3}{6}\to\rep{7}$: \bfl 
\begin{align*}
\rep{1} & =\subcgt{3}{\bar{3}}{1}
\end{align*}
\begin{align*}
\rep{3} & =\frac{1}{\sqrt{3}}\subcgt{3}{3}{3}+\sqrt{\frac{2}{3}}\subcgt{3}{\bar{3}}{3}
\end{align*}
\begin{align*}
\rep{\bar{3}} & =-\frac{1}{\sqrt{6}}\subcgt{3}{3}{\bar{3}_{s}}+\sqrt{\frac{2}{3}}\subcgt{3}{3}{\bar{3}_{a}}\\
 & +\frac{1}{\sqrt{6}}\subcgt{3}{\bar{3}}{\bar{3}}
\end{align*}
\efl 
\item $\tsprodx{3}{6}\to\rep{8}$: \bfl 
\begin{align*}
\rep{1^{\prime}} & =\subcgt{3}{\bar{3}}{1^{\prime}}
\end{align*}
\begin{align*}
\rep{\bar{1}^{\prime}} & =-\subcgt{3}{\bar{3}}{\bar{1}^{\prime}}
\end{align*}
\begin{align*}
\rep{3} & =i\sqrt{\frac{2}{3}}\subcgt{3}{3}{3}-\frac{i}{\sqrt{3}}\subcgt{3}{\bar{3}}{3}
\end{align*}
\begin{align*}
\rep{\bar{3}} & =-\frac{i}{\sqrt{3}}\subcgt{3}{3}{\bar{3}_{s}}-\frac{i}{\sqrt{3}}\subcgt{3}{3}{\bar{3}_{a}}\\
 & +\frac{i}{\sqrt{3}}\subcgt{3}{\bar{3}}{\bar{3}}
\end{align*}
\efl 
\end{itemize}

\paragraph*{\cgEqFontsize $\tsprod{3}{7}\to\rep{6}+\rep{7}+\rep{8}$}
\begin{itemize}
\item $\tsprodx{3}{7}\to\rep{6}$: \bfl 
\begin{align*}
\rep{3} & =\sqrt{\frac{2}{7}}\subcgt{3}{1}{3}+\frac{1}{\sqrt{7}}\subcgt{3}{3}{3}\\
 & +\frac{2}{\sqrt{7}}\subcgt{3}{\bar{3}}{3}
\end{align*}
\begin{align*}
\rep{\bar{3}} & =-\frac{1}{\sqrt{7}}\subcgt{3}{3}{\bar{3}_{s}}-\frac{2}{\sqrt{7}}\subcgt{3}{3}{\bar{3}_{a}}\\
 & +\sqrt{\frac{2}{7}}\subcgt{3}{\bar{3}}{\bar{3}}
\end{align*}
\efl 
\item $\tsprodx{3}{7}\to\rep{7}$: \bfl 
\begin{align*}
\rep{1} & =\subcgt{3}{\bar{3}}{1}
\end{align*}
\begin{align*}
\rep{3} & =-\frac{1}{\sqrt{3}}\subcgt{3}{1}{3}+\sqrt{\frac{2}{3}}\subcgt{3}{3}{3}
\end{align*}
\begin{align*}
\rep{\bar{3}} & =-\frac{1}{\sqrt{3}}\subcgt{3}{3}{\bar{3}_{a}}-\sqrt{\frac{2}{3}}\subcgt{3}{\bar{3}}{\bar{3}}
\end{align*}
\efl 
\item $\tsprodx{3}{7}\to\rep{8}$: \bfl 
\begin{align*}
\rep{1^{\prime}} & =\subcgt{3}{\bar{3}}{1^{\prime}}
\end{align*}
\begin{align*}
\rep{\bar{1}^{\prime}} & =-e^{2i\alpha}\subcgt{3}{\bar{3}}{\bar{1}^{\prime}}
\end{align*}
\begin{align*}
\rep{3} & =-2i\sqrt{\frac{2}{21}}e^{i\alpha}\subcgt{3}{1}{3}-\frac{2ie^{i\alpha}}{\sqrt{21}}\subcgt{3}{3}{3}\\
 & +i\sqrt{\frac{3}{7}}e^{i\alpha}\subcgt{3}{\bar{3}}{3}
\end{align*}
\begin{align*}
\rep{\bar{3}} & =-i\sqrt{\frac{6}{7}}e^{i\alpha}\subcgt{3}{3}{\bar{3}_{s}}+i\sqrt{\frac{2}{21}}e^{i\alpha}\subcgt{3}{3}{\bar{3}_{a}}\\
 & -\frac{ie^{i\alpha}}{\sqrt{21}}\subcgt{3}{\bar{3}}{\bar{3}}
\end{align*}
\efl 
\end{itemize}

\paragraph*{\cgEqFontsize $\tsprod{3}{8}\to\rep{3}+\rep{6}+\rep{7}+\rep{8}$}
\begin{itemize}
\item $\tsprodx{3}{8}\to\rep{3}$: \bfl 
\begin{align*}
\rep{3} & =\frac{1}{2\sqrt{2}}\subcgt{3}{1^{\prime}}{3}+\frac{1}{2\sqrt{2}}\subcgt{3}{\bar{1}^{\prime}}{3}\\
 & +\frac{\sqrt{\frac{3}{2}}}{2}\subcgt{3}{3}{3}+\frac{\sqrt{\frac{3}{2}}}{2}\subcgt{3}{\bar{3}}{3}
\end{align*}
\efl 
\item $\tsprodx{3}{8}\to\rep{6}$: \bfl 
\begin{align*}
\rep{3} & =\frac{1}{2}\subcgt{3}{1^{\prime}}{3}-\frac{1}{2}\subcgt{3}{\bar{1}^{\prime}}{3}\\
 & -\frac{i}{2}\subcgt{3}{3}{3}+\frac{i}{2}\subcgt{3}{\bar{3}}{3}
\end{align*}
\begin{align*}
\rep{\bar{3}} & =\frac{i}{2}\subcgt{3}{3}{\bar{3}_{s}}-\frac{i}{2}\subcgt{3}{3}{\bar{3}_{a}}\\
 & -\frac{i}{\sqrt{2}}\subcgt{3}{\bar{3}}{\bar{3}}
\end{align*}
\efl 
\item $\tsprodx{3}{8}\to\rep{7}$: \bfl 
\begin{align*}
\rep{1} & =\subcgt{3}{\bar{3}}{1}
\end{align*}
\begin{align*}
\rep{3} & =-\frac{1}{2}i\sqrt{\frac{7}{6}}e^{i\alpha}\subcgt{3}{1^{\prime}}{3}+\frac{1}{2}i\sqrt{\frac{7}{6}}e^{-i\alpha}\subcgt{3}{\bar{1}^{\prime}}{3}\\
 & +\frac{1}{2\sqrt{6}}\subcgt{3}{3}{3}-\frac{\sqrt{\frac{3}{2}}}{2}\subcgt{3}{\bar{3}}{3}
\end{align*}
\begin{align*}
\rep{\bar{3}} & =\frac{\sqrt{3}}{2}\subcgt{3}{3}{\bar{3}_{s}}+\frac{1}{2\sqrt{3}}\subcgt{3}{3}{\bar{3}_{a}}\\
 & +\frac{1}{\sqrt{6}}\subcgt{3}{\bar{3}}{\bar{3}}
\end{align*}
\efl 
\item $\tsprodx{3}{8}\to\rep{8}$: \bfl 
\begin{align*}
\rep{1^{\prime}} & =\subcgt{3}{\bar{3}}{1^{\prime}}
\end{align*}
\begin{align*}
\rep{\bar{1}^{\prime}} & =-\omega^{2}\subcgt{3}{\bar{3}}{\bar{1}^{\prime}}
\end{align*}
\begin{align*}
\rep{3} & =\frac{\omega^{2}}{\sqrt{3}}\subcgt{3}{1^{\prime}}{3}-\frac{1}{\sqrt{3}}\subcgt{3}{\bar{1}^{\prime}}{3}\\
 & -\frac{i\omega}{\sqrt{3}}\subcgt{3}{3}{3}
\end{align*}
\begin{align*}
\rep{\bar{3}} & =-i\sqrt{\frac{2}{3}}\omega\subcgt{3}{3}{\bar{3}_{a}}+\frac{i\omega}{\sqrt{3}}\subcgt{3}{\bar{3}}{\bar{3}}
\end{align*}
\efl 
\end{itemize}

\paragraph*{\cgEqFontsize $\tsprod{6}{6}\to\left(\rep{1}+\rep{6}^{(1)}+\rep{6}^{(2)}+\rep{8}\right)_{s}+\left(\rep{7}+\rep{8}\right)_{a}$}
\begin{itemize}
\item $\tsprodx{6}{6}\to\rep{1}_{s}$: \bfl 
\begin{align*}
\rep{1} & =\frac{1}{\sqrt{2}}\left(\subcgt{3}{\bar{3}}{1}+\subcgt{\bar{3}}{3}{1}\right)
\end{align*}
\efl 
\item $\tsprodx{6}{6}\to\rep{6}_{s}^{(1)}$: \bfl 
\begin{align*}
\rep{3} & =\sqrt{\frac{1}{14}\left(3-\sqrt{2}\right)}\subcgt{3}{3}{3}\\
 & +\sqrt{\frac{1}{14}\left(3-\sqrt{2}\right)}\left(\subcgt{3}{\bar{3}}{3}+\subcgt{\bar{3}}{3}{3}\right)\\
 & -\sqrt{\frac{1}{14}\left(5+3\sqrt{2}\right)}\subcgt{\bar{3}}{\bar{3}}{3_{s}}
\end{align*}
\begin{align*}
\rep{\bar{3}} & =-\sqrt{\frac{1}{14}\left(5+3\sqrt{2}\right)}\subcgt{3}{3}{\bar{3}_{s}}\\
 & +\sqrt{\frac{1}{14}\left(3-\sqrt{2}\right)}\left(\subcgt{3}{\bar{3}}{\bar{3}}+\subcgt{\bar{3}}{3}{\bar{3}}\right)\\
 & +\sqrt{\frac{1}{14}\left(3-\sqrt{2}\right)}\subcgt{\bar{3}}{\bar{3}}{\bar{3}}
\end{align*}
\efl 
\item $\tsprodx{6}{6}\to\rep{6}_{s}^{(2)}$: \bfl 
\begin{align*}
\rep{3} & =i\sqrt{\frac{1}{14}\left(3+\sqrt{2}\right)}\subcgt{3}{3}{3}\\
 & -i\sqrt{\frac{1}{14}\left(3+\sqrt{2}\right)}\left(\subcgt{3}{\bar{3}}{3}+\subcgt{\bar{3}}{3}{3}\right)\\
 & -i\sqrt{\frac{1}{14}\left(5-3\sqrt{2}\right)}\subcgt{\bar{3}}{\bar{3}}{3_{s}}
\end{align*}
\begin{align*}
\rep{\bar{3}} & =i\sqrt{\frac{1}{14}\left(5-3\sqrt{2}\right)}\subcgt{3}{3}{\bar{3}_{s}}\\
 & +i\sqrt{\frac{1}{14}\left(3+\sqrt{2}\right)}\left(\subcgt{3}{\bar{3}}{\bar{3}}+\subcgt{\bar{3}}{3}{\bar{3}}\right)\\
 & -i\sqrt{\frac{1}{14}\left(3+\sqrt{2}\right)}\subcgt{\bar{3}}{\bar{3}}{\bar{3}}
\end{align*}
\efl 
\item $\tsprodx{6}{6}\to\rep{8}_{s}$: \bfl 
\begin{align*}
\rep{1^{\prime}} & =\frac{e^{-i\alpha}}{\sqrt{2}}\left(\subcgt{3}{\bar{3}}{1^{\prime}}+\subcgt{\bar{3}}{3}{1^{\prime}}\right)
\end{align*}
\begin{align*}
\rep{\bar{1}^{\prime}} & =\frac{e^{i\alpha}}{\sqrt{2}}\left(\subcgt{3}{\bar{3}}{\bar{1}^{\prime}}+\subcgt{\bar{3}}{3}{\bar{1}^{\prime}}\right)
\end{align*}
\begin{align*}
\rep{3} & =-\frac{2}{\sqrt{7}}\subcgt{3}{3}{3}\\
 & -\frac{1}{\sqrt{14}}\left(\subcgt{3}{\bar{3}}{3}+\subcgt{\bar{3}}{3}{3}\right)\\
 & -\sqrt{\frac{2}{7}}\subcgt{\bar{3}}{\bar{3}}{3_{s}}
\end{align*}
\begin{align*}
\rep{\bar{3}} & =-\sqrt{\frac{2}{7}}\subcgt{3}{3}{\bar{3}_{s}}\\
 & -\frac{1}{\sqrt{14}}\left(\subcgt{3}{\bar{3}}{\bar{3}}+\subcgt{\bar{3}}{3}{\bar{3}}\right)\\
 & -\frac{2}{\sqrt{7}}\subcgt{\bar{3}}{\bar{3}}{\bar{3}}
\end{align*}
\efl 
\item $\tsprodx{6}{6}\to\rep{7}_{a}$: \bfl 
\begin{align*}
\rep{1} & =\frac{i}{\sqrt{2}}\left(\subcgt{3}{\bar{3}}{1}-\subcgt{\bar{3}}{3}{1}\right)
\end{align*}
\begin{align*}
\rep{3} & =-\frac{i}{\sqrt{3}}\left(\subcgt{3}{\bar{3}}{3}-\subcgt{\bar{3}}{3}{3}\right)\\
 & -\frac{i}{\sqrt{3}}\subcgt{\bar{3}}{\bar{3}}{3_{a}}
\end{align*}
\begin{align*}
\rep{\bar{3}} & =\frac{i}{\sqrt{3}}\subcgt{3}{3}{\bar{3}_{a}}\\
 & -\frac{i}{\sqrt{3}}\left(\subcgt{3}{\bar{3}}{\bar{3}}-\subcgt{\bar{3}}{3}{\bar{3}}\right)
\end{align*}
\efl 
\item $\tsprodx{6}{6}\to\rep{8}_{a}$: \bfl 
\begin{align*}
\rep{1^{\prime}} & =\frac{1}{\sqrt{2}}\left(\subcgt{3}{\bar{3}}{1^{\prime}}-\subcgt{\bar{3}}{3}{1^{\prime}}\right)
\end{align*}
\begin{align*}
\rep{\bar{1}^{\prime}} & =-\frac{1}{\sqrt{2}}\left(\subcgt{3}{\bar{3}}{\bar{1}^{\prime}}-\subcgt{\bar{3}}{3}{\bar{1}^{\prime}}\right)
\end{align*}
\begin{align*}
\rep{3} & =\frac{i}{\sqrt{6}}\left(\subcgt{3}{\bar{3}}{3}-\subcgt{\bar{3}}{3}{3}\right)\\
 & -i\sqrt{\frac{2}{3}}\subcgt{\bar{3}}{\bar{3}}{3_{a}}
\end{align*}
\begin{align*}
\rep{\bar{3}} & =i\sqrt{\frac{2}{3}}\subcgt{3}{3}{\bar{3}_{a}}\\
 & +\frac{i}{\sqrt{6}}\left(\subcgt{3}{\bar{3}}{\bar{3}}-\subcgt{\bar{3}}{3}{\bar{3}}\right)
\end{align*}
\efl 
\end{itemize}

\paragraph*{\cgEqFontsize $\tsprod{6}{7}\to\rep{3}+\rep{\bar{3}}+\rep{6}+\rep{7}^{(1)}+\rep{7}^{(2)}+\rep{8}^{(1)}+\rep{8}^{(2)}$}
\begin{itemize}
\item $\tsprodx{6}{7}\to\rep{3}$: \bfl 
\begin{align*}
\rep{3} & =\frac{1}{\sqrt{7}}\subcgt{3}{1}{3}+\sqrt{\frac{2}{7}}\subcgt{3}{3}{3}\\
 & +\frac{1}{\sqrt{14}}\subcgt{3}{\bar{3}}{3}+\frac{1}{\sqrt{7}}\subcgt{\bar{3}}{3}{3}\\
 & -\frac{1}{\sqrt{14}}\subcgt{\bar{3}}{\bar{3}}{3_{s}}+\sqrt{\frac{2}{7}}\subcgt{\bar{3}}{\bar{3}}{3_{a}}
\end{align*}
\efl 
\item $\tsprodx{6}{7}\to\rep{\bar{3}}$: \bfl 
\begin{align*}
\rep{\bar{3}} & =\frac{1}{\sqrt{14}}\subcgt{3}{3}{\bar{3}_{s}}-\sqrt{\frac{2}{7}}\subcgt{3}{3}{\bar{3}_{a}}\\
 & -\frac{1}{\sqrt{7}}\subcgt{3}{\bar{3}}{\bar{3}}-\frac{1}{\sqrt{7}}\subcgt{\bar{3}}{1}{\bar{3}}\\
 & -\frac{1}{\sqrt{14}}\subcgt{\bar{3}}{3}{\bar{3}}-\sqrt{\frac{2}{7}}\subcgt{\bar{3}}{\bar{3}}{\bar{3}}
\end{align*}
\efl 
\item $\tsprodx{6}{7}\to\rep{6}$: \bfl 
\begin{align*}
\rep{3} & =\frac{i}{\sqrt{7}}\subcgt{3}{1}{3}-i\sqrt{\frac{2}{7}}\subcgt{3}{3}{3}\\
 & -i\sqrt{\frac{2}{7}}\subcgt{3}{\bar{3}}{3}+i\sqrt{\frac{2}{7}}\subcgt{\bar{3}}{\bar{3}}{3_{a}}
\end{align*}
\begin{align*}
\rep{\bar{3}} & =-i\sqrt{\frac{2}{7}}\subcgt{3}{3}{\bar{3}_{a}}-\frac{i}{\sqrt{7}}\subcgt{\bar{3}}{1}{\bar{3}}\\
 & +i\sqrt{\frac{2}{7}}\subcgt{\bar{3}}{3}{\bar{3}}+i\sqrt{\frac{2}{7}}\subcgt{\bar{3}}{\bar{3}}{\bar{3}}
\end{align*}
\efl 
\item $\tsprodx{6}{7}\to\rep{7}^{(1)}$: \bfl 
\begin{align*}
\rep{1} & =\frac{1}{\sqrt{2}}\subcgt{3}{\bar{3}}{1}+\frac{1}{\sqrt{2}}\subcgt{\bar{3}}{3}{1}
\end{align*}
\begin{align*}
\rep{3} & =\frac{1}{\sqrt{6}}\subcgt{3}{1}{3}-\frac{1}{\sqrt{6}}\subcgt{3}{3}{3}\\
 & +\frac{1}{6}\left(\sqrt{3}+\sqrt{6}\right)\subcgt{3}{\bar{3}}{3}-\frac{1}{\sqrt{6}}\subcgt{\bar{3}}{3}{3}\\
 & +\frac{1}{6}\left(\sqrt{3}-\sqrt{6}\right)\subcgt{\bar{3}}{\bar{3}}{3_{s}}
\end{align*}
\begin{align*}
\rep{\bar{3}} & =\frac{1}{6}\left(\sqrt{3}-\sqrt{6}\right)\subcgt{3}{3}{\bar{3}_{s}}-\frac{1}{\sqrt{6}}\subcgt{3}{\bar{3}}{\bar{3}}\\
 & +\frac{1}{\sqrt{6}}\subcgt{\bar{3}}{1}{\bar{3}}+\frac{1}{6}\left(\sqrt{3}+\sqrt{6}\right)\subcgt{\bar{3}}{3}{\bar{3}}\\
 & -\frac{1}{\sqrt{6}}\subcgt{\bar{3}}{\bar{3}}{\bar{3}}
\end{align*}
\efl 
\item $\tsprodx{6}{7}\to\rep{7}^{(2)}$: \bfl 
\begin{align*}
\rep{1} & =\frac{i}{\sqrt{2}}\subcgt{3}{\bar{3}}{1}-\frac{i}{\sqrt{2}}\subcgt{\bar{3}}{3}{1}
\end{align*}
\begin{align*}
\rep{3} & =\frac{i}{\sqrt{6}}\subcgt{3}{1}{3}+\frac{i}{\sqrt{6}}\subcgt{3}{3}{3}\\
 & -\frac{i\left(\sqrt{2}-1\right)}{2\sqrt{3}}\subcgt{3}{\bar{3}}{3}-\frac{i}{\sqrt{6}}\subcgt{\bar{3}}{3}{3}\\
 & +\frac{1}{6}i\left(\sqrt{3}+\sqrt{6}\right)\subcgt{\bar{3}}{\bar{3}}{3_{s}}
\end{align*}
\begin{align*}
\rep{\bar{3}} & =-\frac{1}{6}i\left(\sqrt{3}+\sqrt{6}\right)\subcgt{3}{3}{\bar{3}_{s}}+\frac{i}{\sqrt{6}}\subcgt{3}{\bar{3}}{\bar{3}}\\
 & -\frac{i}{\sqrt{6}}\subcgt{\bar{3}}{1}{\bar{3}}+\frac{i\left(\sqrt{2}-1\right)}{2\sqrt{3}}\subcgt{\bar{3}}{3}{\bar{3}}\\
 & -\frac{i}{\sqrt{6}}\subcgt{\bar{3}}{\bar{3}}{\bar{3}}
\end{align*}
\efl 
\item $\tsprodx{6}{7}\to\rep{8}^{(1)}$: \bfl 
\begin{align*}
\rep{1^{\prime}} & =-\frac{i\omega e^{2i\beta-i\alpha}}{\sqrt{2}}\subcgt{3}{\bar{3}}{1^{\prime}}-\frac{i\omega e^{2i\beta-i\alpha}}{\sqrt{2}}\subcgt{\bar{3}}{3}{1^{\prime}}
\end{align*}
\begin{align*}
\rep{\bar{1}^{\prime}} & =\frac{i\omega^{2}e^{i\alpha-2i\beta}}{\sqrt{2}}\subcgt{3}{\bar{3}}{\bar{1}^{\prime}}+\frac{i\omega^{2}e^{i\alpha-2i\beta}}{\sqrt{2}}\subcgt{\bar{3}}{3}{\bar{1}^{\prime}}
\end{align*}
\begin{align*}
\rep{3} & =\frac{2}{\sqrt{21}}\subcgt{3}{1}{3}+\frac{1}{\sqrt{21}}\subcgt{3}{3}{3}\\
 & -\frac{2+\sqrt{2}}{2\sqrt{21}}\subcgt{3}{\bar{3}}{3}+\frac{2-3\sqrt{2}}{2\sqrt{21}}\subcgt{\bar{3}}{3}{3}\\
 & -\frac{4+\sqrt{2}}{2\sqrt{21}}\subcgt{\bar{3}}{\bar{3}}{3_{s}}-\sqrt{\frac{3}{14}}\subcgt{\bar{3}}{\bar{3}}{3_{a}}
\end{align*}
\begin{align*}
\rep{\bar{3}} & =-\frac{4+\sqrt{2}}{2\sqrt{21}}\subcgt{3}{3}{\bar{3}_{s}}-\sqrt{\frac{3}{14}}\subcgt{3}{3}{\bar{3}_{a}}\\
 & +\frac{2-3\sqrt{2}}{2\sqrt{21}}\subcgt{3}{\bar{3}}{\bar{3}}+\frac{2}{\sqrt{21}}\subcgt{\bar{3}}{1}{\bar{3}}\\
 & -\frac{2+\sqrt{2}}{2\sqrt{21}}\subcgt{\bar{3}}{3}{\bar{3}}+\frac{1}{\sqrt{21}}\subcgt{\bar{3}}{\bar{3}}{\bar{3}}
\end{align*}
\efl 
\item $\tsprodx{6}{7}\to\rep{8}^{(2)}$: \bfl 
\begin{align*}
\rep{1^{\prime}} & =-\frac{\omega^{2}e^{i\alpha-2i\beta}}{\sqrt{2}}\subcgt{3}{\bar{3}}{1^{\prime}}+\frac{\omega^{2}e^{i\alpha-2i\beta}}{\sqrt{2}}\subcgt{\bar{3}}{3}{1^{\prime}}
\end{align*}
\begin{align*}
\rep{\bar{1}^{\prime}} & =\frac{\omega e^{2i\beta-i\alpha}}{\sqrt{2}}\subcgt{3}{\bar{3}}{\bar{1}^{\prime}}-\frac{\omega e^{2i\beta-i\alpha}}{\sqrt{2}}\subcgt{\bar{3}}{3}{\bar{1}^{\prime}}
\end{align*}
\begin{align*}
\rep{3} & =-\frac{2i}{\sqrt{21}}\subcgt{3}{1}{3}+\frac{i}{\sqrt{21}}\subcgt{3}{3}{3}\\
 & +\frac{i\left(\sqrt{2}-2\right)}{2\sqrt{21}}\subcgt{3}{\bar{3}}{3}-\frac{i\left(2+3\sqrt{2}\right)}{2\sqrt{21}}\subcgt{\bar{3}}{3}{3}\\
 & +\frac{i\left(\sqrt{2}-4\right)}{2\sqrt{21}}\subcgt{\bar{3}}{\bar{3}}{3_{s}}+i\sqrt{\frac{3}{14}}\subcgt{\bar{3}}{\bar{3}}{3_{a}}
\end{align*}
\begin{align*}
\rep{\bar{3}} & =-\frac{i\left(\sqrt{2}-4\right)}{2\sqrt{21}}\subcgt{3}{3}{\bar{3}_{s}}-i\sqrt{\frac{3}{14}}\subcgt{3}{3}{\bar{3}_{a}}\\
 & +\frac{i\left(2+3\sqrt{2}\right)}{2\sqrt{21}}\subcgt{3}{\bar{3}}{\bar{3}}+\frac{2i}{\sqrt{21}}\subcgt{\bar{3}}{1}{\bar{3}}\\
 & -\frac{i\left(\sqrt{2}-2\right)}{2\sqrt{21}}\subcgt{\bar{3}}{3}{\bar{3}}-\frac{i}{\sqrt{21}}\subcgt{\bar{3}}{\bar{3}}{\bar{3}}
\end{align*}
\efl 
\end{itemize}

\paragraph*{\cgEqFontsize $\tsprod{6}{8}\to\rep{3}+\rep{\bar{3}}+\rep{6}^{(1)}+\rep{6}^{(2)}+\rep{7}^{(1)}+\rep{7}^{(2)}+\rep{8}^{(1)}+\rep{8}^{(2)}$}
\begin{itemize}
\item $\tsprodx{6}{8}\to\rep{3}$: \bfl 
\begin{align*}
\rep{3} & =\frac{1}{2\sqrt{2}}\subcgt{3}{1^{\prime}}{3}-\frac{1}{2\sqrt{2}}\subcgt{3}{\bar{1}^{\prime}}{3}\\
 & +\frac{i}{2\sqrt{2}}\subcgt{3}{3}{3}-\frac{i}{2\sqrt{2}}\subcgt{3}{\bar{3}}{3}\\
 & -\frac{i}{2}\subcgt{\bar{3}}{3}{3}+\frac{i}{2\sqrt{2}}\subcgt{\bar{3}}{\bar{3}}{3_{s}}\\
 & +\frac{i}{2\sqrt{2}}\subcgt{\bar{3}}{\bar{3}}{3_{a}}
\end{align*}
\efl 
\item $\tsprodx{6}{8}\to\rep{\bar{3}}$: \bfl 
\begin{align*}
\rep{\bar{3}} & =\frac{1}{2\sqrt{2}}\subcgt{3}{3}{\bar{3}_{s}}+\frac{1}{2\sqrt{2}}\subcgt{3}{3}{\bar{3}_{a}}\\
 & -\frac{1}{2}\subcgt{3}{\bar{3}}{\bar{3}}-\frac{i}{2\sqrt{2}}\subcgt{\bar{3}}{1^{\prime}}{\bar{3}}\\
 & +\frac{i}{2\sqrt{2}}\subcgt{\bar{3}}{\bar{1}^{\prime}}{\bar{3}}-\frac{1}{2\sqrt{2}}\subcgt{\bar{3}}{3}{\bar{3}}\\
 & +\frac{1}{2\sqrt{2}}\subcgt{\bar{3}}{\bar{3}}{\bar{3}}
\end{align*}
\efl 
\item $\tsprodx{6}{8}\to\rep{6}^{(1)}$: \bfl 
\begin{align*}
\rep{3} & =-\frac{e^{i\alpha}}{2\sqrt{2}}\subcgt{3}{1^{\prime}}{3}-\frac{e^{-i\alpha}}{2\sqrt{2}}\subcgt{3}{\bar{1}^{\prime}}{3}\\
 & +\frac{\sqrt{\frac{3}{14}}}{2}\subcgt{3}{3}{3}+\frac{\sqrt{\frac{3}{14}}}{2}\subcgt{3}{\bar{3}}{3}\\
 & +\sqrt{\frac{3}{7}}\subcgt{\bar{3}}{3}{3}+\sqrt{\frac{3}{14}}\subcgt{\bar{3}}{\bar{3}}{3_{s}}
\end{align*}
\begin{align*}
\rep{\bar{3}} & =\sqrt{\frac{3}{14}}\subcgt{3}{3}{\bar{3}_{s}}+\sqrt{\frac{3}{7}}\subcgt{3}{\bar{3}}{\bar{3}}\\
 & -\frac{e^{i\alpha}}{2\sqrt{2}}\subcgt{\bar{3}}{1^{\prime}}{\bar{3}}-\frac{e^{-i\alpha}}{2\sqrt{2}}\subcgt{\bar{3}}{\bar{1}^{\prime}}{\bar{3}}\\
 & +\frac{\sqrt{\frac{3}{14}}}{2}\subcgt{\bar{3}}{3}{\bar{3}}+\frac{\sqrt{\frac{3}{14}}}{2}\subcgt{\bar{3}}{\bar{3}}{\bar{3}}
\end{align*}
\efl 
\item $\tsprodx{6}{8}\to\rep{6}^{(2)}$: \bfl 
\begin{align*}
\rep{3} & =-\frac{1}{2\sqrt{2}}\subcgt{3}{1^{\prime}}{3}+\frac{1}{2\sqrt{2}}\subcgt{3}{\bar{1}^{\prime}}{3}\\
 & +\frac{i}{2\sqrt{2}}\subcgt{3}{3}{3}+\frac{i}{2\sqrt{2}}\subcgt{3}{\bar{3}}{3}\\
 & +\frac{i}{\sqrt{2}}\subcgt{\bar{3}}{\bar{3}}{3_{a}}
\end{align*}
\begin{align*}
\rep{\bar{3}} & =-\frac{i}{\sqrt{2}}\subcgt{3}{3}{\bar{3}_{a}}+\frac{1}{2\sqrt{2}}\subcgt{\bar{3}}{1^{\prime}}{\bar{3}}\\
 & -\frac{1}{2\sqrt{2}}\subcgt{\bar{3}}{\bar{1}^{\prime}}{\bar{3}}-\frac{i}{2\sqrt{2}}\subcgt{\bar{3}}{3}{\bar{3}}\\
 & -\frac{i}{2\sqrt{2}}\subcgt{\bar{3}}{\bar{3}}{\bar{3}}
\end{align*}
\efl 
\item $\tsprodx{6}{8}\to\rep{7}^{(1)}$: \bfl 
\begin{align*}
\rep{1} & =\frac{1}{\sqrt{2}}\subcgt{3}{\bar{3}}{1}+\frac{1}{\sqrt{2}}\subcgt{\bar{3}}{3}{1}
\end{align*}
\begin{align*}
\rep{3} & =\frac{i\sqrt{7}\omega^{2}}{4\sqrt{3}}e^{i\alpha-2i\beta}\subcgt{3}{1^{\prime}}{3}-\frac{i\sqrt{7}\omega}{4\sqrt{3}}e^{2i\beta-i\alpha}\subcgt{3}{\bar{1}^{\prime}}{3}\\
 & +\frac{\sqrt{2}-3}{4\sqrt{3}}\subcgt{3}{3}{3}-\frac{1+\sqrt{2}}{4\sqrt{3}}\subcgt{3}{\bar{3}}{3}\\
 & +\frac{1}{2\sqrt{6}}\subcgt{\bar{3}}{3}{3}-\frac{1+2\sqrt{2}}{4\sqrt{3}}\subcgt{\bar{3}}{\bar{3}}{3_{s}}\\
 & +\frac{\sqrt{3}}{4}\subcgt{\bar{3}}{\bar{3}}{3_{a}}
\end{align*}
\begin{align*}
\rep{\bar{3}} & =-\frac{1+2\sqrt{2}}{4\sqrt{3}}\subcgt{3}{3}{\bar{3}_{s}}+\frac{\sqrt{3}}{4}\subcgt{3}{3}{\bar{3}_{a}}\\
 & +\frac{1}{2\sqrt{6}}\subcgt{3}{\bar{3}}{\bar{3}}+\frac{1}{4}i\sqrt{\frac{7}{3}}\omega^{2}e^{i\alpha-2i\beta}\subcgt{\bar{3}}{1^{\prime}}{\bar{3}}\\
 & -\frac{1}{4}i\sqrt{\frac{7}{3}}\omega e^{2i\beta-i\alpha}\subcgt{\bar{3}}{\bar{1}^{\prime}}{\bar{3}}-\frac{1+\sqrt{2}}{4\sqrt{3}}\subcgt{\bar{3}}{3}{\bar{3}}\\
 & +\frac{\sqrt{2}-3}{4\sqrt{3}}\subcgt{\bar{3}}{\bar{3}}{\bar{3}}
\end{align*}
\efl 
\item $\tsprodx{6}{8}\to\rep{7}^{(2)}$: \bfl 
\begin{align*}
\rep{1} & =\frac{i}{\sqrt{2}}\subcgt{3}{\bar{3}}{1}-\frac{i}{\sqrt{2}}\subcgt{\bar{3}}{3}{1}
\end{align*}
\begin{align*}
\rep{3} & =-\frac{\sqrt{7}\omega^{2}}{4\sqrt{3}}e^{2i\beta-i\alpha}\subcgt{3}{1^{\prime}}{3}+\frac{\sqrt{7}\omega^{2}}{4\sqrt{3}}e^{i\alpha-2i\beta}\subcgt{3}{\bar{1}^{\prime}}{3}\\
 & -\frac{i\left(3+\sqrt{2}\right)}{4\sqrt{3}}\subcgt{3}{3}{3}+\frac{i\left(\sqrt{2}-1\right)}{4\sqrt{3}}\subcgt{3}{\bar{3}}{3}\\
 & +\frac{i}{2\sqrt{6}}\subcgt{\bar{3}}{3}{3}+\frac{1}{12}i\left(2\sqrt{6}-\sqrt{3}\right)\subcgt{\bar{3}}{\bar{3}}{3_{s}}\\
 & +\frac{i\sqrt{3}}{4}\subcgt{\bar{3}}{\bar{3}}{3_{a}}
\end{align*}
\begin{align*}
\rep{\bar{3}} & =-\frac{i\left(2\sqrt{2}-1\right)}{4\sqrt{3}}\subcgt{3}{3}{\bar{3}_{s}}-\frac{i\sqrt{3}}{4}\subcgt{3}{3}{\bar{3}_{a}}\\
 & -\frac{i}{2\sqrt{6}}\subcgt{3}{\bar{3}}{\bar{3}}+\frac{1}{4}\sqrt{\frac{7}{3}}\omega e^{2i\beta-i\alpha}\subcgt{\bar{3}}{1^{\prime}}{\bar{3}}\\
 & -\frac{1}{4}\sqrt{\frac{7}{3}}\omega^{2}e^{i\alpha-2i\beta}\subcgt{\bar{3}}{\bar{1}^{\prime}}{\bar{3}}-\frac{i\left(\sqrt{2}-1\right)}{4\sqrt{3}}\subcgt{\bar{3}}{3}{\bar{3}}\\
 & +\frac{i\left(3+\sqrt{2}\right)}{4\sqrt{3}}\subcgt{\bar{3}}{\bar{3}}{\bar{3}}
\end{align*}
\efl 
\item $\tsprodx{6}{8}\to\rep{8}^{(1)}$: \bfl 
\begin{align*}
\rep{1^{\prime}} & =-\frac{ie^{-i\beta}}{\sqrt{2}}\subcgt{3}{\bar{3}}{1^{\prime}}-\frac{ie^{-i\beta}}{\sqrt{2}}\subcgt{\bar{3}}{3}{1^{\prime}}
\end{align*}
\begin{align*}
\rep{\bar{1}^{\prime}} & =\frac{ie^{i\beta}}{\sqrt{2}}\subcgt{3}{\bar{3}}{\bar{1}^{\prime}}+\frac{ie^{i\beta}}{\sqrt{2}}\subcgt{\bar{3}}{3}{\bar{1}^{\prime}}
\end{align*}
\begin{align*}
\rep{3} & =\frac{ie^{i\beta}}{\sqrt{6}}\subcgt{3}{1^{\prime}}{3}-\frac{ie^{-i\beta}}{\sqrt{6}}\subcgt{3}{\bar{1}^{\prime}}{3}\\
 & -\sqrt{\frac{5+3\sqrt{2}}{42}}\subcgt{3}{3}{3}+\sqrt{\frac{6-2\sqrt{2}}{21}}\subcgt{3}{\bar{3}}{3}\\
 & -\sqrt{\frac{5+3\sqrt{2}}{42}}\subcgt{\bar{3}}{3}{3}+\sqrt{\frac{3-\sqrt{2}}{21}}\subcgt{\bar{3}}{\bar{3}}{3_{s}}
\end{align*}
\begin{align*}
\rep{\bar{3}} & =\sqrt{\frac{3-\sqrt{2}}{21}}\subcgt{3}{3}{\bar{3}_{s}}-\sqrt{\frac{5+3\sqrt{2}}{42}}\subcgt{3}{\bar{3}}{\bar{3}}\\
 & +\frac{ie^{i\beta}}{\sqrt{6}}\subcgt{\bar{3}}{1^{\prime}}{\bar{3}}-\frac{ie^{-i\beta}}{\sqrt{6}}\subcgt{\bar{3}}{\bar{1}^{\prime}}{\bar{3}}\\
 & +\sqrt{\frac{6-2\sqrt{2}}{21}}\subcgt{\bar{3}}{3}{\bar{3}}-\sqrt{\frac{5+3\sqrt{2}}{42}}\subcgt{\bar{3}}{\bar{3}}{\bar{3}}
\end{align*}
\efl 
\item $\tsprodx{6}{8}\to\rep{8}^{(2)}$: \bfl 
\begin{align*}
\rep{1^{\prime}} & =\frac{\omega e^{i\beta-i\alpha}}{\sqrt{2}}\subcgt{3}{\bar{3}}{1^{\prime}}-\frac{\omega e^{i\beta-i\alpha}}{\sqrt{2}}\subcgt{\bar{3}}{3}{1^{\prime}}
\end{align*}
\begin{align*}
\rep{\bar{1}^{\prime}} & =-\frac{\omega^{2}e^{i\alpha-i\beta}}{\sqrt{2}}\subcgt{3}{\bar{3}}{\bar{1}^{\prime}}+\frac{\omega^{2}e^{i\alpha-i\beta}}{\sqrt{2}}\subcgt{\bar{3}}{3}{\bar{1}^{\prime}}
\end{align*}
\begin{align*}
\rep{3} & =-\frac{\omega^{2}e^{i\alpha-i\beta}}{\sqrt{6}}\subcgt{3}{1^{\prime}}{3}+\frac{\omega e^{i\beta-i\alpha}}{\sqrt{6}}\subcgt{3}{\bar{1}^{\prime}}{3}\\
 & +i\sqrt{\frac{5-3\sqrt{2}}{42}}\subcgt{3}{3}{3}+i\sqrt{\frac{6+2\sqrt{2}}{21}}\subcgt{3}{\bar{3}}{3}\\
 & -i\sqrt{\frac{5-3\sqrt{2}}{42}}\subcgt{\bar{3}}{3}{3}-i\sqrt{\frac{3+\sqrt{2}}{21}}\subcgt{\bar{3}}{\bar{3}}{3_{s}}
\end{align*}
\begin{align*}
\rep{\bar{3}} & =i\sqrt{\frac{3+\sqrt{2}}{21}}\subcgt{3}{3}{\bar{3}_{s}}+i\sqrt{\frac{5-3\sqrt{2}}{42}}\subcgt{3}{\bar{3}}{\bar{3}}\\
 & +\frac{\omega^{2}e^{i\alpha-i\beta}}{\sqrt{6}}\subcgt{\bar{3}}{1^{\prime}}{\bar{3}}-\frac{\omega e^{i\beta-i\alpha}}{\sqrt{6}}\subcgt{\bar{3}}{\bar{1}^{\prime}}{\bar{3}}\\
 & -i\sqrt{\frac{6+2\sqrt{2}}{21}}\subcgt{\bar{3}}{3}{\bar{3}}-i\sqrt{\frac{5-3\sqrt{2}}{42}}\subcgt{\bar{3}}{\bar{3}}{\bar{3}}
\end{align*}
\efl 
\end{itemize}

\paragraph*{\cgEqFontsize $\tsprod{7}{7}\to\left(\rep{1}+\rep{6}^{(1)}+\rep{6}^{(2)}+\rep{7}+\rep{8}\right)_{s}+\left(\rep{3}+\rep{\bar{3}}+\rep{7}+\rep{8}\right)_{a}$}
\begin{itemize}
\item $\tsprodx{7}{7}\to\rep{1}_{s}$: \bfl 
\begin{align*}
\rep{1} & =\frac{1}{\sqrt{7}}\subcgt{1}{1}{1}\\
 & +\sqrt{\frac{3}{7}}\left(\subcgt{3}{\bar{3}}{1}+\subcgt{\bar{3}}{3}{1}\right)
\end{align*}
\efl 
\item $\tsprodx{7}{7}\to\rep{6}_{s}^{(1)}$: \bfl 
\begin{align*}
\rep{3} & =\frac{1}{\sqrt{7}}\left(\subcgt{1}{3}{3}+\subcgt{3}{1}{3}\right)\\
 & +\frac{1}{14}\left(2\sqrt{7}+\sqrt{14}\right)\subcgt{3}{3}{3}\\
 & -\frac{1}{\sqrt{7}}\left(\subcgt{3}{\bar{3}}{3}+\subcgt{\bar{3}}{3}{3}\right)\\
 & +\frac{1}{14}\left(\sqrt{14}-2\sqrt{7}\right)\subcgt{\bar{3}}{\bar{3}}{3_{s}}
\end{align*}
\begin{align*}
\rep{\bar{3}} & =\frac{1}{\sqrt{7}}\left(\subcgt{1}{\bar{3}}{\bar{3}}+\subcgt{\bar{3}}{1}{\bar{3}}\right)\\
 & +\frac{1}{14}\left(\sqrt{14}-2\sqrt{7}\right)\subcgt{3}{3}{\bar{3}_{s}}\\
 & -\frac{1}{\sqrt{7}}\left(\subcgt{3}{\bar{3}}{\bar{3}}+\subcgt{\bar{3}}{3}{\bar{3}}\right)\\
 & +\frac{1}{14}\left(2\sqrt{7}+\sqrt{14}\right)\subcgt{\bar{3}}{\bar{3}}{\bar{3}}
\end{align*}
\efl 
\item $\tsprodx{7}{7}\to\rep{6}_{s}^{(2)}$: \bfl 
\begin{align*}
\rep{3} & =\frac{i}{\sqrt{7}}\left(\subcgt{1}{3}{3}+\subcgt{3}{1}{3}\right)\\
 & +\frac{i\left(\sqrt{2}-2\right)}{2\sqrt{7}}\subcgt{3}{3}{3}\\
 & +\frac{i}{\sqrt{7}}\left(\subcgt{3}{\bar{3}}{3}+\subcgt{\bar{3}}{3}{3}\right)\\
 & -\frac{1}{14}i\left(2\sqrt{7}+\sqrt{14}\right)\subcgt{\bar{3}}{\bar{3}}{3_{s}}
\end{align*}
\begin{align*}
\rep{\bar{3}} & =-\frac{i}{\sqrt{7}}\left(\subcgt{1}{\bar{3}}{\bar{3}}+\subcgt{\bar{3}}{1}{\bar{3}}\right)\\
 & +\frac{1}{14}i\left(2\sqrt{7}+\sqrt{14}\right)\subcgt{3}{3}{\bar{3}_{s}}\\
 & -\frac{i}{\sqrt{7}}\left(\subcgt{3}{\bar{3}}{\bar{3}}+\subcgt{\bar{3}}{3}{\bar{3}}\right)\\
 & -\frac{1}{14}i\left(\sqrt{14}-2\sqrt{7}\right)\subcgt{\bar{3}}{\bar{3}}{\bar{3}}
\end{align*}
\efl 
\item $\tsprodx{7}{7}\to\rep{7}_{s}$: \bfl 
\begin{align*}
\rep{1} & =\sqrt{\frac{6}{7}}\subcgt{1}{1}{1}\\
 & -\frac{1}{\sqrt{14}}\left(\subcgt{3}{\bar{3}}{1}+\subcgt{\bar{3}}{3}{1}\right)
\end{align*}
\begin{align*}
\rep{3} & =-\frac{1}{\sqrt{42}}\left(\subcgt{1}{3}{3}+\subcgt{3}{1}{3}\right)\\
 & -\frac{2}{\sqrt{21}}\subcgt{3}{3}{3}\\
 & -\frac{2}{\sqrt{21}}\left(\subcgt{3}{\bar{3}}{3}+\subcgt{\bar{3}}{3}{3}\right)\\
 & -2\sqrt{\frac{2}{21}}\subcgt{\bar{3}}{\bar{3}}{3_{s}}
\end{align*}
\begin{align*}
\rep{\bar{3}} & =-\frac{1}{\sqrt{42}}\left(\subcgt{1}{\bar{3}}{\bar{3}}+\subcgt{\bar{3}}{1}{\bar{3}}\right)\\
 & -2\sqrt{\frac{2}{21}}\subcgt{3}{3}{\bar{3}_{s}}\\
 & -\frac{2}{\sqrt{21}}\left(\subcgt{3}{\bar{3}}{\bar{3}}+\subcgt{\bar{3}}{3}{\bar{3}}\right)\\
 & -\frac{2}{\sqrt{21}}\subcgt{\bar{3}}{\bar{3}}{\bar{3}}
\end{align*}
\efl 
\item $\tsprodx{7}{7}\to\rep{8}_{s}$: \bfl 
\begin{align*}
\rep{1^{\prime}} & =\frac{ie^{-i\alpha}}{\sqrt{2}}\left(\subcgt{3}{\bar{3}}{1^{\prime}}+\subcgt{\bar{3}}{3}{1^{\prime}}\right)
\end{align*}
\begin{align*}
\rep{\bar{1}^{\prime}} & =-\frac{ie^{i\alpha}}{\sqrt{2}}\left(\subcgt{3}{\bar{3}}{\bar{1}^{\prime}}+\subcgt{\bar{3}}{3}{\bar{1}^{\prime}}\right)
\end{align*}
\begin{align*}
\rep{3} & =-\frac{2}{\sqrt{21}}\left(\subcgt{1}{3}{3}+\subcgt{3}{1}{3}\right)\\
 & +2\sqrt{\frac{2}{21}}\subcgt{3}{3}{3}\\
 & +\frac{1}{\sqrt{42}}\left(\subcgt{3}{\bar{3}}{3}+\subcgt{\bar{3}}{3}{3}\right)\\
 & -\frac{2}{\sqrt{21}}\subcgt{\bar{3}}{\bar{3}}{3_{s}}
\end{align*}
\begin{align*}
\rep{\bar{3}} & =-\frac{2}{\sqrt{21}}\left(\subcgt{1}{\bar{3}}{\bar{3}}+\subcgt{\bar{3}}{1}{\bar{3}}\right)\\
 & -\frac{2}{\sqrt{21}}\subcgt{3}{3}{\bar{3}_{s}}\\
 & +\frac{1}{\sqrt{42}}\left(\subcgt{3}{\bar{3}}{\bar{3}}+\subcgt{\bar{3}}{3}{\bar{3}}\right)\\
 & +2\sqrt{\frac{2}{21}}\subcgt{\bar{3}}{\bar{3}}{\bar{3}}
\end{align*}
\efl 
\item $\tsprodx{7}{7}\to\rep{3}_{a}$: \bfl 
\begin{align*}
\rep{3} & =\frac{1}{\sqrt{7}}\left(\subcgt{1}{3}{3}-\subcgt{3}{1}{3}\right)\\
 & +\sqrt{\frac{2}{7}}\left(\subcgt{3}{\bar{3}}{3}-\subcgt{\bar{3}}{3}{3}\right)\\
 & +\frac{1}{\sqrt{7}}\subcgt{\bar{3}}{\bar{3}}{3_{a}}
\end{align*}
\efl 
\item $\tsprodx{7}{7}\to\rep{\bar{3}}_{a}$: \bfl 
\begin{align*}
\rep{\bar{3}} & =\frac{1}{\sqrt{7}}\left(\subcgt{1}{\bar{3}}{\bar{3}}-\subcgt{\bar{3}}{1}{\bar{3}}\right)\\
 & +\frac{1}{\sqrt{7}}\subcgt{3}{3}{\bar{3}_{a}}\\
 & -\sqrt{\frac{2}{7}}\left(\subcgt{3}{\bar{3}}{\bar{3}}-\subcgt{\bar{3}}{3}{\bar{3}}\right)
\end{align*}
\efl 
\item $\tsprodx{7}{7}\to\rep{7}_{a}$: \bfl 
\begin{align*}
\rep{1} & =\frac{i}{\sqrt{2}}\left(\subcgt{3}{\bar{3}}{1}-\subcgt{\bar{3}}{3}{1}\right)
\end{align*}
\begin{align*}
\rep{3} & =\frac{i}{\sqrt{6}}\left(\subcgt{1}{3}{3}-\subcgt{3}{1}{3}\right)\\
 & -i\sqrt{\frac{2}{3}}\subcgt{\bar{3}}{\bar{3}}{3_{a}}
\end{align*}
\begin{align*}
\rep{\bar{3}} & =-\frac{i}{\sqrt{6}}\left(\subcgt{1}{\bar{3}}{\bar{3}}-\subcgt{\bar{3}}{1}{\bar{3}}\right)\\
 & +i\sqrt{\frac{2}{3}}\subcgt{3}{3}{\bar{3}_{a}}
\end{align*}
\efl 
\item $\tsprodx{7}{7}\to\rep{8}_{a}$: \bfl 
\begin{align*}
\rep{1^{\prime}} & =\frac{e^{i\alpha}}{\sqrt{2}}\left(\subcgt{3}{\bar{3}}{1^{\prime}}-\subcgt{\bar{3}}{3}{1^{\prime}}\right)
\end{align*}
\begin{align*}
\rep{\bar{1}^{\prime}} & =-\frac{e^{-i\alpha}}{\sqrt{2}}\left(\subcgt{3}{\bar{3}}{\bar{1}^{\prime}}-\subcgt{\bar{3}}{3}{\bar{1}^{\prime}}\right)
\end{align*}
\begin{align*}
\rep{3} & =\frac{2i}{\sqrt{21}}\left(\subcgt{1}{3}{3}-\subcgt{3}{1}{3}\right)\\
 & -i\sqrt{\frac{3}{14}}\left(\subcgt{3}{\bar{3}}{3}-\subcgt{\bar{3}}{3}{3}\right)\\
 & +\frac{2i}{\sqrt{21}}\subcgt{\bar{3}}{\bar{3}}{3_{a}}
\end{align*}
\begin{align*}
\rep{\bar{3}} & =-\frac{2i}{\sqrt{21}}\left(\subcgt{1}{\bar{3}}{\bar{3}}-\subcgt{\bar{3}}{1}{\bar{3}}\right)\\
 & -\frac{2i}{\sqrt{21}}\subcgt{3}{3}{\bar{3}_{a}}\\
 & -i\sqrt{\frac{3}{14}}\left(\subcgt{3}{\bar{3}}{\bar{3}}-\subcgt{\bar{3}}{3}{\bar{3}}\right)
\end{align*}
\efl 
\end{itemize}

\paragraph*{\cgEqFontsize $\tsprod{7}{8}\to\rep{3}+\rep{\bar{3}}+\rep{6}^{(1)}+\rep{6}^{(2)}+\rep{7}^{(1)}+\rep{7}^{(2)}+\rep{8}^{(1)}+\rep{8}^{(2)}+\rep{8}^{(3)}$}
\begin{itemize}
\item $\tsprodx{7}{8}\to\rep{3}$: \bfl 
\begin{align*}
\rep{3} & =\frac{1}{\sqrt{7}}\subcgt{1}{3}{3}-\frac{ie^{i\alpha}}{2\sqrt{2}}\subcgt{3}{1^{\prime}}{3}\\
 & +\frac{ie^{-i\alpha}}{2\sqrt{2}}\subcgt{3}{\bar{1}^{\prime}}{3}-\frac{3}{2\sqrt{14}}\subcgt{3}{3}{3}\\
 & +\frac{1}{2\sqrt{14}}\subcgt{3}{\bar{3}}{3}+\frac{1}{\sqrt{14}}\subcgt{\bar{3}}{3}{3}\\
 & +\frac{3}{2\sqrt{7}}\subcgt{\bar{3}}{\bar{3}}{3_{s}}-\frac{1}{2\sqrt{7}}\subcgt{\bar{3}}{\bar{3}}{3_{a}}
\end{align*}
\efl 
\item $\tsprodx{7}{8}\to\rep{\bar{3}}$: \bfl 
\begin{align*}
\rep{\bar{3}} & =\frac{1}{\sqrt{7}}\subcgt{1}{\bar{3}}{\bar{3}}+\frac{3}{2\sqrt{7}}\subcgt{3}{3}{\bar{3}_{s}}\\
 & -\frac{1}{2\sqrt{7}}\subcgt{3}{3}{\bar{3}_{a}}+\frac{1}{\sqrt{14}}\subcgt{3}{\bar{3}}{\bar{3}}\\
 & -\frac{ie^{i\alpha}}{2\sqrt{2}}\subcgt{\bar{3}}{1^{\prime}}{\bar{3}}+\frac{ie^{-i\alpha}}{2\sqrt{2}}\subcgt{\bar{3}}{\bar{1}^{\prime}}{\bar{3}}\\
 & +\frac{1}{2\sqrt{14}}\subcgt{\bar{3}}{3}{\bar{3}}-\frac{3}{2\sqrt{14}}\subcgt{\bar{3}}{\bar{3}}{\bar{3}}
\end{align*}
\efl 
\item $\tsprodx{7}{8}\to\rep{6}^{(1)}$: \bfl 
\begin{align*}
\rep{3} & =\frac{1}{\sqrt{7}}\subcgt{1}{3}{3}+\frac{i\omega^{2}e^{i\alpha-2i\beta}}{2\sqrt{2}}\subcgt{3}{1^{\prime}}{3}\\
 & -\frac{i\omega e^{2i\beta-i\alpha}}{2\sqrt{2}}\subcgt{3}{\bar{1}^{\prime}}{3}-\frac{2+\sqrt{2}}{4\sqrt{7}}\subcgt{3}{3}{3}\\
 & -\frac{3\sqrt{2}-2}{4\sqrt{7}}\subcgt{3}{\bar{3}}{3}+\frac{1}{2\sqrt{7}}\subcgt{\bar{3}}{3}{3}\\
 & -\frac{4+\sqrt{2}}{4\sqrt{7}}\subcgt{\bar{3}}{\bar{3}}{3_{s}}-\frac{3}{2\sqrt{14}}\subcgt{\bar{3}}{\bar{3}}{3_{a}}
\end{align*}
\begin{align*}
\rep{\bar{3}} & =\frac{1}{\sqrt{7}}\subcgt{1}{\bar{3}}{\bar{3}}-\frac{4+\sqrt{2}}{4\sqrt{7}}\subcgt{3}{3}{\bar{3}_{s}}\\
 & -\frac{3}{2\sqrt{14}}\subcgt{3}{3}{\bar{3}_{a}}+\frac{1}{2\sqrt{7}}\subcgt{3}{\bar{3}}{\bar{3}}\\
 & +\frac{i\omega^{2}e^{i\alpha-2i\beta}}{2\sqrt{2}}\subcgt{\bar{3}}{1^{\prime}}{\bar{3}}-\frac{i\omega e^{2i\beta-i\alpha}}{2\sqrt{2}}\subcgt{\bar{3}}{\bar{1}^{\prime}}{\bar{3}}\\
 & -\frac{3\sqrt{2}-2}{4\sqrt{7}}\subcgt{\bar{3}}{3}{\bar{3}}-\frac{2+\sqrt{2}}{4\sqrt{7}}\subcgt{\bar{3}}{\bar{3}}{\bar{3}}
\end{align*}
\efl 
\item $\tsprodx{7}{8}\to\rep{6}^{(2)}$: \bfl 
\begin{align*}
\rep{3} & =\frac{i}{\sqrt{7}}\subcgt{1}{3}{3}-\frac{\omega e^{2i\beta-i\alpha}}{2\sqrt{2}}\subcgt{3}{1^{\prime}}{3}\\
 & +\frac{\omega^{2}e^{i\alpha-2i\beta}}{2\sqrt{2}}\subcgt{3}{\bar{1}^{\prime}}{3}-\frac{i\left(\sqrt{2}-2\right)}{4\sqrt{7}}\subcgt{3}{3}{3}\\
 & -\frac{i\left(2+3\sqrt{2}\right)}{4\sqrt{7}}\subcgt{3}{\bar{3}}{3}-\frac{i}{2\sqrt{7}}\subcgt{\bar{3}}{3}{3}\\
 & +\frac{i\left(\sqrt{2}-4\right)}{4\sqrt{7}}\subcgt{\bar{3}}{\bar{3}}{3_{s}}+\frac{3i}{2\sqrt{14}}\subcgt{\bar{3}}{\bar{3}}{3_{a}}
\end{align*}
\begin{align*}
\rep{\bar{3}} & =-\frac{i}{\sqrt{7}}\subcgt{1}{\bar{3}}{\bar{3}}-\frac{i\left(\sqrt{2}-4\right)}{4\sqrt{7}}\subcgt{3}{3}{\bar{3}_{s}}\\
 & -\frac{3i}{2\sqrt{14}}\subcgt{3}{3}{\bar{3}_{a}}+\frac{i}{2\sqrt{7}}\subcgt{3}{\bar{3}}{\bar{3}}\\
 & +\frac{\omega e^{2i\beta-i\alpha}}{2\sqrt{2}}\subcgt{\bar{3}}{1^{\prime}}{\bar{3}}-\frac{\omega^{2}e^{i\alpha-2i\beta}}{2\sqrt{2}}\subcgt{\bar{3}}{\bar{1}^{\prime}}{\bar{3}}\\
 & +\frac{i\left(2+3\sqrt{2}\right)}{4\sqrt{7}}\subcgt{\bar{3}}{3}{\bar{3}}+\frac{i\left(\sqrt{2}-2\right)}{4\sqrt{7}}\subcgt{\bar{3}}{\bar{3}}{\bar{3}}
\end{align*}
\efl 
\item $\tsprodx{7}{8}\to\rep{7}^{(1)}$: \bfl 
\begin{align*}
\rep{1} & =\frac{1}{\sqrt{2}}\subcgt{3}{\bar{3}}{1}+\frac{1}{\sqrt{2}}\subcgt{\bar{3}}{3}{1}
\end{align*}
\begin{align*}
\rep{3} & =\frac{1}{\sqrt{6}}\subcgt{1}{3}{3}+\frac{1}{4}i\sqrt{\frac{7}{3}}e^{i\alpha}\subcgt{3}{1^{\prime}}{3}\\
 & -\frac{1}{4}i\sqrt{\frac{7}{3}}e^{-i\alpha}\subcgt{3}{\bar{1}^{\prime}}{3}-\frac{1}{4\sqrt{3}}\subcgt{3}{3}{3}\\
 & -\frac{1}{4\sqrt{3}}\subcgt{3}{\bar{3}}{3}-\frac{1}{\sqrt{3}}\subcgt{\bar{3}}{3}{3}\\
 & +\frac{1}{\sqrt{6}}\subcgt{\bar{3}}{\bar{3}}{3_{s}}
\end{align*}
\begin{align*}
\rep{\bar{3}} & =\frac{1}{\sqrt{6}}\subcgt{1}{\bar{3}}{\bar{3}}+\frac{1}{\sqrt{6}}\subcgt{3}{3}{\bar{3}_{s}}\\
 & -\frac{1}{\sqrt{3}}\subcgt{3}{\bar{3}}{\bar{3}}+\frac{1}{4}i\sqrt{\frac{7}{3}}e^{i\alpha}\subcgt{\bar{3}}{1^{\prime}}{\bar{3}}\\
 & -\frac{1}{4}i\sqrt{\frac{7}{3}}e^{-i\alpha}\subcgt{\bar{3}}{\bar{1}^{\prime}}{\bar{3}}-\frac{1}{4\sqrt{3}}\subcgt{\bar{3}}{3}{\bar{3}}\\
 & -\frac{1}{4\sqrt{3}}\subcgt{\bar{3}}{\bar{3}}{\bar{3}}
\end{align*}
\efl 
\item $\tsprodx{7}{8}\to\rep{7}^{(2)}$: \bfl 
\begin{align*}
\rep{1} & =\frac{i}{\sqrt{2}}\subcgt{3}{\bar{3}}{1}-\frac{i}{\sqrt{2}}\subcgt{\bar{3}}{3}{1}
\end{align*}
\begin{align*}
\rep{3} & =\frac{i}{\sqrt{6}}\subcgt{1}{3}{3}+\frac{1}{4}\sqrt{\frac{7}{3}}e^{-i\alpha}\subcgt{3}{1^{\prime}}{3}\\
 & -\frac{1}{4}\sqrt{\frac{7}{3}}e^{i\alpha}\subcgt{3}{\bar{1}^{\prime}}{3}+\frac{i\sqrt{3}}{4}\subcgt{3}{3}{3}\\
 & +\frac{i\sqrt{3}}{4}\subcgt{3}{\bar{3}}{3}+\frac{i}{\sqrt{6}}\subcgt{\bar{3}}{\bar{3}}{3_{a}}
\end{align*}
\begin{align*}
\rep{\bar{3}} & =-\frac{i}{\sqrt{6}}\subcgt{1}{\bar{3}}{\bar{3}}-\frac{i}{\sqrt{6}}\subcgt{3}{3}{\bar{3}_{a}}\\
 & -\frac{1}{4}\sqrt{\frac{7}{3}}e^{-i\alpha}\subcgt{\bar{3}}{1^{\prime}}{\bar{3}}+\frac{1}{4}\sqrt{\frac{7}{3}}e^{i\alpha}\subcgt{\bar{3}}{\bar{1}^{\prime}}{\bar{3}}\\
 & -\frac{i\sqrt{3}}{4}\subcgt{\bar{3}}{3}{\bar{3}}-\frac{i\sqrt{3}}{4}\subcgt{\bar{3}}{\bar{3}}{\bar{3}}
\end{align*}
\efl 
\item $\tsprodx{7}{8}\to\rep{8}^{(1)}$: \bfl 
\begin{align*}
\rep{1^{\prime}} & =\frac{3\sqrt{3}-4i}{\sqrt{91}}\subcgt{1}{1^{\prime}}{1^{\prime}}-2\sqrt{\frac{6}{91}}\subcgt{3}{\bar{3}}{1^{\prime}}\\
 & -2\sqrt{\frac{6}{91}}\subcgt{\bar{3}}{3}{1^{\prime}}
\end{align*}
\begin{align*}
\rep{\bar{1}^{\prime}} & =\frac{3\sqrt{3}+4i}{\sqrt{91}}\subcgt{1}{\bar{1}^{\prime}}{\bar{1}^{\prime}}-2\sqrt{\frac{6}{91}}\subcgt{3}{\bar{3}}{\bar{1}^{\prime}}\\
 & -2\sqrt{\frac{6}{91}}\subcgt{\bar{3}}{3}{\bar{1}^{\prime}}
\end{align*}
\begin{align*}
\rep{3} & =-\sqrt{\frac{3}{91}}\subcgt{1}{3}{3}+2\sqrt{\frac{2}{91}}\omega^{2}\subcgt{3}{1^{\prime}}{3}\\
 & +2\sqrt{\frac{2}{91}}\omega\subcgt{3}{\bar{1}^{\prime}}{3}+2\sqrt{\frac{6}{91}}\subcgt{3}{3}{3}\\
 & -2\sqrt{\frac{6}{91}}\subcgt{3}{\bar{3}}{3}+2\sqrt{\frac{3}{91}}\subcgt{\bar{3}}{\bar{3}}{3_{s}}\\
 & -2\sqrt{\frac{3}{91}}\subcgt{\bar{3}}{\bar{3}}{3_{a}}
\end{align*}
\begin{align*}
\rep{\bar{3}} & =-\sqrt{\frac{3}{91}}\subcgt{1}{\bar{3}}{\bar{3}}+2\sqrt{\frac{3}{91}}\subcgt{3}{3}{\bar{3}_{s}}\\
 & -2\sqrt{\frac{3}{91}}\subcgt{3}{3}{\bar{3}_{a}}+2\sqrt{\frac{2}{91}}\omega^{2}\subcgt{\bar{3}}{1^{\prime}}{\bar{3}}\\
 & +2\sqrt{\frac{2}{91}}\omega\subcgt{\bar{3}}{\bar{1}^{\prime}}{\bar{3}}-2\sqrt{\frac{6}{91}}\subcgt{\bar{3}}{3}{\bar{3}}\\
 & +2\sqrt{\frac{6}{91}}\subcgt{\bar{3}}{\bar{3}}{\bar{3}}
\end{align*}
\efl 
\item $\tsprodx{7}{8}\to\rep{8}^{(2)}$: \bfl 
\begin{align*}
\rep{1^{\prime}} & =-4\sqrt{\frac{3}{91}}\omega^{2}\subcgt{1}{1^{\prime}}{1^{\prime}}-\frac{\sqrt{3}-13i}{2\sqrt{182}}\subcgt{3}{\bar{3}}{1^{\prime}}\\
 & -\frac{\sqrt{3}-13i}{2\sqrt{182}}\subcgt{\bar{3}}{3}{1^{\prime}}
\end{align*}
\begin{align*}
\rep{\bar{1}^{\prime}} & =-4\sqrt{\frac{3}{91}}\omega\subcgt{1}{\bar{1}^{\prime}}{\bar{1}^{\prime}}-\frac{\sqrt{3}+13i}{2\sqrt{182}}\subcgt{3}{\bar{3}}{\bar{1}^{\prime}}\\
 & -\frac{\sqrt{3}+13i}{2\sqrt{182}}\subcgt{\bar{3}}{3}{\bar{1}^{\prime}}
\end{align*}
\begin{align*}
\rep{3} & =-\frac{2}{\sqrt{273}}\subcgt{1}{3}{3}+\frac{i\left(5\sqrt{3}+21i\right)}{6\sqrt{182}}\subcgt{3}{1^{\prime}}{3}\\
 & -\frac{i\left(5\sqrt{3}-21i\right)}{6\sqrt{182}}\subcgt{3}{\bar{1}^{\prime}}{3}-\frac{5}{\sqrt{546}}\subcgt{3}{3}{3}\\
 & -4\sqrt{\frac{2}{273}}\subcgt{3}{\bar{3}}{3}+\sqrt{\frac{13}{42}}\subcgt{\bar{3}}{3}{3}\\
 & +\frac{4}{\sqrt{273}}\subcgt{\bar{3}}{\bar{3}}{3_{s}}+3\sqrt{\frac{3}{91}}\subcgt{\bar{3}}{\bar{3}}{3_{a}}
\end{align*}
\begin{align*}
\rep{\bar{3}} & =-\frac{2}{\sqrt{273}}\subcgt{1}{\bar{3}}{\bar{3}}+\frac{4}{\sqrt{273}}\subcgt{3}{3}{\bar{3}_{s}}\\
 & +3\sqrt{\frac{3}{91}}\subcgt{3}{3}{\bar{3}_{a}}+\sqrt{\frac{13}{42}}\subcgt{3}{\bar{3}}{\bar{3}}\\
 & +\frac{i\left(5\sqrt{3}+21i\right)}{6\sqrt{182}}\subcgt{\bar{3}}{1^{\prime}}{\bar{3}}-\frac{i\left(5\sqrt{3}-21i\right)}{6\sqrt{182}}\subcgt{\bar{3}}{\bar{1}^{\prime}}{\bar{3}}\\
 & -4\sqrt{\frac{2}{273}}\subcgt{\bar{3}}{3}{\bar{3}}-\frac{5}{\sqrt{546}}\subcgt{\bar{3}}{\bar{3}}{\bar{3}}
\end{align*}
\efl 
\item $\tsprodx{7}{8}\to\rep{8}^{(3)}$: \bfl 
\begin{align*}
\rep{1^{\prime}} & =-\frac{e^{-i\alpha}\omega^{2}}{\sqrt{2}}\subcgt{3}{\bar{3}}{1^{\prime}}+\frac{e^{-i\alpha}\omega^{2}}{\sqrt{2}}\subcgt{\bar{3}}{3}{1^{\prime}}
\end{align*}
\begin{align*}
\rep{\bar{1}^{\prime}} & =\frac{e^{i\alpha}\omega}{\sqrt{2}}\subcgt{3}{\bar{3}}{\bar{1}^{\prime}}-\frac{e^{i\alpha}\omega}{\sqrt{2}}\subcgt{\bar{3}}{3}{\bar{1}^{\prime}}
\end{align*}
\begin{align*}
\rep{3} & =-\frac{2i}{\sqrt{21}}\subcgt{1}{3}{3}-\frac{e^{i\alpha}\omega}{\sqrt{6}}\subcgt{3}{1^{\prime}}{3}\\
 & +\frac{e^{-i\alpha}\omega^{2}}{\sqrt{6}}\subcgt{3}{\bar{1}^{\prime}}{3}-i\sqrt{\frac{3}{14}}\subcgt{3}{3}{3}\\
 & -i\sqrt{\frac{3}{14}}\subcgt{\bar{3}}{3}{3}+\frac{i}{\sqrt{21}}\subcgt{\bar{3}}{\bar{3}}{3_{a}}
\end{align*}
\begin{align*}
\rep{\bar{3}} & =\frac{2i}{\sqrt{21}}\subcgt{1}{\bar{3}}{\bar{3}}-\frac{i}{\sqrt{21}}\subcgt{3}{3}{\bar{3}_{a}}\\
 & +i\sqrt{\frac{3}{14}}\subcgt{3}{\bar{3}}{\bar{3}}+\frac{e^{i\alpha}\omega}{\sqrt{6}}\subcgt{\bar{3}}{1^{\prime}}{\bar{3}}\\
 & -\frac{e^{-i\alpha}\omega^{2}}{\sqrt{6}}\subcgt{\bar{3}}{\bar{1}^{\prime}}{\bar{3}}+i\sqrt{\frac{3}{14}}\subcgt{\bar{3}}{\bar{3}}{\bar{3}}
\end{align*}
\efl 
\end{itemize}

\paragraph*{\cgEqFontsize $\tsprod{8}{8}\to\left(\rep{1}+\rep{6}^{(1)}+\rep{6}^{(2)}+\rep{7}+\rep{8}^{(1)}+\rep{8}^{(2)}\right)_{s}+\left(\rep{3}+\rep{\bar{3}}+\rep{7}^{(1)}+\rep{7}^{(2)}+\rep{8}\right)_{a}$}
\begin{itemize}
\item $\tsprodx{8}{8}\to\rep{1}_{s}$: \bfl 
\begin{align*}
\rep{1} & =\frac{1}{2\sqrt{2}}\left(\subcgt{1^{\prime}}{\bar{1}^{\prime}}{1}+\subcgt{\bar{1}^{\prime}}{1^{\prime}}{1}\right)\\
 & +\frac{\sqrt{\frac{3}{2}}}{2}\left(\subcgt{3}{\bar{3}}{1}+\subcgt{\bar{3}}{3}{1}\right)
\end{align*}
\efl 
\item $\tsprodx{8}{8}\to\rep{6}_{s}^{(1)}$: \bfl 
\begin{align*}
\rep{3} & =\frac{ie^{i\beta}}{2\sqrt{2}}\left(\subcgt{1^{\prime}}{3}{3}+\subcgt{3}{1^{\prime}}{3}\right)\\
 & -\frac{ie^{-i\beta}}{2\sqrt{2}}\left(\subcgt{\bar{1}^{\prime}}{3}{3}+\subcgt{3}{\bar{1}^{\prime}}{3}\right)\\
 & +\sqrt{\frac{1}{14}\left(3-\sqrt{2}\right)}\subcgt{3}{3}{3}\\
 & -\frac{1}{2}\sqrt{\frac{1}{14}\left(5+3\sqrt{2}\right)}\left(\subcgt{3}{\bar{3}}{3}+\subcgt{\bar{3}}{3}{3}\right)\\
 & +\frac{1}{2}\sqrt{\frac{1}{7}\left(3-\sqrt{2}\right)}\subcgt{\bar{3}}{\bar{3}}{3_{s}}
\end{align*}
\begin{align*}
\rep{\bar{3}} & =\frac{ie^{i\beta}}{2\sqrt{2}}\left(\subcgt{1^{\prime}}{\bar{3}}{\bar{3}}+\subcgt{\bar{3}}{1^{\prime}}{\bar{3}}\right)\\
 & -\frac{ie^{-i\beta}}{2\sqrt{2}}\left(\subcgt{\bar{1}^{\prime}}{\bar{3}}{\bar{3}}+\subcgt{\bar{3}}{\bar{1}^{\prime}}{\bar{3}}\right)\\
 & +\frac{1}{2}\sqrt{\frac{1}{7}\left(3-\sqrt{2}\right)}\subcgt{3}{3}{\bar{3}_{s}}\\
 & -\frac{1}{2}\sqrt{\frac{1}{14}\left(5+3\sqrt{2}\right)}\left(\subcgt{3}{\bar{3}}{\bar{3}}+\subcgt{\bar{3}}{3}{\bar{3}}\right)\\
 & +\sqrt{\frac{1}{14}\left(3-\sqrt{2}\right)}\subcgt{\bar{3}}{\bar{3}}{\bar{3}}
\end{align*}
\efl 
\item $\tsprodx{8}{8}\to\rep{6}_{s}^{(2)}$: \bfl 
\begin{align*}
\rep{3} & =-\frac{\omega^{2}e^{i\alpha-i\beta}}{2\sqrt{2}}\left(\subcgt{1^{\prime}}{3}{3}+\subcgt{3}{1^{\prime}}{3}\right)\\
 & +\frac{\omega e^{i\beta-i\alpha}}{2\sqrt{2}}\left(\subcgt{\bar{1}^{\prime}}{3}{3}+\subcgt{3}{\bar{1}^{\prime}}{3}\right)\\
 & +i\sqrt{\frac{1}{14}\left(3+\sqrt{2}\right)}\subcgt{3}{3}{3}\\
 & +\frac{1}{2}i\sqrt{\frac{1}{14}\left(5-3\sqrt{2}\right)}\left(\subcgt{3}{\bar{3}}{3}+\subcgt{\bar{3}}{3}{3}\right)\\
 & +\frac{1}{2}i\sqrt{\frac{1}{7}\left(3+\sqrt{2}\right)}\subcgt{\bar{3}}{\bar{3}}{3_{s}}
\end{align*}
\begin{align*}
\rep{\bar{3}} & =\frac{\omega^{2}e^{i\alpha-i\beta}}{2\sqrt{2}}\left(\subcgt{1^{\prime}}{\bar{3}}{\bar{3}}+\subcgt{\bar{3}}{1^{\prime}}{\bar{3}}\right)\\
 & -\frac{\omega e^{i\beta-i\alpha}}{2\sqrt{2}}\left(\subcgt{\bar{1}^{\prime}}{\bar{3}}{\bar{3}}+\subcgt{\bar{3}}{\bar{1}^{\prime}}{\bar{3}}\right)\\
 & -\frac{1}{2}i\sqrt{\frac{1}{7}\left(3+\sqrt{2}\right)}\subcgt{3}{3}{\bar{3}_{s}}\\
 & -\frac{1}{2}i\sqrt{\frac{1}{14}\left(5-3\sqrt{2}\right)}\left(\subcgt{3}{\bar{3}}{\bar{3}}+\subcgt{\bar{3}}{3}{\bar{3}}\right)\\
 & -i\sqrt{\frac{1}{14}\left(3+\sqrt{2}\right)}\subcgt{\bar{3}}{\bar{3}}{\bar{3}}
\end{align*}
\efl 
\item $\tsprodx{8}{8}\to\rep{7}_{s}$: \bfl 
\begin{align*}
\rep{1} & =\frac{\sqrt{\frac{3}{2}}}{2}\left(\subcgt{1^{\prime}}{\bar{1}^{\prime}}{1}+\subcgt{\bar{1}^{\prime}}{1^{\prime}}{1}\right)\\
 & -\frac{1}{2\sqrt{2}}\left(\subcgt{3}{\bar{3}}{1}+\subcgt{\bar{3}}{3}{1}\right)
\end{align*}
\begin{align*}
\rep{3} & =\frac{i\omega}{2\sqrt{3}}\left(\subcgt{1^{\prime}}{3}{3}+\subcgt{3}{1^{\prime}}{3}\right)\\
 & -\frac{i\omega^{2}}{2\sqrt{3}}\left(\subcgt{\bar{1}^{\prime}}{3}{3}+\subcgt{3}{\bar{1}^{\prime}}{3}\right)\\
 & -\frac{1}{\sqrt{3}}\subcgt{3}{3}{3}\\
 & +\frac{1}{2\sqrt{3}}\left(\subcgt{3}{\bar{3}}{3}+\subcgt{\bar{3}}{3}{3}\right)\\
 & +\frac{1}{\sqrt{6}}\subcgt{\bar{3}}{\bar{3}}{3_{s}}
\end{align*}
\begin{align*}
\rep{\bar{3}} & =\frac{i\omega}{2\sqrt{3}}\left(\subcgt{1^{\prime}}{\bar{3}}{\bar{3}}+\subcgt{\bar{3}}{1^{\prime}}{\bar{3}}\right)\\
 & -\frac{i\omega^{2}}{2\sqrt{3}}\left(\subcgt{\bar{1}^{\prime}}{\bar{3}}{\bar{3}}+\subcgt{\bar{3}}{\bar{1}^{\prime}}{\bar{3}}\right)\\
 & +\frac{1}{\sqrt{6}}\subcgt{3}{3}{\bar{3}_{s}}\\
 & +\frac{1}{2\sqrt{3}}\left(\subcgt{3}{\bar{3}}{\bar{3}}+\subcgt{\bar{3}}{3}{\bar{3}}\right)\\
 & -\frac{1}{\sqrt{3}}\subcgt{\bar{3}}{\bar{3}}{\bar{3}}
\end{align*}
\efl 
\item $\tsprodx{8}{8}\to\rep{8}_{s}^{(1)}$: \bfl 
\begin{align*}
\rep{1^{\prime}} & =-\frac{\omega^{2}}{\sqrt{3}}\subcgt{\bar{1}^{\prime}}{\bar{1}^{\prime}}{1^{\prime}}\\
 & -\frac{i}{\sqrt{3}}\left(\subcgt{3}{\bar{3}}{1^{\prime}}+\subcgt{\bar{3}}{3}{1^{\prime}}\right)
\end{align*}
\begin{align*}
\rep{\bar{1}^{\prime}} & =-\frac{\omega}{\sqrt{3}}\subcgt{1^{\prime}}{1^{\prime}}{\bar{1}^{\prime}}\\
 & +\frac{i}{\sqrt{3}}\left(\subcgt{3}{\bar{3}}{\bar{1}^{\prime}}+\subcgt{\bar{3}}{3}{\bar{1}^{\prime}}\right)
\end{align*}
\begin{align*}
\rep{3} & =\frac{i}{3}\left(\subcgt{1^{\prime}}{3}{3}+\subcgt{3}{1^{\prime}}{3}\right)\\
 & -\frac{i}{3}\left(\subcgt{\bar{1}^{\prime}}{3}{3}+\subcgt{3}{\bar{1}^{\prime}}{3}\right)\\
 & +\frac{1}{3}\subcgt{3}{3}{3}\\
 & +\frac{1}{3}\left(\subcgt{3}{\bar{3}}{3}+\subcgt{\bar{3}}{3}{3}\right)\\
 & +\frac{\sqrt{2}}{3}\subcgt{\bar{3}}{\bar{3}}{3_{s}}
\end{align*}
\begin{align*}
\rep{\bar{3}} & =\frac{i}{3}\left(\subcgt{1^{\prime}}{\bar{3}}{\bar{3}}+\subcgt{\bar{3}}{1^{\prime}}{\bar{3}}\right)\\
 & -\frac{i}{3}\left(\subcgt{\bar{1}^{\prime}}{\bar{3}}{\bar{3}}+\subcgt{\bar{3}}{\bar{1}^{\prime}}{\bar{3}}\right)\\
 & +\frac{\sqrt{2}}{3}\subcgt{3}{3}{\bar{3}_{s}}\\
 & +\frac{1}{3}\left(\subcgt{3}{\bar{3}}{\bar{3}}+\subcgt{\bar{3}}{3}{\bar{3}}\right)\\
 & +\frac{1}{3}\subcgt{\bar{3}}{\bar{3}}{\bar{3}}
\end{align*}
\efl 
\item $\tsprodx{8}{8}\to\rep{8}_{s}^{(2)}$: \bfl 
\begin{align*}
\rep{1^{\prime}} & =\sqrt{\frac{2}{3}}e^{-i\alpha}\subcgt{\bar{1}^{\prime}}{\bar{1}^{\prime}}{1^{\prime}}\\
 & -\frac{ie^{-i\alpha}\omega}{\sqrt{6}}\left(\subcgt{3}{\bar{3}}{1^{\prime}}+\subcgt{\bar{3}}{3}{1^{\prime}}\right)
\end{align*}
\begin{align*}
\rep{\bar{1}^{\prime}} & =\sqrt{\frac{2}{3}}e^{i\alpha}\subcgt{1^{\prime}}{1^{\prime}}{\bar{1}^{\prime}}\\
 & +\frac{ie^{i\alpha}\omega^{2}}{\sqrt{6}}\left(\subcgt{3}{\bar{3}}{\bar{1}^{\prime}}+\subcgt{\bar{3}}{3}{\bar{1}^{\prime}}\right)
\end{align*}
\begin{align*}
\rep{3} & =\frac{ie^{i\alpha}\omega^{2}}{3\sqrt{2}}\left(\subcgt{1^{\prime}}{3}{3}+\subcgt{3}{1^{\prime}}{3}\right)\\
 & -\frac{ie^{-i\alpha}\omega}{3\sqrt{2}}\left(\subcgt{\bar{1}^{\prime}}{3}{3}+\subcgt{3}{\bar{1}^{\prime}}{3}\right)\\
 & -\frac{2\sqrt{\frac{2}{7}}}{3}\subcgt{3}{3}{3}\\
 & -\frac{2\sqrt{\frac{2}{7}}}{3}\left(\subcgt{3}{\bar{3}}{3}+\subcgt{\bar{3}}{3}{3}\right)\\
 & +\frac{5}{3\sqrt{7}}\subcgt{\bar{3}}{\bar{3}}{3_{s}}
\end{align*}
\begin{align*}
\rep{\bar{3}} & =\frac{ie^{i\alpha}\omega^{2}}{3\sqrt{2}}\left(\subcgt{1^{\prime}}{\bar{3}}{\bar{3}}+\subcgt{\bar{3}}{1^{\prime}}{\bar{3}}\right)\\
 & -\frac{ie^{-i\alpha}\omega}{3\sqrt{2}}\left(\subcgt{\bar{1}^{\prime}}{\bar{3}}{\bar{3}}+\subcgt{\bar{3}}{\bar{1}^{\prime}}{\bar{3}}\right)\\
 & +\frac{5}{3\sqrt{7}}\subcgt{3}{3}{\bar{3}_{s}}\\
 & -\frac{2\sqrt{\frac{2}{7}}}{3}\left(\subcgt{3}{\bar{3}}{\bar{3}}+\subcgt{\bar{3}}{3}{\bar{3}}\right)\\
 & -\frac{2\sqrt{\frac{2}{7}}}{3}\subcgt{\bar{3}}{\bar{3}}{\bar{3}}
\end{align*}
\efl 
\item $\tsprodx{8}{8}\to\rep{3}_{a}$: \bfl 
\begin{align*}
\rep{3} & =\frac{1}{2\sqrt{2}}\left(\subcgt{1^{\prime}}{3}{3}-\subcgt{3}{1^{\prime}}{3}\right)\\
 & -\frac{\omega}{2\sqrt{2}}\left(\subcgt{\bar{1}^{\prime}}{3}{3}-\subcgt{3}{\bar{1}^{\prime}}{3}\right)\\
 & +\frac{i\omega^{2}}{2\sqrt{2}}\left(\subcgt{3}{\bar{3}}{3}-\subcgt{\bar{3}}{3}{3}\right)\\
 & -\frac{i\omega^{2}}{2}\subcgt{\bar{3}}{\bar{3}}{3_{a}}
\end{align*}
\efl 
\item $\tsprodx{8}{8}\to\rep{\bar{3}}_{a}$: \bfl 
\begin{align*}
\rep{\bar{3}} & =\frac{1}{2\sqrt{2}}\left(\subcgt{1^{\prime}}{\bar{3}}{\bar{3}}-\subcgt{\bar{3}}{1^{\prime}}{\bar{3}}\right)\\
 & -\frac{\omega}{2\sqrt{2}}\left(\subcgt{\bar{1}^{\prime}}{\bar{3}}{\bar{3}}-\subcgt{\bar{3}}{\bar{1}^{\prime}}{\bar{3}}\right)\\
 & -\frac{i\omega^{2}}{2}\subcgt{3}{3}{\bar{3}_{a}}\\
 & -\frac{i\omega^{2}}{2\sqrt{2}}\left(\subcgt{3}{\bar{3}}{\bar{3}}-\subcgt{\bar{3}}{3}{\bar{3}}\right)
\end{align*}
\efl 
\item $\tsprodx{8}{8}\to\rep{7}_{a}^{(1)}$: \bfl 
\begin{align*}
\rep{1} & =\frac{i}{\sqrt{2}}\left(\subcgt{1^{\prime}}{\bar{1}^{\prime}}{1}-\subcgt{\bar{1}^{\prime}}{1^{\prime}}{1}\right)
\end{align*}
\begin{align*}
\rep{3} & =\frac{\omega}{4}\left(\subcgt{1^{\prime}}{3}{3}-\subcgt{3}{1^{\prime}}{3}\right)\\
 & +\frac{\omega^{2}}{4}\left(\subcgt{\bar{1}^{\prime}}{3}{3}-\subcgt{3}{\bar{1}^{\prime}}{3}\right)\\
 & +\frac{\sqrt{3}}{4}\left(\subcgt{3}{\bar{3}}{3}-\subcgt{\bar{3}}{3}{3}\right)\\
 & +\frac{\sqrt{\frac{3}{2}}}{2}\subcgt{\bar{3}}{\bar{3}}{3_{a}}
\end{align*}
\begin{align*}
\rep{\bar{3}} & =\frac{\omega}{4}\left(\subcgt{1^{\prime}}{\bar{3}}{\bar{3}}-\subcgt{\bar{3}}{1^{\prime}}{\bar{3}}\right)\\
 & +\frac{\omega^{2}}{4}\left(\subcgt{\bar{1}^{\prime}}{\bar{3}}{\bar{3}}-\subcgt{\bar{3}}{\bar{1}^{\prime}}{\bar{3}}\right)\\
 & +\frac{\sqrt{\frac{3}{2}}}{2}\subcgt{3}{3}{\bar{3}_{a}}\\
 & -\frac{\sqrt{3}}{4}\left(\subcgt{3}{\bar{3}}{\bar{3}}-\subcgt{\bar{3}}{3}{\bar{3}}\right)
\end{align*}
\efl 
\item $\tsprodx{8}{8}\to\rep{7}_{a}^{(2)}$: \bfl 
\begin{align*}
\rep{1} & =\frac{i}{\sqrt{2}}\left(\subcgt{3}{\bar{3}}{1}-\subcgt{\bar{3}}{3}{1}\right)
\end{align*}
\begin{align*}
\rep{3} & =-\frac{1}{4}\sqrt{\frac{7}{3}}e^{i\alpha}\omega\left(\subcgt{1^{\prime}}{3}{3}-\subcgt{3}{1^{\prime}}{3}\right)\\
 & +\frac{1}{4}\sqrt{\frac{7}{3}}e^{-i\alpha}\omega^{2}\left(\subcgt{\bar{1}^{\prime}}{3}{3}-\subcgt{3}{\bar{1}^{\prime}}{3}\right)\\
 & +\frac{i\sqrt{3}}{4}\left(\subcgt{3}{\bar{3}}{3}-\subcgt{\bar{3}}{3}{3}\right)\\
 & -\frac{i}{2\sqrt{6}}\subcgt{\bar{3}}{\bar{3}}{3_{a}}
\end{align*}
\begin{align*}
\rep{\bar{3}} & =\frac{1}{4}\sqrt{\frac{7}{3}}e^{i\alpha}\omega\left(\subcgt{1^{\prime}}{\bar{3}}{\bar{3}}-\subcgt{\bar{3}}{1^{\prime}}{\bar{3}}\right)\\
 & -\frac{1}{4}\sqrt{\frac{7}{3}}e^{-i\alpha}\omega^{2}\left(\subcgt{\bar{1}^{\prime}}{\bar{3}}{\bar{3}}-\subcgt{\bar{3}}{\bar{1}^{\prime}}{\bar{3}}\right)\\
 & +\frac{i}{2\sqrt{6}}\subcgt{3}{3}{\bar{3}_{a}}\\
 & +\frac{i\sqrt{3}}{4}\left(\subcgt{3}{\bar{3}}{\bar{3}}-\subcgt{\bar{3}}{3}{\bar{3}}\right)
\end{align*}
\efl 
\item $\tsprodx{8}{8}\to\rep{8}_{a}$: \bfl 
\begin{align*}
\rep{1^{\prime}} & =-\frac{\omega}{\sqrt{2}}\left(\subcgt{3}{\bar{3}}{1^{\prime}}-\subcgt{\bar{3}}{3}{1^{\prime}}\right)
\end{align*}
\begin{align*}
\rep{\bar{1}^{\prime}} & =\frac{\omega^{2}}{\sqrt{2}}\left(\subcgt{3}{\bar{3}}{\bar{1}^{\prime}}-\subcgt{\bar{3}}{3}{\bar{1}^{\prime}}\right)
\end{align*}
\begin{align*}
\rep{3} & =\frac{\omega^{2}}{\sqrt{6}}\left(\subcgt{1^{\prime}}{3}{3}-\subcgt{3}{1^{\prime}}{3}\right)\\
 & -\frac{\omega}{\sqrt{6}}\left(\subcgt{\bar{1}^{\prime}}{3}{3}-\subcgt{3}{\bar{1}^{\prime}}{3}\right)\\
 & -\frac{i}{\sqrt{3}}\subcgt{\bar{3}}{\bar{3}}{3_{a}}
\end{align*}
\begin{align*}
\rep{\bar{3}} & =-\frac{\omega^{2}}{\sqrt{6}}\left(\subcgt{1^{\prime}}{\bar{3}}{\bar{3}}-\subcgt{\bar{3}}{1^{\prime}}{\bar{3}}\right)\\
 & +\frac{\omega}{\sqrt{6}}\left(\subcgt{\bar{1}^{\prime}}{\bar{3}}{\bar{3}}-\subcgt{\bar{3}}{\bar{1}^{\prime}}{\bar{3}}\right)\\
 & +\frac{i}{\sqrt{3}}\subcgt{3}{3}{\bar{3}_{a}}
\end{align*}
\efl 
\end{itemize}
\bibliographystyle{plain}
\bibliography{rp}

\end{document}